\documentclass[journal]{IEEEtran}
\ifCLASSINFOpdf
\else
\fi

\usepackage{epsfig}\usepackage{epsfig}
\usepackage{amssymb}
\usepackage{amsthm}
\usepackage{amsmath}
\usepackage{cite}
\usepackage{mathrsfs}
\usepackage{cases}
\usepackage{algorithm} 
\usepackage{algpseudocode} 
\usepackage{booktabs}
\usepackage{multirow}
\usepackage{threeparttable}
\usepackage{makecell}
\usepackage{url}
\usepackage{subeqnarray}
\usepackage{subfigure}
\graphicspath{ {figure/} }
\usepackage{pifont,bbm,amsfonts,mathrsfs,amsmath,stmaryrd,amssymb}
\usepackage{stfloats}
\usepackage{amscd,graphicx,array,dsfont,texdraw,tikz,footnote,verbatim}%
\usepackage[normalem]{ulem}
\usepackage{mathtools}
\usepackage{bbm,bbding,amssymb,pifont,mathrsfs,amsfonts,graphicx,mathtools,color}%
\usepackage{amscd,array,enumerate,dsfont,texdraw,tikz,multicol,bbm}

\usetikzlibrary{arrows,shapes,snakes,shadows,positioning,automata,patterns}
\usetikzlibrary{trees,decorations.pathmorphing,decorations.markings}
\usepackage{setspace} 
\usepackage{tikz}
\usepackage{pgf}
\usepackage{amsmath}

\usepackage{amsmath, amsfonts, amssymb, amsthm}
\usepackage{pgfplots}
\pgfplotsset{compat=newest}
\usepackage{xxcolor}
\usetikzlibrary{arrows, decorations.markings}
\usetikzlibrary{arrows,shadows,petri}
\newcommand{\beq}{\begin{equation}}
\newcommand{\eeq}{\end{equation}}
\newcommand{\bqa}{\begin{eqnarray}}
\newcommand{\eqa}{\end{eqnarray}}

\definecolor{maroon}{rgb}{0.7,0,0}

\definecolor{ngreen}{rgb}{0.3,0.7,0.3}

\definecolor{golden}{rgb}{0.8,0.6,0.1}

\newtheorem{theorem}{\indent Theorem}
\newtheorem{lemma}{\indent Lemma}

\newtheorem{definition}{\indent Definition}
\newtheorem{assumption}{\indent Assumption}
\newtheorem{myremark}{\indent Remark}

\newenvironment{remark}{\begin{myremark}\normalfont}
{\end{myremark}}

\begin{document}
\title{Ensuring Truthfulness in Distributed Aggregative Optimization}

\author{Ziqin Chen, Magnus Egerstedt, \textit{Fellow, IEEE}, and Yongqiang Wang, \textit{Senior Member, IEEE}
\thanks{This work was supported by the National Science Foundation under Grant ECCS-1912702, Grant CCF-2106293, Grant CCF-2215088, Grant CNS-2219487, Grant CCF-2334449, and Grant CNS-2422312. (Corresponding author: Yongqiang Wang, email:yongqiw@clemson.edu).}
\thanks{Ziqin Chen and Yongqiang Wang are with the Department of Electrical and Computer Engineering, Clemson University, Clemson, SC 29634 USA and Magnus Egerstedt is with the Department of Electrical Engineering and Computer Science, University of California, Irvine, Irvine, CA 92697 USA.}}

\maketitle
\begin{abstract}
Distributed aggregative optimization methods are gaining increased traction due to their ability to address cooperative control and optimization problems, where the objective function of each agent depends not only on its own decision variable but also on the aggregation of other agents' decision variables. Nevertheless, existing distributed aggregative optimization methods implicitly assume all agents to be truthful in information sharing, which can be unrealistic in real-world scenarios, where agents may act selfishly or strategically. In fact, an opportunistic agent may deceptively share false information in its own favor to minimize its own loss, which, however, will compromise the network-level global performance. To solve this issue, we propose a new distributed aggregative optimization algorithm that can ensure truthfulness of agents and convergence performance. To the best of our knowledge, this is the first algorithm that ensures truthfulness in a fully distributed setting, where no ``centralized" aggregator exists to collect private information/decision variables from participating agents. We systematically characterize the convergence rate of our algorithm under nonconvex/convex/strongly convex objective functions, which generalizes existing distributed aggregative optimization results that only focus on convex objective functions. We also rigorously quantify the tradeoff between convergence performance and the level of enabled truthfulness under different convexity conditions. Numerical simulations using distributed charging of electric vehicles confirm the efficacy of our algorithm.
\end{abstract}

\IEEEpeerreviewmaketitle
\section{Introduction}
Recently, there has been a surge of interest in distributed optimization which underpins numerous applications in cooperative control~\cite{cooperativecontrol,cooperativecontrol1}, signal processing~\cite{signalprocess}, and machine learning~\cite{machinelearning}. In distributed optimization, a group of agents cooperatively learns a common decision variable that minimizes a global objective function that is the sum of individual agents' objective functions. Usually, the local objective function of one agent is assumed to be only dependent on its own local decision variable $x^{i}$, i.e., it has the form of $f_{i}(x^{i})$, and hence, the global objective function has the form of $\sum_{i=1}^{m}f_{i}(x^{i})$. However, in many emerging applications such as placement optimization of warehouses~\cite{lixiuxian2021}, multi-vehicle charging~\cite{charging2011}, and cooperative robot surveillance~\cite{carnevale2022online}, the local objective function of each agent depends not only on its own decision variable but also on the aggregation of all agents' decisions. These types of problems can be mathematically formulated as the following aggregative optimization problem:
\begin{equation}
\vspace{-0.3em}
\text{min}_{x\in\mathcal{X}}F(x)\!=\!\sum_{i=1}^{m}f_{i}(x^{i},\phi(x)),~~\phi(x)\!=\!\frac{1}{m}\sum_{i=1}^{m}g_{i}(x^{i}),~\label{primal}
\vspace{-0.2em}
\end{equation}
where $x=\text{col}(x^{1},\cdots,x^{m})$, $f_{i}(x^{i},\phi(x))$ is the local objective function of agent $i$, which is determined by both its own local decision variable $x^{i}$ and an aggregative term $\phi(x)$. Here, $\mathcal{X}_{i}\subseteq\mathbb{R}^{n_{i}}$ represents agent $i$'s local constraint set and $\mathcal{X}=\textstyle \prod_{i=1}^{m}\mathcal{X}_{i}\subseteq \mathbb{R}^{n}$ with $n=\sum_{i=1}^{m}n_{i}$. We consider the case where $f_{i}$, $g_{i}$, and $\mathcal{X}_{i}$ are private to agent $i$, and hence, $\phi(x)$ is not directly accessible to any individual agents.

To solve problem~\eqref{primal}, several gradient-tracking-based algorithms have been proposed for strongly convex objective functions~\cite{lixiuxian2021,charging2011,carnevale2022online,aggregative2,aggregative3,aggregative4,aggregative5} and convex objective functions~\cite{aggregative6,lixiuxianonline,aggregative7,aggregative8}. Recently, some results have also been reported for nonconvex objective functions~\cite{aggregative9,carnevale2024nonconvex}. However, the approach in~\cite{aggregative9} is only applicable to the special case where $f_{i}$ is solely dependent on the aggregative term. In addition, the approach in~\cite{carnevale2024nonconvex} considers the case of continuous-time optimization, and its Lasalle-invariance-principle-based derivation can only prove asymptotic convergence which can be arbitrarily slow. To the best of our knowledge, no convergence results with explicit convergence rates have been reported for distributed nonconvex aggregative optimization.

Another potential limitation of existing distributed aggregative optimization algorithms in~\cite{lixiuxian2021,charging2011,carnevale2022online,aggregative2,aggregative3,aggregative4,aggregative5,aggregative6,lixiuxianonline,aggregative7,aggregative8,aggregative9,carnevale2024nonconvex} is that they implicitly assume all agents to be truthful in information sharing. However, opportunistic agents may intentionally share deceptive information with others for their own benefits. It is worth
noting that we consider untruthful behaviors of agents (which
aim to gain benefits) to be different from malicious behaviors
(which aim to destroy the system): agents are opportunistic,
and hence, they may share untruthful/deceitful information
with others to skew the final computation result, but they
are not malicious, i.e., they will not send information that
prevents the agents from running distributed algorithms. This
is a practical assumption. For example, in crowd-sensing-based traffic navigation, e.g., Google Maps and Waze~\cite{Waze}, users may provide deceitful traffic data to influence the app's routing decisions. Their intent is to divert vehicles to or away from certain areas rather than to disrupt the app's overall functionality. Similarly, in electric-vehicle (EV) charging, untruthful EV users may report false charging specifications to change the charging schedules of other EVs so as to gain personal benefits such as preferred charging time or reduced cost~\cite{trEVcharging,truthfulEV} (see Section~\ref{SectionIIB} and~\cite{truthfulEV} for details). But they will not prevent other EVs from participating in charging scheduling.

In the special case where monetary payments are allowed to modify agents' local objective functions, Vickrey-Clarke-Groves (VCG) mechanisms can be used to guarantee truthfulness~\cite{VCG95,VCG96,VCG97,researchagenda}. However, existing VCG-based results, e.g., \cite{VCG1,VCG2,VCG3,VCG4,VCG5,VCG6}, do not address the scenario where the objective functions of different agents are coupled, e.g., through an aggregative term considered in problem~\eqref{primal}, which is the case in practice when the price of electricity is affected by demand. Moreover, VCG-based approaches are often not budget-balanced and involve surplus payments (i.e., the sum of all agents' monetary payments does not equal zero), which inevitably imposes additional economic overhead on individual agents~\cite{VCG4}. Furthermore, implementing such monetary payments requires a ``centralized" aggregator to collect truthful gradient/function information from all agents~\cite{VCG1,VCG2,VCG3,VCG4,VCG5,VCG6}, which may not be acceptable or available in distributed multi-agent networks. 

Another approach to achieving truthfulness is joint differential privacy (JDP). The basic idea is to make the untruthful information shared by an agent have a negligible influence on the decisions of other agents, thereby limiting the benefit that an untruthful agent
can gain~\cite{DPtruthful2,DPtruthful3,DPtruthful1}. JDP has been used to promote truthfulness of agents in equilibrium computation in games~\cite{DPtruthful2,DPtruthful3,DPtruthful1,JDPgame1,JDPgame2,JDPgame3,JDPgame4,JDPgame5}, distributed EV charging with the assistance of a server~\cite{truthfulEV}, distributed cloud computation/optimization~\cite{JDPoptimization1}, and linearly separable convex programming~\cite{JDPoptimization2}. However, these results all rely on a ``centralized" aggregator to collect iteration/decision variables from all agents. To our knowledge, no results have been reported which can use JDP to enhance truthfulness in a fully distributed setting. Moreover, the results in~\cite{truthfulEV,JDPoptimization1,JDPoptimization2} can only ensure truthfulness for a finite number of computation iterations (i.e., the level of ensured truthfulness will diminish to zero as the number of iterations tends to infinity). 

In this paper, we aim to achieve JDP-based truthfulness in a fully distributed setting without requiring the assistance of any aggregators. To this end, we first enhance an existing distributed aggregative optimization algorithm in~\cite{lixiuxian2021} to ensure accurate convergence despite the presence of Laplace noises (Theorem~\ref{T1}). Then, by co-designing the stepsize and Laplace-noise injection mechanism, we achieve rigorous truthfulness guarantee in the sense that the maximal loss gain that an agent can obtain from untruthful behaviors is bounded by a constant $\eta$ (called $\eta$-truthfulness) (Theorem~\ref{T2}). Furthermore, we rigorously quantify the price and tradeoff in convergence performance for achieving $\eta$-truthfulness (Theorem~\ref{T3}). Note that our truthfulness approach can address two types of untruthful agent behaviors: 1) providing falsified input information at the beginning of the algorithm, and 2) sharing deceitful intermediate updates during algorithm execution. Nevertheless, here we mainly focus on input manipulation, as it is a more direct and practically effective way for an untruthful agent to increase its personal gains. By comparison, manipulating intermediate updates does not provide a clear strategic advantage, since only the final computational outcome (not the intermediate update values) is executed in final implementation and affects agents' personal gains. Moreover, the effect of changing intermediate updates on the final computational outcome is often hard to evaluate in run time, which makes it hard for agents to exploit such manipulation effectively to make personal gains.
The main contributions are summarized as follows:
\begin{itemize}
\item We provide rigorous convergence rate analysis for distributed aggregative optimization under nonconvex/convex/strongly convex objective functions. This is different from existing results in~\cite{lixiuxian2021,charging2011,carnevale2022online,aggregative2,aggregative3,aggregative4,aggregative5,aggregative6,lixiuxianonline,aggregative7,aggregative8}, which only consider convex objective functions. This is also different from the result in~\cite{carnevale2024nonconvex}, which 
only proves asymptotic convergence without giving explicit convergence rates (in addition, this result considers continuous-time optimization, which is simpler in convergence analysis). To the best of our knowledge, we are the first to obtain convergence with explicit convergence rates in distributed nonconvex aggregative optimization even in the presence of Laplace noises.

\item We prove that our algorithm can ensure $\eta$-truthfulness. Compared with existing truthfulness solutions for games and cooperative optimization~\cite{VCG1,VCG2,VCG3,VCG4,VCG5,VCG6,truthfulEV,JDPoptimization1,JDPoptimization2}, all of which rely on a ``centralized" aggregator, our approach is the first to ensure truthfulness in a fully distributed setting without the assistance of any ``centralized" aggregators.

\item Our guaranteed truthfulness level, i.e., $\eta$, is ensured to be finite even when the number of iterations tends to infinity. This is different from existing results in~\cite{truthfulEV,JDPoptimization1,JDPoptimization2}, whose truthfulness level (or strength) diminishes to zero as the number of iterations tends to infinity. One side contribution is achieving JDP in a fully distributed setting. It is worth noting that the conventional differential privacy cannot be used to ensure truthfulness because it makes all inputs indistinguishable, and hence, eliminates the incentive for agents to behave truthfully.

\item In addition to ensuring $\eta$-truthfulness in distributed aggregative optimization, we rigorously quantify the tradeoff between convergence performance and the level of enabled truthfulness under nonconvex/convex/strongly convex objective functions.

\item Using a distributed EV charging problem, we test our algorithm in a real-world application scenario with actual EV-charging specifications and real load data. The results confirm the effectiveness of our algorithm in practical applications.
\end{itemize}

The organization of the paper is as follows. Sec. \ref{SectionII} introduces the problem formulation and
the definitions of JDP and truthfulness. Sec. \ref{SectionIII}
presents our truthfully distributed aggregative optimization algorithm. Sec.~\ref{SectionIV} analyzes the optimization accuracy and convergence rate. Sec.~\ref{SectionV} establishes the truthfulness guarantee. Sec. \ref{SectionVI} provides
numerical simulation results. Sec. \ref{SectionVII} concludes the paper.

\textit{Notations:} We use $\mathbb{R}^{n}$ ($\mathbb{R}_{+}^{n}$) to denote the $n$-dimensional real (non-negative) Euclidean space and $\mathbb{N}$ ($\mathbb{N}^{+}$) to denote the set of non-negative (positive) integers. We
write $\boldsymbol{1}_{n}$ and $I_{n}$ for the $n$-dimensional column vector of all ones and the identity matrix, respectively. We use $\langle\cdot,\cdot\rangle$ to denote the inner product of
two vectors and $\|\cdot\|$ to denote the Euclidean norm of a vector.
We add an overbar to a letter to denote the average of $m$ agents, e.g., $\bar{x}\!=\!\frac{1}{m}\sum_{i=1}^{m}x^{i}$.
We use a letter without a superscript to represent the stacked vector of $m$ agents. We write $\mathbb{P}[\mathcal{A}]$ for the probability of an event $\mathcal{A}$ and $\mathbb{E}[x]$ for the expected value of a random variable $x$. We use $\mathcal{B}_{r}$ to denote a closed ball centered at the origin with diameter $r>0$. $\boldsymbol{P}_{\mathcal{X}}(x)$ denotes the Euclidean projection of a vector $x$ on the set $\mathcal{X}$, i.e., $\boldsymbol{P}_{\mathcal{X}}(x)=\text{argmin}_{y\in\mathcal{X}}\|y-x\|^2$. For a differentiable function $f(x,\psi)$, we let 
$\nabla_{1}f(x,\psi)$ and $\nabla_{2}f(x,\psi)$ denote the partial derivatives $\nabla_{x}f(x,\psi)$ and $\nabla_{\psi}f(x,\psi)$, respectively.
We use $\text{Lap}(\nu)$ to denote the Laplace distribution with
a parameter $\nu\!>\!0$, featuring a probability density function $\frac{1}{2\nu}e^{-\frac{|x|}{\nu}}$. $\text{Lap}(\nu)$ has a mean of zero and a variance of $2\nu^2$.

\section{Preliminaries and Problem Statement}\label{SectionII}
\subsection{Distributed aggregative optimization}\label{SectionIIA}
We consider $m$ agents that cooperatively learn a common optimal decision $x^*$ to the distributed aggregative optimization problem in~\eqref{primal}. The objective function $f_{i}(x^{i},\phi(x)):\mathcal{X}\mapsto\mathbb{R}$ and the function $g_{i}(x^{i}):\mathcal{X}_{i}\mapsto\mathbb{R}^{d}$ satisfy the following standard assumptions:
\begin{assumption}\label{A1}
The constraint set $\mathcal{X}_{i}$ is nonempty, compact, and convex. In addition, 

(i) $f_{i}(x,\psi)$ is $L_{f,1}$- and $L_{f,2}$-Lipschitz continuous with respect to $x$ and $\psi$, respectively. $\nabla f_{i}(x,\psi)$ is $\bar{L}_{f,1}$- and $\bar{L}_{f,2}$-Lipschitz continuous with respect to $x$ and $\psi$, respectively; 

(ii) $g_{i}(x)$ is $L_{g}$-Lipschitz continuous. $\nabla g_{i}(x)$ is $\bar{L}_{g}$-Lipschitz continuous. 
\end{assumption}
\begin{remark}
\vspace{-0.2em}
Assumption~\ref{A1} is commonly used in distributed aggregative optimization~\cite{aggregative6,lixiuxianonline,aggregative9} and nonconvex optimization with compositional structures~\cite{bilevel1,bilevel2,bilevel3}.
For example, it can be verified to hold in the problem of optimal placement of warehouses~\cite{lixiuxian2021} and the problem of coordinated EV charging considered in~\cite{truthfulEV} (see Section~\ref{SectionIIB} for details). It is worth noting that we allow both $f_{i}$ and $g_{i}$ to be nonconvex, which is more general than the strongly convex or convex condition in existing distributed aggregative optimization results~\cite{lixiuxian2021,charging2011,carnevale2022online,aggregative2,aggregative3,aggregative4,aggregative5,aggregative6,lixiuxianonline,aggregative7,aggregative8}.
\vspace{-0.2em}
\end{remark}

We describe the communication pattern among agents using an $m\times m$ matrix $W$. If agents $i$ and $j$ can communicate with each other, then $w_{ij}$ is positive, and $w_{ij}=0$ otherwise. The set of agents that can directly interact with agent $i$ is called the neighboring set of agent $i$ and is represented as $\mathcal{N}_{i}$. We let $w_{ii}=-\sum_{j\in{\mathcal{N}_{i}}}w_{ij}.$ The matrix $W$ satisfies the following assumption:
\begin{assumption}\label{A2}
\vspace{-0.2em}
The communication among all agents is modeled by an undirected and connected graph. The matrix $W$\footnote{Our matrix $I+\epsilon W$ corresponds to the Perron matrix $P_{\epsilon}=I-\epsilon L$ used in~\cite{consensusL}, where $L$ is the Laplacian matrix.} satisfies $\boldsymbol{1}_{m}^{\top}W=\boldsymbol{0}_{m}^{\top}$ and $W\boldsymbol{1}_{m}=\boldsymbol{0}_{m}$ and its eigenvalues satisfy (after arranged in an increasing order) $-1<\delta_{m}\leq\cdots\leq\delta_{2}< \delta_{1}=0$.
\vspace{-0.2em}
\end{assumption}

In problem~\eqref{algorithm1}, we assume that opportunistic agents may share deceitful messages strategically so as to reduce their own losses. This, however, will increase the network-level global loss.
\begin{figure}
\centering
\includegraphics[width=0.45\textwidth]{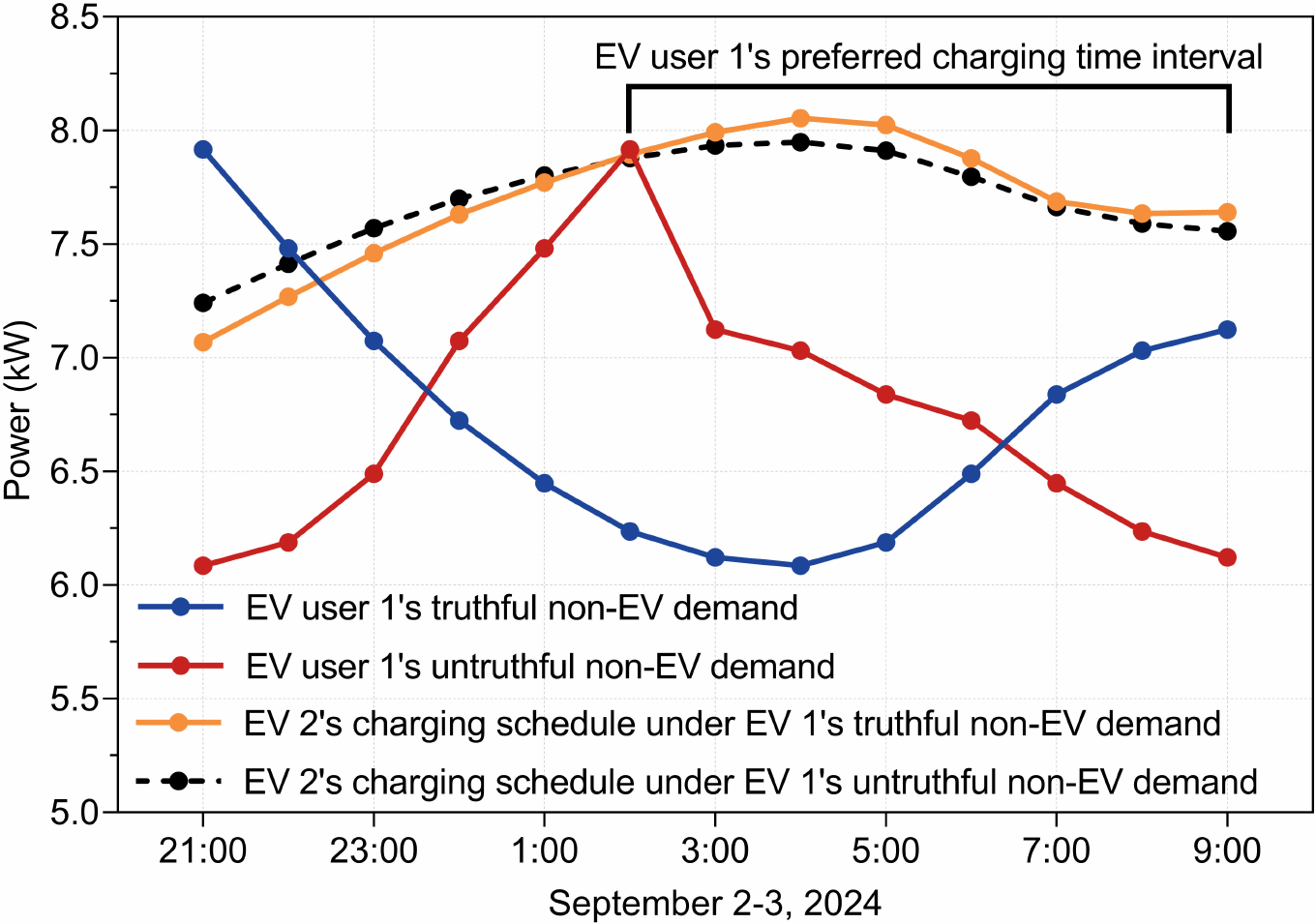}
\caption{Charging schedule curves for EV 2 under EV user 1's truthful and untruthful non-EV demand profiles, respectively. Each EV user ran the conventional noise-free distributed aggregative optimization algorithm in~\cite{lixiuxian2021} over $1,000$ iterations. The solid orange curve was computed under EV parameters with $x_{\max}^{1}=x_{\max}^{2}$, $E^{1}=E^{2}$, and $d^{1}=d^{2}$, while the dashed black curve was computed under EV user $1$'s untruthful parameter ${d'}^{1}\neq d^{1}$. The price function is given by $p(r)=\text{col}(0.15[r]_{1}^{1.5},\cdots,0.15[r]_{13}^{1.5})$, where $[r]_{k}$ denotes the $k$th element of $r$ with $k=1,\cdots,13$. It is clear that the algorithm schedules more charging power for EV 2 before midnight and less charging power after midnight when EV user 1 provides an untruthful non-EV demand profile ${d'}^{1}$.} % 图片标题
\label{challenge} % 设置标签用于引用
\vspace{-1.2em}
\end{figure}
\subsection{Truthfulness issue in distributed EV charging applications}\label{SectionIIB}
In this subsection, we exemplify the truthfulness issue in distributed aggregative optimization using a distributed EV charging problem~\cite{charging2011,truthfulEV} involving $m$ EVs.

We denote $x^{i}\in\mathbb{R}_{+}^{K}$ as EV $i$'s charging profile over a time interval $\mathcal{K}=\{1,\cdots,K\}$. Hence, $x\triangleq\text{col}(x^{1},\cdots,x^{m})$ denotes a charging schedule for all EVs. By the end of the time interval, each EV $i$ is required to charge a total amount of energy $E^{i}\in\mathbb{R}_{+}$. Each EV $i$ can specify its maximal charging rate as a vector $x_{\max}^{i}\in\mathbb{R}_{+}^{K}$ so that each element of $x^{i}$ does not exceed the corresponding element of $x_{\max}^{i}$. Both $E^{i}$ and $x_{\max}^{i}$ constitute the charging specification of EV $i$, which can be expressed as the following set constraint: 
\begin{equation}
\mathcal{X}_{i}=\{x^{i}\in\mathbb{R}_{+}^{K}~|~ \boldsymbol{0}_{K}\leq x^{i}\leq x_{\max}^{i}~\text{and}~\boldsymbol{1}_{K}^{\top}x^{i}= E^{i}\}.\nonumber
\end{equation}

Following~\cite{charging2011}, the global cost function is given by
\begin{equation}
\textstyle F(x)=p\Big(\frac{1}{C_{\text{tot}}}\left(\sum_{i=1 }^{m}x^{i}+d\right)\Big)^{\top}\sum_{i=1}^{m}({x}^{i}+d^{i}),\label{Evcost}
\end{equation}
where $p(r):\mathbb{R}_{+}^{K}\mapsto\mathbb{R}_{+}^{K}$ is an electricity price function over $\mathcal{K}$, with each element $[p(r)]_{k}$ being an increasing function of $r$ for all $k=1,\cdots,K$. Here, $d=\sum_{i=1}^{m}d^{i}\in\mathbb{R}_{+}^{K}$ denotes the non-EV demand of all EV users, where $d^{i}$ represents EV user $i$'s basic electricity demand for its daily non-charging purposes. $C_{\text{tot}}\in\mathbb{R}_{+}$ represents the total generation capacity. 

From~\eqref{Evcost}, we have that minimizing the global cost can be formulated as an aggregative optimization problem:
\begin{equation}
\begin{aligned}
\text{min}_{x\in\mathcal{X}}~F(x)&\textstyle=\sum_{i=1}^{m}p(\phi(x))^{\top}(x^{i}+d^{i}),\\
\phi(x)&\textstyle=\frac{1}{m}\sum_{i=1}^{m}\frac{m}{C_{\text{tot}}}\left(x^{i}+d^{i}\right).~\label{primalEV}
\end{aligned}
\end{equation}

To solve problem~\eqref{primalEV}, all EV users typically execute a distributed aggregative optimization algorithm to collaboratively determine an optimal charging schedule $x^*$ that can minimize the global cost $F(x)$. However, in this process, untruthful EV users may provide untruthful information to the algorithm for personal benefits such as preferred charging time or reduced cost~\cite{trEVcharging,truthfulEV}, causing the computed charging schedule to deviate from the true $x^*$, and hence, increase the global cost. For the sake of simplicity, we exemplify the idea using a simple network composed of one untruthful EV 1 and one truthful EV 2.
As shown in Fig.~\ref{challenge}, we assume that EV user 1 prefers to charge after midnight. To reduce its own costs, EV user 1 may provide a falsified non-EV demand (red curve in Fig.~\ref{challenge}), claiming a lower than actual non-EV demand before midnight and a higher than actual non-EV demand after midnight. This untruthful information manipulates the algorithm into scheduling more charging power for EV $2$ before midnight and less charging power after midnight (dashed black curve in Fig.~\ref{challenge}). Because EV user 2 follows the charging schedule suggested by the algorithm, the total electricity consumption after midnight decreases, resulting in lower electricity prices at the time when EV user 1 really wants to charge. As a result, EV user 1 reduces its own costs at the expense of an increased global cost. In addition to distributed EV charging applications, similar truthfulness issues also emerge in traffic navigation~\cite{Waze}, vehicle platooning~\cite{VCG6}, and many other scenarios.

Motivated by these observations, we aim to restrict the loss reduction that an agent can gain through untruthful behaviors in distributed aggregative optimization.
\vspace{-0.5em}

\subsection{Truthfulness and joint differential privacy}\label{SectionIIC}
In this subsection, we introduce the concept of truthfulness in a fully distributed scenario, where no centralized aggregator exists to aggregate private information/decision variables and execute a truthfulness mechanism. It is in distinct contrast to existing JDP-based truthfulness framework in~\cite{truthfulEV,JDPoptimization1,JDPoptimization2} and VCG-based truthfulness framework in~\cite{VCG1,VCG2,VCG3,VCG4,VCG5,VCG6}, all of which rely on a ``centralized" aggregator to collect private information/decision variables from participating agents and then execute truthfulness mechanisms. In fact, these approaches require that the aggregator has the ability to verify information reported by participating agents, which may not be practical in distributed multi-agent networks. In fact, in distributed aggregative optimization, each agent $i$ possesses $\mathcal{P}_{i}\triangleq \text{col}(f_{i},g_{i},\mathcal{X}_{i})$ and shares iteration variables with its neighbors to collaboratively minimize the global objective function. This information exchange, combined with the fact that $\mathcal{P}_{i}$ is only known to agent $i$, provides opportunities for every agent to deceptively share false information in its own favor to minimize its own loss. We consider the worst scenario where every agent may act untruthfully for its personal benefit.

To facilitate truthfulness analysis, we denote the distributed aggregative optimization problem in~\eqref{primal} by a tuple $(\mathcal{X},\mathcal{J},\mathcal{G})$, with $\mathcal{X}\!=\!\{\mathcal{X}_{i}\}_{i=1}^{m}$, $\mathcal{J}\!=\!\{f_{i}\}_{i=1}^{m}$, and $\mathcal{G}\!=\!\{g_{i}\}_{i=1}^{m}$. Then, we define adjacency between distributed aggregative optimization problems as follows: 

\begin{definition}[Adjacency~\cite{truthfulEV}]\label{Adjacency}
Two aggregative optimization problems $\mathcal{P}=(\mathcal{J},\mathcal{G},\mathcal{X})$ and $\mathcal{P}'=(\mathcal{J}',\mathcal{G}',\mathcal{X}')$ are adjacent if there exists an $i\in[m]$ such that $\text{col}(f_{i},g_{i},\mathcal{X}_{i})\neq \text{col}({f}'_{i},{g}'_{i},\mathcal{X}_{i}')$, but $\text{col}(f_{j},g_{j},\mathcal{X}_{j})=\text{col}({f}'_{j},{g}'_{j},\mathcal{X}_{j}')$ for all $j\in[m]$ and $j\neq i$.
\end{definition}

Definition~\ref{Adjacency} implies that two distributed aggregative optimization problems $\mathcal{P}$ and $\mathcal{P}'$ are adjacent if and only if $\mathcal{P}$ and $\mathcal{P}'$ differ in a single entry while all other entries are the same. We denote the adjacent relationship between $\mathcal{P}$ and $\mathcal{P}'$ as $\text{Adj}(\mathcal{P},\mathcal{P}')$.
\begin{remark}
\vspace{-0.2em}
It is important to note that the adjacency relationship can also be defined in terms of shared intermediate updates, which aids in the truthful analysis of the manipulation of these updates. Specifically, we denote the shared intermediate update variable as $\theta=\text{col}(\theta^{1},\cdots,\theta^{m})$ with $\theta^{i}=\{\theta_{t}^{i}\}_{t=0}^{T}$ for any $i\in[m]$. Then, two shared intermediate update variables $\theta$ and $\theta^{\prime}$ are adjacent if there exists an $i\in[m]$ such that $\theta^{i}\neq \theta^{\prime i}$, but $\theta^{j}=\theta^{\prime j}$ for all $j\in[m]$ and $j\neq i$. In fact, employing this alternative definition does not change our truthfulness analysis. This is because, according to the update rules of our algorithm, any change in an agent's input naturally leads to changes of its intermediate updates. Hence, our truthfulness analysis technique is also applicable to the case of manipulating shared intermediate update variables. However, as mentioned earlier, such manipulation is difficult for agents to exploit effectively to achieve personal gains, since only the final computational outcome---and not the intermediate update values---is executed in the final implementation and impacts agents' personal gains. Therefore, we only focus on the manipulation of the algorithm's input.
\vspace{-0.2em}
\end{remark}

For a distributed aggregative optimization problem $\mathcal{P}$ and an initial state $\boldsymbol{\vartheta}_{0}$ (note that $\boldsymbol{\vartheta}$ can be an augmentation of multiple iteration variables), we denote a randomized execution as $\mathcal{A}_{\boldsymbol{\vartheta}_{0}}(\mathcal{P})$. Using this notation, we introduce the definition of $\eta$-truthfulness as follows:

\begin{definition}[$\eta$-truthfulness~\cite{DPtruthful3}]\label{truthfulness}
For a given $\eta>0$ and any $T\in\mathbb{N}^{+}$ (which includes the case of $T\!=\!\infty$), a distributed algorithm for problem~\eqref{primal} is $\eta$-truthful if for any agent $i$, any two adjacent problems $\mathcal{P}$ and $\mathcal{P'}$, and any outputs $x_{T}$ and $x'_{T}$ obtained by the algorithm's executions $\mathcal{A}_{\boldsymbol{\vartheta}_{0}}(\mathcal{P})$ and $\mathcal{A}_{\boldsymbol{\vartheta}_{0}}(\mathcal{P'})$, respectively, we always have 
\begin{equation}
\begin{aligned}
&\textstyle\mathbb{E}_{x_{T}\sim\mathcal{A}_{\boldsymbol{\vartheta}_{0}}(\mathcal{P})}[f_{i}(x_{T}^{i},\phi(x_{T}))]\\
&\textstyle\leq \mathbb{E}_{x'_{T}\sim\mathcal{A}_{\boldsymbol{\vartheta}_{0}}(\mathcal{P'})}[f_{i}({x'}_{\hspace{-0.1cm}T}^{i},\phi(x'_{T}))]+\eta,\label{truthful}
\end{aligned}
\end{equation}
where $x_{T}$ and $x'_{T}$ are the stacked column vectors of decision variables $x_{T}^{i}$ and ${x'}_{\hspace{-0.1cm}T}^{i}$ for all $i\in[m]$, respectively.
\end{definition}

In Definition~\ref{truthfulness}, if \eqref{truthful} holds for each agent $i\in[m]$, then the distributed algorithm is $\eta$-truthful (note that in the literature~\cite{truthfulEV,JDPoptimization1,JDPoptimization2}, Definition~\ref{truthfulness} is sometimes called $\eta$-approximate truthfulness because it cannot ensure that an agent will get absolute zero benefit from untruthful behaviors). The parameter $\eta$ quantifies the extent of loss reduction that an agent $i\in[m]$ can gain by using untruthful $\mathcal{P}'_{i}$ in distributed computation. It can be seen that a smaller $\eta$ implies a higher level of truthfulness.

We now define JDP, which is our main tool to ensure $\eta$-truthfulness:

\begin{definition}[Joint differential privacy~\cite{DPtruthful2}]\label{DefinitionJDP}
For a given $\epsilon>0$, a randomized distributed algorithm for problem~\eqref{primal} is $\epsilon$-jointly differentially private if for any $i\in[m]$, any two adjacent $\mathcal{P}$ and $\mathcal{P}'$, any set of outputs with the $i$th one removed $\mathcal{O}^{-i}\!\subseteq\! \mathbb{O}^{-i}$ (where $\mathbb{O}^{-i}$ denotes the set of all possible outputs except the $i$th one), and any initial state $\boldsymbol{\vartheta}_{0}$, we always have
\begin{equation}
\mathbb{P}[\mathcal{A}_{\boldsymbol{\vartheta}_{0}}(\mathcal{P})^{-i}\in\mathcal{O}^{-i}]\leq e^{\epsilon}\mathbb{P}[\mathcal{A}_{\boldsymbol{\vartheta}_{0}}(\mathcal{P}')^{-i}\in\mathcal{O}^{-i}],\nonumber
\end{equation}
where $\mathcal{A}_{\boldsymbol{\vartheta}_{0}}(\mathcal{P})^{-i}$ denotes the outputs of $\mathcal{A}_{\boldsymbol{\vartheta}_{0}}(\mathcal{P})$ for every agent except agent $i$.
\end{definition}

The parameter $\epsilon$ measures the indistinguishability under adjacent problems, and hence, the strength of JDP. A smaller value of $\epsilon$ indicates a stronger JDP. $\epsilon$-JDP is fundamentally different from the conventional $\epsilon$-differential privacy definition, which requires the indistinguishability of the outputs of all agents in distribution under two adjacent $\mathcal{P}$ and $\mathcal{P}'$ (including the output of agent $i$):
\begin{definition}[Differential privacy~\cite{Dwork2014}]\label{DP}
For a given $\epsilon>0$, a randomized distributed
algorithm for problem~\eqref{primal} is $\epsilon$-differentially private if for any two adjacent $\mathcal{P}$ and $\mathcal{P}'$, any set of outputs $\mathcal{O}\subset \mathbb{O}$ (with $\mathbb{O}$ denoting the set of all possible outputs), and any initial state $\boldsymbol{\vartheta_{0}}$, we always have
\begin{equation}
\mathbb{P}[\mathcal{A}_{\boldsymbol{\vartheta}_{0}}(\mathcal{P})\in\mathcal{O}]\leq e^{\epsilon}\mathbb{P}[\mathcal{A}_{\boldsymbol{\vartheta}_{0}}(\mathcal{P}')\in\mathcal{O}].\nonumber
\end{equation}
\end{definition}

As indicated in~\cite{DPtruthful3,DPtruthful2}, $\epsilon$-differential privacy is too stringent for ensuring truthfulness. Specifically, $\epsilon$-differential privacy requires that the input information of any single agent has a negligible influence on the outputs of all agents, which
makes the computed results almost unaffected by agents' information, thereby eliminating their incentive to truthfully report information.
\begin{remark}
\vspace{-0.2em}
The reason for the $\epsilon$-JDP mechanism to achieve $\eta$-truthfulness is that $\epsilon$-JDP can ensure that any single agent's input only has a small effect on the joint distribution of the outputs of other agents. This property
limits the benefit that an agent can gain from sharing untruthful/deceitful information to others, and hence, incentivizes truthful behaviors. Intuitively, a smaller $\epsilon$ enforces stronger JDP, implying that an untruthful agent's input has a smaller influence on the decision variables of other agents, thereby leading to a smaller $\eta$ and a stronger truthfulness guarantee. In fact, this intuition has been formalized in previous studies on server-assisted distributed optimization~\cite{truthfulEV} and proxy-based Bayesian games~\cite{DPtruthful2} (see Theorem 10 in~\cite{truthfulEV} and Theorem 6 in~\cite{DPtruthful2} for details). In our paper, we provide the quantitative relationship between $\eta$ and $\epsilon$ in Theorem~\ref{T2}.
\vspace{-0.8em}
\end{remark}

\section{Truthful distributed aggregative optimization algorithm design}\label{SectionIII}
The proposed truthful distributed algorithm to solve problem~\eqref{primal} is summarized in Algorithm~\ref{algorithm1}, in which Laplace noises $\zeta_{t}^{i}\in\mathbb{R}^{d}$ and $\xi_{t}^{i}\in\mathbb{R}^{d}$ with the following properties are employed to enable truthfulness:

\textbf{Noise parameter setting:}
For every $i\in[m]$ and any time $t\in\mathbb{N}$, the Laplace noises $\zeta_{t}^{i}$ and $\xi_{t}^{i}$ are zero-mean and independent over iterations. The noise variance $\mathbb{E}[\|\zeta_{t}^{i}\|^2]=(\sigma_{t,\zeta}^{i})^2$ is set to decrease with time and has a form of $\sigma_{t,\zeta}^{i}=\sigma_{\zeta}(t+1)^{-\varsigma_{\zeta}^{i}}$ with $\sigma_{\zeta}>0$ and $0<\varsigma_{\zeta}^{i}<1$. The noise variance $\mathbb{E}[\|\xi_{t}^{i}\|^2]=(\sigma_{t,\xi}^{i})^2$ is also set to decrease with time following $\sigma_{t,\xi}^{i}=\sigma_{\xi}(t+1)^{-\varsigma_{\xi}^{i}}$ with $\sigma_{\xi}>0$ and $0\!<\varsigma_{\xi}^{i}\!<1$. For example, we can set Laplace noise variances as $\sigma_{t,\zeta}^{i} = (t+1)^{-0.57}$ and $\sigma_{t,\xi}^{i} = (t+1)^{-0.79}$.

We use Laplace noises with decreasing variances to achieve truthfulness, which, as shown in Theorem~\ref{T1}, is key to maintaining accurate convergence. 
\begin{algorithm}[H]
\caption{Truthful distributed aggregative optimization algorithm for agent $i\in[m]$}
\label{algorithm1} 
\begin{algorithmic}[1]
\State {\bfseries Input:} Initialization $x_{0}^{i}\in \mathcal{X}_{i}$, $\psi_{0}^{i}=g_{i}(x_{0}^{i})$, and $y_{0}^{i}\!=\!\nabla_{2}f_{i}(x_{0}^{i},\psi_{0}^{i})$; 
stepsize $\lambda_{t}\!=\!\frac{\lambda_{0}}{(t+1)^{u}}$ with $\lambda_{0}\!>\!0$ and $u\!>\!0$;
decaying sequences $\alpha_{t}\!=\!\frac{\alpha_{0}}{(t+1)^{v}}$, $\gamma_{t,1}\!=\!\frac{\gamma_{1}}{(t+1)^{w_{1}}}$, and $\gamma_{t,2}\!=\!\frac{\gamma_{2}}{(t+1)^{w_{2}}}$ with $\alpha_{0},\gamma_{1},\gamma_{2}\!>\!0$ and $v,w_{1},w_{2}\!>\!0$; 
set $\Omega_{0}\!=\!\mathcal{B}_{L_{f,2}}$ and $\Omega_{t}\!=\!\mathcal{B}_{(1+\sum_{p=0}^{t-1}\gamma_{p,1})L_{f,2}}$ for any $t
\!\in\!\mathbb{N}^{+}$; Laplace noises $\zeta_{t}^{i}$ and $\xi_{t}^{i}$ following the ``Noise parameter setting".
\For{$t=0,1,\cdots,T-1$}
\State \!Receive neighbors' parameters $\tilde{y}_{t}^{j}=y_{t}^{j}+\zeta_{t}^{j},~j\in\mathcal{N}_{i}$;
\State  \!$y_{t+1}^{i}\!=\!(1+w_{ii})y_{t}^{i}+\sum_{j\in{\mathcal{N}_{i}}}w_{ij}\boldsymbol{P}_{\Omega_{t}}(\tilde{y}_{t}^{j})$;
\Statex \hspace{4em} $+\gamma_{t,1}\nabla_{2}f_{i}(x_{t}^{i},\psi_{t}^{i});$
\State \!$x_{t+1}^{i}\!=\!\boldsymbol{P}_{\mathcal{X}_{i}}\!\!\left(x_{t}^{i}\!-\!\lambda_{t}\left[\nabla_{1}f_{i}(x_{t}^{i},\psi_{t}^{i})\!+\!\frac{\nabla g_{i}(x_{t}^{i})(y_{t+1}^{i}-y_{t}^{i})}{\gamma_{t,1}}\right]\right)$;
\State \!Receive neighbors' parameters $\tilde{\psi}_{t}^{j}=\psi_{t}^{j}+\xi_{t}^{j},~j\in\mathcal{N}_{i}$;
\State  \!$\psi_{t+1}^{i}\!=\!(1-\alpha_{t}+\gamma_{t,2}w_{ii})\psi_{t}^{i}+\gamma_{t,2}\sum_{j\in{\mathcal{N}_{i}}}w_{ij}\tilde{\psi}_{t}^{j}$
\Statex \hspace{4em} \!$+g_{i}(x_{t+1}^{i})-(1-\alpha_{t})g_{i}(x_{t}^{i});$
\State Add Laplace noises $\zeta_{t+1}^{i}$ and $\xi_{t+1}^{i}$ to $y_{t+1}^{i}$ and $\psi_{t+1}^{i}$, respectively, and send the 
obscured parameters $\tilde{y}_{t+1}^{i}=y_{t+1}^{i}+\zeta_{t+1}^{i}$ and $\tilde{\psi}_{t+1}^{i}=\psi_{t+1}^{i}+\xi_{t+1}^{i}$ to its neighbors.
\EndFor
\end{algorithmic}
\end{algorithm}

In Algorithm~\ref{algorithm1}, Line 4 enables agent $i$ to track the weighted global gradient $\gamma_{t,1}\frac{1}{m}\sum_{i=1}^{m}\nabla_{2}f_{i}(x_{t}^{i},\phi(x_{t}))$, which is not directly accessible to individual agents. Similarly, Line 7 allows agent $i$ to track the global average $\phi(x_{t}) = \frac{1}{m}\sum_{i=1}^{m}g_{i}(x_{t}^{i})$, which is also unavailable to individual agents. Line 5 implements a projected local gradient descent step for agent $i$ to update its decision variable $x_{t}^{i}$. Furthermore, to ensure strong truthfulness, persistent Laplace noise is injected into all shared messages (see Line 8 for details).

As proven later in Theorem~\ref{T1} in Sec.~\ref{SectionIV}, our algorithm can ensure accurate convergence despite the presence of Laplace noises. This is in distinct difference from existing approaches for distributed aggregative optimization that employ the conventional gradient-tracking technique (see, e.g.,~\cite{lixiuxian2021,aggregative2,aggregative3,aggregative4,aggregative6,lixiuxianonline,aggregative7,aggregative8,aggregative9,carnevale2024nonconvex}), which are not robust to noises due to the accumulation of noise variances in
local agents' estimate of the global gradient $\frac{1}{m}\sum_{i=1}^{m}\nabla_{2}f_{i}(x_{t}^{i},\phi(x_{t}))$\footnote{As indicated in~\cite{GT1,GT2,wangGT,zijiGT}, the injected Laplace noises in conventional gradient-tracking-based algorithms will accumulate over time in the estimate of global gradients. Hence, if we employ existing gradient-tracking-based aggregative optimization algorithms, in e.g.,~\cite{lixiuxian2021,aggregative2,aggregative3,aggregative4,aggregative6,lixiuxianonline,aggregative7,aggregative8,aggregative9,carnevale2024nonconvex}, where the gradient-tracking variable $y_{t}^{i}$ is directly fed into the decision variable update, then optimization accuracy will be compromised. A detailed discussion on this issue is available in Section III of~\cite{wangGT} and Section III-A of~\cite{zijiGT}.}. To the best of our knowledge, our approach is the first to successfully eliminate the effect of unbounded Laplace noises on convergence accuracy in distributed aggregative optimization.

One key reason for our algorithm to eliminate the effect of Laplace noises on convergence accuracy is the use of a robust gradient-tracking method in Line 4. Specifically, defining $\boldsymbol{y}_{t}\!=\!\text{col}((y_{t}^{1})^{\top},\cdots,(y_{t}^{m})^{\top})\!\in\!\mathbb{R}^{m\times d}$, $\zeta_{t,w}^{i}\!=\!\sum_{j\in\mathcal{N}_{i}}w_{ij}\boldsymbol{P}_{\Omega_{t}}(\zeta_{t}^{j})\!\in\!\mathbb{R}^{m\times d}$, $\boldsymbol{\zeta}_{t,w}\!=\!\text{col}((\zeta_{t,w}^{1})^{\top},\cdots,(\zeta_{t,w}^{m})^{\top})\!\in\!\mathbb{R}^{m\times d}$,  and $\nabla_{2}\boldsymbol{f}(x_{t},\psi_{t})\!=\!\text{col}(\nabla_{2}f_{1}(x_{t}^{1},\psi_{t}^{1})^{\top},\cdots,\nabla_{2}f_{m}(x_{t}^{m},\psi_{t}^{m})^{\top})\!\in\!\mathbb{R}^{m\times d}$, Line 4 in Algorithm~\ref{algorithm1} implies
\begin{equation}
\boldsymbol{1}_{m}^{\top}(\boldsymbol{y}_{t+1}-\boldsymbol{y}_{t})=\boldsymbol{1}_{m}^{\top}(\boldsymbol{\zeta_{t,w}}+\gamma_{t,1}\nabla_{2}\boldsymbol{f}(x_{t},\psi_{t})),\nonumber
\end{equation}
where we have used the property $\boldsymbol{1}_{m}^{\top}W=\boldsymbol{0}_{m}^{\top}$. It is clear that the robust gradient tracking method avoids Laplace noises from accumulating on the gradient estimate (which is a problem faced by the conventional gradient tracking approach), and thus ensures robustness against noises. The robust gradient tracking method is inspired by our prior work~\cite{wangGT,zijiGT} as well as others~\cite{GT1,GT2} on distributed optimization. Unlike~\cite{GT1,GT2,wangGT,zijiGT} where each $f_{i}$ is only dependent on the local decision variable $x^{i}$, here we consider aggregative optimization problems where each $f_{i}$ depends not only on the local decision $x^{i}$ but also on the aggregative term $\phi(x)$. This compositional structure complicates both convergence and truthfulness analysis. Moreover, different from~\cite{GT1,GT2,wangGT,zijiGT} where the tracking variable $y_{t}^{i}$ can be
any point in $\mathbb{R}^{d}$, our algorithm constraints $y_{t}^{i}$ in a time-varying set $\Omega_{t}$ (see Lemma~\ref{ybound} for details), which makes convergence analysis much more challenging because the nonlinearity induced by projection
(necessary to address set constraints) poses challenges to both optimality analysis and consensus characterization. In addition, it is worth noting that compared with
the results in~\cite{GT1,GT2} where the objective function is strongly convex, the objective function considered here can be nonconvex/convex/strongly convex.     

Another key reason for our algorithm to be robust to Laplace noises is the use of decaying sequences $\alpha_{t}$ and $\gamma_{t,2}$ in Line 7, which can effectively suppress the influence of Laplace noises on the estimation variable $\psi_{t}^{i}$. Specifically, defining $\xi_{t,w}^{i}\!=\!\sum_{j\in\mathcal{N}_{i}}w_{ij}\xi_{t}^{j}$, $\boldsymbol{\xi}_{t,w}\!=\!\text{col}((\xi_{t,w}^{1})^{\top},\cdots,(\xi_{t,w}^{m})^{\top})\!\in\!\mathbb{R}^{m\times d}$, $\boldsymbol{\psi}_{t}\!=\!\text{col}((\psi_{t}^{1})^{\top},\cdots,(\psi_{t}^{m})^{\top})\!\in\!\mathbb{R}^{m\times d}$, and $\boldsymbol{g}(x_{t})\!=\!\text{col}(g_{1}(x_{t}^{1})^{\top},\cdots,g_{m}(x_{t}^{m})^{\top})\!\in\!\mathbb{R}^{m\times d}$, Line 7 in Algorithm~\ref{algorithm1} implies
\begin{equation}
\begin{aligned}
\boldsymbol{1}_{m}^{\top}\boldsymbol{\psi}_{t+1}&=\boldsymbol{1}_{m}^{\top}\big((1-\alpha_{t})\boldsymbol{\psi}_{t}+\gamma_{t,2}\boldsymbol{\xi}_{t,w}\\
&\quad+\boldsymbol{g}(x_{t+1})-(1-\alpha_{t})\boldsymbol{g}(x_{t})\big),\nonumber
\end{aligned}
\vspace{-0.2em}
\end{equation}
where we have used $\boldsymbol{1}_{m}^{\top}W=\boldsymbol{0}_{m}^{\top}$ and $w_{ii}=-\sum_{j\in\mathcal{N}_{i}}w_{ij}$.

Using the initial condition $\boldsymbol{\psi}_{0}=\boldsymbol{g}(x_{0})$, we can establish the following equality through induction:
\begin{equation}
\begin{aligned}
&\boldsymbol{1}_{m}^{\top}\boldsymbol{\psi}_{t+1}=\boldsymbol{1}_{m}^{\top}\big(\boldsymbol{g}(x_{t+1})\\
&\quad+{\textstyle\sum_{l=0}^{t-1}\prod_{p=l}^{t-1}}(1-\alpha_{p+1})\gamma_{l,2}\boldsymbol{\xi}_{l,w}+\gamma_{t,2}\boldsymbol{\xi}_{t,w}\big).\label{noseacc2}
\end{aligned}
\end{equation}
By judiciously designing the decaying rates of the sequences $\alpha_{t}$ and $\gamma_{t,2}$, we can gradually attenuate the influence of Laplace noises on the estimation variable $\boldsymbol{\psi}_{t}$, and hence, on the final optimization result.

Moreover, the co-design of stepsizes $\lambda_{t}$, decaying sequences $\alpha_{t}$, $\gamma_{t,1}$, and $\gamma_{t,2}$, as well as the Laplace-noise mechanism is crucial for our algorithm to ensure $\eta$-truthfulness. By judiciously coordinating the rates of stepsizes, decaying sequences,  and Laplace-noise variances, we can ensure truthfulness even in the infinite time horizon. This is different from existing JDP-based truthfulness results on cooperative optimization (e.g.,~\cite{truthfulEV,JDPoptimization1,JDPoptimization2}), whose truthfulness guarantee diminishes as iterations tend to infinity.

In Algorithm~\ref{algorithm1}, each agent $i$ projects the received parameters $y_{t}^{j}+\zeta_{t}^{j}$ from its neighbors $j\in\mathcal{N}_{i}$ onto an expanding closed ball $\Omega_{t}$. This strategy ensures the following result:
\begin{lemma}\label{ybound}
Under Assumption~\ref{A1}-(i), for any $t
\in\mathbb{N}$, the gradient-tracking variable $y_{t}^{i}$ of Algorithm~\ref{algorithm1} is constrained in the closed ball $\Omega_{t}=\mathcal{B}_{(1+\sum_{p=0}^{t-1}\gamma_{p,1})L_{f,2}}$.
\end{lemma}
\begin{proof}
See Appendix~A.
\end{proof}
According to Lemma~\ref{ybound} and the relationship $\sum_{p=0}^{t-1}\gamma_{p,1}\leq \gamma_{1}\!+\!\int_{1}^{\infty}\frac{\gamma_{1}}{x^{w_{1}}}\!=\!\frac{\gamma_{1}w_{1}}{w_{1}-1}$, we can ensure $\|y_{t}^{i}\|\leq (1+\frac{\gamma_{1}\omega_{1}}{\omega_{1}-1})L_{f,2}$ for any $t\in\mathbb{N}$ and $w_{1}>1$. The boundedness of $\|y_{t}^{i}\|$ is necessary for our $\eta$-truthfulness analysis (see Eqs.~\eqref{Dp3}-\eqref{Dp51} for details).

\begin{remark}
We employ two distinct tracking techniques in Algorithm~\ref{algorithm1} since they serve different purposes in our algorithm design. More specifically, by using the conventional tracking approach to update $\psi_{t}^{i}$, $\psi_{t}^{i}$ represents the tracking target function and we can use $\psi_{t}^{i}$ directly to evaluate the gradients $\nabla_{1}f_{i}(x_{t}^{i},\psi_{t}^{i})$ and $\nabla_{2}f_{i}(x_{t}^{i},\psi_{t}^{i})$ in each iteration. In contrast, if we were to use a robust tracking approach in which $\psi_{t}^{i}$ represents a cumulative value of the target function, the gradient evaluation would have to be performed on the difference of local iterates, i.e., $\nabla_{1}f_{i}(x_{t}^{i}, \frac{\psi_{t+1}^{i}-\psi_{t}^{i}}{\gamma_{t,2}})$ and $\nabla_{2}f_{i}(x_{t}^{i}, \frac{\psi_{t+1}^{i}-\psi_{t}^{i}}{\gamma_{t,2}})$, which would introduce unnecessary complexity in both the algorithm design and convergence analysis. Unlike $\psi_{t}^{i}$, the variable $y_{t}^{i}$ is used to estimate the global gradient $\gamma_{t,1}\frac{1}{m}\sum_{i=1}^{m}\nabla_{2}f_{i}(x_{t}^{i},\psi_{t}^{i})$ that does not participate in the gradient computation. Hence, we employ robust gradient tracking to update $y_{t}^{i}$, which not only mitigates the effect of noises on convergence accuracy but also avoids introducing additional decaying sequences required by conventional tracking approaches to combat noises (e.g., $\alpha_{t}$ and $\gamma_{t,2}$ in Line 7).
\end{remark}
\begin{remark}
Different from existing JDP-based truthfulness approaches for cooperative optimization~\cite{truthfulEV,JDPoptimization1,JDPoptimization2}, which require a ``centralized" aggregator to collect all agents' iteration/decision variables, our algorithm is executed in a fully distributed scenario without the assistance of any aggregators. Moreover, our algorithm avoids sharing truthful gradient/function information to a centralized aggregator, which is necessary in existing VCG-based approaches~\cite{VCG1,VCG2,VCG3,VCG4,VCG5,VCG6}.
\end{remark}
\begin{remark}
Although the approach of JDP-enabled truthfulness approach inherits nice properties of JDP, such as robustness to collusion and amenability to composition (which are advantages over VCG-based truthfulness approaches), it does have one downside compared with VCG-based truthfulness approaches: Due to the absence of monetary incentives, agents may have less incentive to participate in the distributed algorithm and may instead choose to make decisions solely based on their personal preferences.
\end{remark}

\section{Optimization Accuracy and Convergence Rate Analysis}\label{SectionIV}
In this section, we systematically analyze the convergence rate of Algorithm~\ref{algorithm1} under strongly convex, general convex, and nonconvex global objective functions, respectively. To this end, we first give the global gradient $\nabla F(x_{t})$ and its estimate $\nabla \tilde{F}(x_{t})$ based on the definition of $F(x)$ in~\eqref{primal} and Line 5 in  Algorithm~\ref{algorithm1} as follows:
\begin{align}
\vspace{-0.2em}
&\textstyle\nabla F(x_{t})=\sum_{i=1}^{m}\nabla_{1}f_{i}(x_{t}^{i},\phi(x_{t}))\nonumber\\
&\textstyle\quad+\sum_{i=1}^{m}\big(\nabla g_{i}(x_{t}^{i})\times\frac{1}{m}\sum_{i=1}^{m} \nabla_{2}f_{i}(x_{t}^{i},\phi(x_{t}))\big).\label{A2L1}\\
&\textstyle\nabla\tilde{F}(x_{t})=\sum_{i=1}^{m}\big(\nabla_{1}f_{i}(x_{t}^{i},\psi_{t}^{i})+\nabla g_{i}(x_{t}^{i})\frac{y_{t+1}^{i}-y_{t}^{i}}{\gamma_{t,1}}\big).\label{Hg1}
\vspace{-0.2em}
\end{align}

For the sake of notational simplicity, we denote $L_{f}=\max\{L_{f,1},L_{f,2}\}$, $\bar{L}_{f}=\max\{\bar{L}_{f,1},\bar{L}_{f,2}\}$, $\varsigma_{\zeta}=\min_{i\in[m]}\{\varsigma_{\zeta}^{i}\}$, and $\varsigma_{\xi}=\min_{i\in[m]}\{\varsigma_{\xi}^{i}\}$. 

The following lemma establishes the $L_{F}$-Lipschitz continuity of the global gradient $\nabla F(x)$:  
\begin{lemma}\label{FLips}
Under Assumption~\ref{A1}, the gradient $\nabla F(x)$ is $L_{F}$-Lipschitz continuous, i.e., for any $x_{1},x_{2}\in\mathcal{X}$, we have
\begin{equation}
\|\nabla F(x_{1})-\nabla F(x_{2})\|\leq L_{F}\|x_{1}-x_{2}\|,\nonumber
\end{equation}
with $L_{F}=\bar{L}_{f,1}+L_{f}\bar{L}_{g}+L_{g}(\bar{L}_{f,1}+\bar{L}_{f,2}+L_{g}\bar{L}_{f,2}).$ 
\end{lemma}
\begin{proof}
See Appendix~B.
\end{proof}

We also present the following lemma to characterize the estimation error of $\nabla \tilde{F}(x_{t})$.
\begin{lemma}\label{Lemhyperg}
Under Assumptions~\ref{A1} and \ref{A2}, if the rate of stepsize $\lambda_{t}$ satisfies $1>u>w_{2}$, the rate of decaying sequence $\alpha_{t}$ satisfies $1>v>w_{2}$, and the rates of Laplace-noise variances satisfy $\varsigma_{\zeta}>\max\{w_{1},\frac{w_{2}}{2}\}$ and $\varsigma_{\xi}>\frac{v}{2}-w_{2}$, then we have the following inequalities for Algorithm~\ref{algorithm1}:
\begin{align}
&\mathbb{E}[\|\nabla\tilde{F}(x_{t})-\nabla F(x_{t})\|^2]\leq \mathcal{O}((t+1)^{-\beta}) ,\label{hypergresult}\\
&\mathbb{E}[\|\nabla\tilde{F}(x_{t})\|^2]\leq \mathcal{O}(1),\label{tildeFresult}
\end{align}
where the convergence rate $\beta$ is given by $\beta=\min\{2u-2w_{2},2v-2w_{2},2\varsigma_{\zeta}-2w_{1},2\varsigma_{\zeta}-w_{2},2\varsigma_{\xi}+2w_{2}-v\}$.
\end{lemma}
\begin{proof}
See Appendix~D.
\end{proof}
Now, we establish the convergence results of Algorithm~\ref{algorithm1}:
\begin{theorem}\label{T1}
Under Assumptions~\ref{A1} and \ref{A2}, for any $T\in\mathbb{N}^{+}$, the following convergence results hold for Algorithm~\ref{algorithm1}:

(i) If $F(x)$ is strongly convex, the rate of stepsize $\lambda_{t}$ satisfies $1>u>w_{2}$, the rates of decaying sequences $\alpha_{t}$, $\gamma_{t,1}$, and $\gamma_{t,2}$ satisfy $1>v>w_{2}$, $1>w_{1}$, and $1>w_{2}$, respectively, and the rates of Laplace-noise variances satisfy $\varsigma_{\zeta}>\max\{w_{1},\frac{w_{2}}{2}\}$ and $\varsigma_{\xi}>\frac{v}{2}-w_{2}$, then we have
\begin{equation}
\textstyle \mathbb{E}[\|x_{T}^{i}-x^{*}\|^2]\leq \mathcal{O}\left(T^{-\beta}\right),\label{T1result1}
\end{equation}
where the convergence rate $\beta$ is given by $\beta=\min\{2u-2w_{2},2v-2w_{2},2\varsigma_{\zeta}-2w_{1},2\varsigma_{\zeta}-w_{2},2\varsigma_{\xi}+2w_{2}-v\}$.

(ii) If $F(x)$ is general convex, the rate of stepsize $\lambda_{t}$ satisfies $1\!>u\!>\frac{1+w_{2}}{2}$, the rates of decaying sequences $\alpha_{t}$, $\gamma_{t,1}$, and $\gamma_{t,2}$ satisfy $1>v>1-u+w_{2}$, $1>w_{1}$, and $1>w_{2}$, respectively, and the rates of Laplace-noise variances satisfy $\varsigma_{\zeta}>1-u+\max\{w_{1},\frac{w_{2}}{2}\}$ and $\varsigma_{\xi}>1-u+\frac{v}{2}-w_{2}$, then we have
\begin{equation}
\frac{\sum_{t=0}^{T}\lambda_{t}\mathbb{E}[F(x_{t})-F(x^{*})]}{\sum_{t=0}^{T}\lambda_{t}}\leq \mathcal{O}\left(T^{-(1-u)}\right).\label{T1result2}
\end{equation}

(iii) If $F(x)$ is nonconvex, the rate of stepsize $\lambda_{t}$ satisfies $1>u>\max\{\frac{1}{2},\frac{1+2w_{2}}{3}\}$, the rates of decaying sequences $\alpha_{t}$, $\gamma_{t,1}$, and $\gamma_{t,2}$ satisfy $1>v>\frac{1-u}{2}+w_{2}$, $1>w_{1}$, and $1>w_{2}$, respectively, and the rates of Laplace-noise variances satisfy $\varsigma_{\zeta}>\frac{1-u}{2}+\max\{w_{1},\frac{w_{2}}{2}\}$ and $\varsigma_{\xi}>\frac{1-u}{2}+\frac{v}{2}-w_{2}$, then we have
\begin{equation}
\frac{\sum_{t=0}^{T}\lambda_{t}\mathbb{E}[\|\nabla F(x_{t})\|^2]}{\sum_{t=0}^{T}\lambda_{t}}\leq \mathcal{O}\left(T^{-(1-u)}\right).\label{TT1result2N9}
\end{equation}
\end{theorem}
\begin{proof}
See Appendix E.
\end{proof}

Theorem~\ref{T1} proves that Algorithm~\ref{algorithm1} converges to an exact optimal solution at rates $\mathcal{O}(T^{-\beta})$, $\mathcal{O}(T^{-(1-u)})$, and $\mathcal{O}(T^{-(1-u)})$ for strongly convex, convex, and nonconvex $F(x)$, respectively. It is more comprehensive than existing distributed aggregative optimization results in~\cite{lixiuxian2021,charging2011,carnevale2022online,aggregative2,aggregative3,aggregative4,aggregative5,aggregative6,lixiuxianonline,aggregative7,aggregative8}, which do not consider the nonconvex case. Our result is also stronger and more precise than the result on distributed nonconvex aggregative optimization in~\cite{carnevale2024nonconvex}, which proves asymptotic convergence but does not provide explicit convergence rates.
\begin{remark}
In Line 7 of Algorithm~\ref{algorithm1}, we introduce a
decaying sequence $\alpha_{t}$ to attenuate the influence of noises on convergence accuracy (see Eq.~\eqref{noseacc2}).  This sequence does not affect the asymptotic convergence of $\psi_{t}^{i}$ to $\phi(x_{t})$ (as evidenced by Eq.~\eqref{Hg6}), but it does influence the convergence rate of Algorithm~\ref{algorithm1}. Specifically, according to the convergence rate result in Theorem~\ref{T1}, i.e.,  $\beta=\min\{2u-2w_{2},2v-2w_{2},2\varsigma_{\zeta}-2w_{1},2\varsigma_{\zeta}-w_{2},2\varsigma_{\xi}+2w_{2}-v\}$ for strongly convex objective functions and $1-u$ for nonconvex/general convex objective functions, we have that a larger $v$ (which denotes the decaying rate of $\alpha_{t}$) leads to a larger $\beta$ and $1-u$ (here, we have used the relation $v-w_{2}>1-u$ given in the statement of Theorem~\ref{T1}-(ii) and-(iii), which implies that a larger $v$ leads to a larger $1-u$). Therefore, a faster decaying $\alpha_{t}$ leads to faster convergence.
\end{remark}

Next, we provide an example of parameter selection that satisfies the conditions given in Theorem~\ref{T1}:

\noindent\textbf{An example of parameter selection in Theorem~\ref{T1}:}

(i) For a strongly convex $F(x)$ and an arbitrarily small $\varepsilon>0$, if we set $u=v=\varsigma_{\zeta}=\frac{3}{4}$, $\varsigma_{\xi}=\frac{9}{8}$, and $w_{1}=w_{2}=\varepsilon$, then we have $\beta=\frac{3}{2}-2\varepsilon$, which implies that the convergence rate of Algorithm 1 is arbitrarily close to $\mathcal{O}(T^{-1.5})$.

(ii) For a general convex $F(x)$ and an arbitrarily small $\varepsilon>0$, if we set $u=v=\varsigma_{\zeta}=\frac{1}{2}+\varepsilon$, $\varsigma_{\xi}=\frac{3}{4}+\varepsilon$, and $w_{1}=w_{2}=\varepsilon$, then we have $1-u=\frac{1}{2}-\varepsilon$, which implies that the convergence rate of Algorithm 1 is arbitrarily close to $\mathcal{O}(T^{-0.5})$.

(iii) For a nonconvex $F(x)$ and an arbitrarily small $\varepsilon>0$, if we set $u=\frac{1}{2}+\varepsilon$, $v=\varsigma_{\zeta}=\frac{1}{4}+\varepsilon$, $\varsigma_{\xi}=\frac{3}{8}$, and $w_{1}=w_{2}=\varepsilon$, then we have $1-u=\frac{1}{2}-\varepsilon$, which implies that the convergence rate of Algorithm 1 is arbitrarily close to $\mathcal{O}(T^{-0.5})$.
\begin{remark}
The above results imply that our algorithm achieves a convergence rate arbitrarily close to $\mathcal{O}(T^{-0.5})$ for nonconvex/general convex objective functions even in the presence of Laplace noises. This closely matches the convergence rate of $\mathcal{O}(T^{-0.5})$ achieved by existing distributed aggregative optimization approaches for convex objective functions in the absence of noises (in, e.g.,~\cite{lixiuxianonline,aggregative6,aggregative7}). Moreover, unlike the result for nonconvex functions in~\cite{aggregative9}, which achieves a convergence rate of $\mathcal{O}(T^{-0.5})$ for the special case where every $f_{i}$ is solely dependent on the aggregative term in the absence of noises, our convergence results remain valid in a general scenario where $f_{i}$ depends on both local decision variables and the aggregative term.
\vspace{-1em}
\end{remark}

\section{Truthfulness Analysis}\label{SectionV}
\subsection{Truthfulness guarantee}\label{SectionIV2}
In this subsection, we prove that Algorithm~\ref{algorithm1} can ensure that one agent's untruthful behavior can gain at most $\eta$ reduction in its loss (i.e., $\eta$-truthfulness) in the implementation of Algorithm~\ref{algorithm1}. To this end, we first 
denote $\hat{\varsigma}_{\zeta}=\max_{i\in[m]}\{\varsigma_{\zeta}^{i}\}$ and $\hat{\varsigma}_{\xi}=\max_{i\in[m]}\{\varsigma_{\xi}^{i}\}$, and give a preliminary result:
\begin{lemma}\label{JDP}
Under Assumptions~\ref{A1} and \ref{A2}, if the rate of stepsize $\lambda_{t}$ satisfies $u>w_{1}+w_{2}+\hat{\varsigma}_{\xi}+1$ and the rates of decaying sequences $\alpha_{t}$, $\gamma_{t,1}$, and $\gamma_{t,2}$ satisfy $v>u-w_{1}$, $w_{1}>\hat{\varsigma}_{\zeta}+1$, and $1>w_{2}$,  respectively, then for any $T\!\in\!\mathbb{N}^{+}$, Algorithm~\ref{algorithm1} is $\epsilon$-jointly differentially private with the cumulative privacy budget bounded by
\begin{equation}
\textstyle\epsilon=\sum_{t=1}^{T}\Big(\frac{\sqrt{2}c_{1}\lambda_{0}}{\sigma_{\xi}\gamma_{1}\gamma_{2}(t+1)^{u-w_{1}-w_{2}-\hat{\varsigma}_{\xi}}}+\frac{\sqrt{2}c_{2}\gamma_{1}}{\sigma_{\zeta}(t+1)^{w_{1}-\hat{\varsigma}_{\zeta}}}\Big),\label{epsilon}
\end{equation}
where $c_{1}$ and $c_{2}$ are given by $c_{1}=\frac{\hat{w}\gamma_{2}}{\hat{w}\gamma_{2}-(u-w_{1}-w_{2})}$ and $c_{2}=(\frac{4w_{1}}{e\ln(\frac{2}{2-\hat{w}})})^{w_{1}}\frac{2}{\hat{w}}$, respectively, with $\hat{w}\triangleq\min_{i\in[m]}\{|w_{ii}|\}$.
\end{lemma}
\begin{proof}
See Appendix~F.
\end{proof}
\begin{remark}\label{remark3}
Lemma~\ref{JDP} proves that our algorithm ensures rigorous $\epsilon$-JDP, where the cumulative privacy budget $\epsilon$ is finite even when $T\!\rightarrow\!\infty$ since $u-w_{1}-w_{2}\!-\!\hat{\varsigma}_{\xi}\!>\!1$ and $w_{1}-\hat{\varsigma}_{\zeta}\!>\!1$ always hold. It also implies that for any given $\epsilon>0$, Algorithm~\ref{algorithm1} is $\epsilon$-jointly differentially private when the noise parameters satisfy $\sigma_{\xi}=\sum_{t=1}^{T}\frac{2\sqrt{2}c_{1}\lambda_{0}}{\epsilon\gamma_{1}\gamma_{2}(t+1)^{u-w_{1}-w_{2}-\hat{\varsigma}_{\xi}}}$ and $\sigma_{\zeta}=\sum_{t=1}^{T}\frac{2\sqrt{2}c_{2}\gamma_{1}}{\epsilon(t+1)^{w_{1}-\hat{\varsigma}_{\zeta}}}$ with $c_{1}$ and $c_{2}$ defined in Lemma~\ref{JDP}.
\end{remark}
We now prove $\eta$-truthfulness of Algorithm~\ref{algorithm1}:
\begin{theorem}\label{T2}
Under Assumptions~\ref{A1} and~\ref{A2}, if the decaying rate of stepsize $\lambda_{t}$ satisfies $u\!>\!w_{1}+w_{2}+\hat{\varsigma}_{\xi}+1$ and the rates of decaying sequences $\alpha_{t}$, $\gamma_{t,1}$, and $\gamma_{t,2}$ satisfy $v\!>\!u-w_{1}$, $w_{1}>\hat{\varsigma}_{\zeta}+1$, and $1>w_{2}$,  respectively, then for any $T\!\in\!\mathbb{N}^{+}$  (which includes the case of~$T\!=\!\infty$), Algorithm~\ref{algorithm1} is $\eta$-truthful, i.e., for any $i\in[m]$ and any $T\in\mathbb{N}^{+}$, we have
\begin{equation}
\begin{aligned}
&\mathbb{E}_{x_{T}\sim \mathcal{A}_{\boldsymbol{\vartheta}_{0}}(\mathcal{P})}[f_{i}(x_{T}^{i},\phi(x_{T}))]\\
&\leq \mathbb{E}_{x'_{T}\sim \mathcal{A}_{\boldsymbol{\vartheta}_{0}}(\mathcal{P}')}[f_{i}({x'}_{\hspace{-0.1cm}T}^{i},\phi(x'_{T}))]+\eta,\label{T4result}
\end{aligned}
\end{equation}
where $\eta$ is given by
\begin{equation}
\eta=(L_{f,1}+L_{f,2}L_{g})D_{\mathcal{X}}+2\epsilon D_{f}, \label{eta}
\end{equation}	
with $D_{\mathcal{X}}$ denoting the diameter of the compact set $\mathcal{X}$ and $D_{f}$ representing the value of $\max_{i\in[m]}\{\|f_{i}(x^{i},\phi(x))\|\}$ on set $\mathcal{X}$.
\end{theorem}
\begin{proof}
To prove Theorem~\ref{T2}, we first introduce the following decomposition:
\begin{equation}
\begin{aligned}
&f_{i}(x_{T}^{i},\phi(x_{T}^{i},x_{T}^{-i}))-f_{i}({x'}_{\hspace{-0.1cm}T}^{i},\phi({x'}_{\hspace{-0.1cm}T}^{i},x_{T}^{-i}))\\
&=f_{i}(x_{T}^{i},\phi(x_{T}^{i},x_{T}^{-i}))-f_{i}({x'}_{\hspace{-0.1cm}T}^{i},\phi(x_{T}^{i},x_{T}^{-i}))\\
&\quad+f_{i}({x'}_{\hspace{-0.1cm}T}^{i},\phi(x_{T}^{i},x_{T}^{-i}))-f_{i}({x'}_{\hspace{-0.1cm}T}^{i},\phi({x'}_{\hspace{-0.1cm}T}^{i},x_{T}^{-i})),\label{4t1}
\end{aligned}
\end{equation}
where we used symbol $(x_{T}^{i},x_{T}^{-i})$ to emphasize that $x_{T}$ can be decomposed into $x_{T}^{i}$ and  $x_{T}^{-i}\!=\!\{x_{T}^{1},\!\cdots\!,x_{T}^{i-1}, x_{T}^{i+1},\!\cdots\!,x_{T}^{m}\}$.

Since the decision variable $x_{T}$ is constrained in the compact set $\mathcal{X}$, we have $\mathbb{E}_{x_{T}^{i}\sim \mathcal{A}_{\boldsymbol{\vartheta}_{0}}(\mathcal{P}),{x'}_{\hspace{-0.1cm}T}^{i}\sim \mathcal{A}_{\boldsymbol{\vartheta}_{0}}(\mathcal{P}')}[\|x_{T}^{i}-{x'}_{\hspace{-0.1cm}T}^{i}\|]\leq D_{\mathcal{X}}$ for all $i\!\in\![m]$ and some $D_{\mathcal{X}}\!>\!0$. Hence, \eqref{4t1} satisfies
\begin{equation}
\begin{aligned}
&f_{i}(x_{T}^{i},\phi(x_{T}^{i},x_{T}^{-i}))-f_{i}({x'}_{\hspace{-0.1cm}T}^{i},\phi({x'}_{\hspace{-0.1cm}T}^{i},x_{T}^{-i}))\\
&\leq (L_{f,1}+L_{f,2}L_{g})D_{\mathcal{X}}.\label{4t11}
\end{aligned}
\end{equation}
By using~\eqref{4t11}, we obtain
\begin{equation}
\begin{aligned}
&\mathbb{E}_{x_{T}\sim \mathcal{A}_{\boldsymbol{\vartheta}_{0}}(\mathcal{P})}[f_{i}(x_{T}^{i},\phi(x_{T}^{i},x_{T}^{-i}))]\\
&\leq \mathbb{E}_{{x'}_{\hspace{-0.1cm}T}^{i}\sim \mathcal{A}_{\boldsymbol{\vartheta}_{0}}(\mathcal{P}')^{i},x_{T}^{-i}\sim \mathcal{A}_{\boldsymbol{\vartheta}_{0}}(\mathcal{P})^{-i}}[f_{i}({x'}_{\hspace{-0.1cm}T}^{i},\phi({x'}_{\hspace{-0.1cm}T}^{i},x_{T}^{-i}))]\\
&\quad+(L_{f,1}+L_{f,2}L_{g})D_{\mathcal{X}},\label{4t2}
\end{aligned}
\end{equation}
where $\mathcal{A}_{\boldsymbol{\vartheta}_{0}}(\mathcal{P}')^{i}$ represents the $i$th output of $\mathcal{A}_{\boldsymbol{\vartheta}_{0}}(\mathcal{P}')$ and $\mathcal{A}_{\boldsymbol{\vartheta}_{0}}(\mathcal{P})^{-i}$ denotes the outputs of $\mathcal{A}_{\boldsymbol{\vartheta}_{0}}(\mathcal{P})$ with the $i$th output removed.

According to the definition of $\epsilon$-JDP and Lemma~\ref{JDP}, the first term on the right hand side of~\eqref{4t2} satisfies
\begin{equation}
\begin{aligned}
&\mathbb{E}_{{x'}_{\hspace{-0.1cm}T}^{i}\sim \mathcal{A}_{\boldsymbol{\vartheta}_{0}}(\mathcal{P}')^{i},x_{T}^{-i}\sim \mathcal{A}_{\boldsymbol{\vartheta}_{0}}(\mathcal{P})^{-i}}[f_{i}({x'}_{\hspace{-0.1cm}T}^{i},\phi({x'}_{\hspace{-0.1cm}T}^{i},x_{T}^{-i}))]\\
&\leq {\textstyle\int_{\mathbb{R}^{d}}f_{i}}({x'}_{\hspace{-0.1cm}T}^{i},\tau)\mathbb{P}(\phi({x'}_{\hspace{-0.1cm}T}^{i},x_{T}^{-i})\in d\tau)\\
&\leq e^{\epsilon}\mathbb{E}_{x'_{T}\sim \mathcal{A}_{\boldsymbol{\vartheta}_{0}}(\mathcal{P}')}[f_{i}({x'}_{\hspace{-0.1cm}T}^{i},\phi(x'_{T}))]\\
&\leq \mathbb{E}_{x'_{T}\sim \mathcal{A}_{\boldsymbol{\vartheta}_{0}}(\mathcal{P}')}[f_{i}({x'}_{\hspace{-0.1cm}T}^{i},\phi(x'_{T}))]+2\epsilon D_{f},\label{4t3}
\end{aligned}
\end{equation}
where in the last inequality we have used relations $e^{\epsilon}\leq 1+2\epsilon$ for $\epsilon\in(0,1)$ and $\max_{i\in[m]}\{\|f_{i}({x'}_{\hspace{-0.1cm}T}^{i},\phi(x'_{T}))\|\}\leq D_{f}$.

Substituting~\eqref{4t3} into~\eqref{4t2}, we arrive at~\eqref{T4result}.
\end{proof}

Theorem~\ref{T2} quantifies the maximal loss reduction that agent $i\in[m]$ can gain through its untruthful behaviors, which is bounded by $\eta\!=\!(L_{f,1}\!+\! L_{f,2}L_{g})D_{\mathcal{X}}\!+\!2\epsilon D_{f}$. The first term in $\eta$ is ``intrinsic" as it represents the effect of agent $i$'s changed decision variable on its own objective function $f_{i}$ (see Eq.~\eqref{4t11} for details) while the second term represents the benefit that agent $i$ can gain by sharing untruthful information (see Eq.~\eqref{4t3} for details). 

According to Eq.~\eqref{epsilon}, we have that $\epsilon$ is finite even in the infinite time horizon. This implies that the truthfulness parameter $\eta$ is also finite even in the infinite time horizon, which is significant because an agent can report untruthful data in each iteration. In contrast, existing JDP-based truthfulness results, e.g.,~\cite{truthfulEV,JDPoptimization1,JDPoptimization2}, have a truthfulness level (or strength) decreasing to zero as the number of iterations tends to infinity, implying that the truthfulness guarantee will eventually be lost.

Furthermore, we can reduce $\eta$ by decreasing $\epsilon$ to achieve a higher level of truthfulness. 
However, decreasing $\epsilon$ leads to an increase in noise parameters $\sigma_{\xi}$ and $\sigma_{\zeta}$ (see Remark~\ref{remark3}), which will compromise convergence accuracy. As such, a quantitative analysis of the tradeoff between the level of truthfulness and convergence performance warrants further exploration.
\begin{remark}
The convergence rate of $\eta$ is determined by the convergence rate of $\epsilon$ (see Eq.~\eqref{eta}). According to Eq.~\eqref{epsilon} in Lemma~\ref{JDP}, the convergence rate of $\epsilon$ is determined by decaying rates $u$, $w_{1}$, $w_{2}$, $\hat{\varsigma}_{\zeta}$, and $\hat{\varsigma}_{\xi}$. A faster decaying stepsize $\lambda_{t}$ (i.e., a larger $u$) and slower decaying Laplace-noise variances (i.e., smaller $\hat{\varsigma}_{\zeta}$ and $\hat{\varsigma}_{\xi}$) lead to faster convergence of both $\epsilon$ and $\eta$ as $T$ tends to infinity. In addition, 
although the number of iterations $T$ does not affect the convergence rate of $\eta$, it does affect the strength of truthfulness. A larger $T$ results in a larger $\epsilon$ under given values of $u$, $w_{1}$, $w_{2}$, $\hat{\varsigma}_{\zeta}$, and $\hat{\varsigma}_{\xi}$ (see Eq.~\eqref{epsilon}), which further leads to a larger $\eta$ (see Eq.~\eqref{eta}), implying a weaker truthfulness guarantee. When $T \to \infty$, $\epsilon$ converges to a finite constant, which ensures that $\eta$ also converges to a finite constant. This implies that our algorithm can guarantee truthfulness even in an infinite time horizon.
\end{remark}

\subsection{Tradeoff between the level of enabled truthfulness and convergence performance}\label{SectionIV3}
In this subsection, we quantify the tradeoff between convergence performance and the level of enabled truthfulness.
\begin{figure*}
\centering
\subfigure[The charging schedule computed by Algorithm~\ref{algorithm1}]{\label{Theorem1a}
\includegraphics[width=0.43\linewidth]{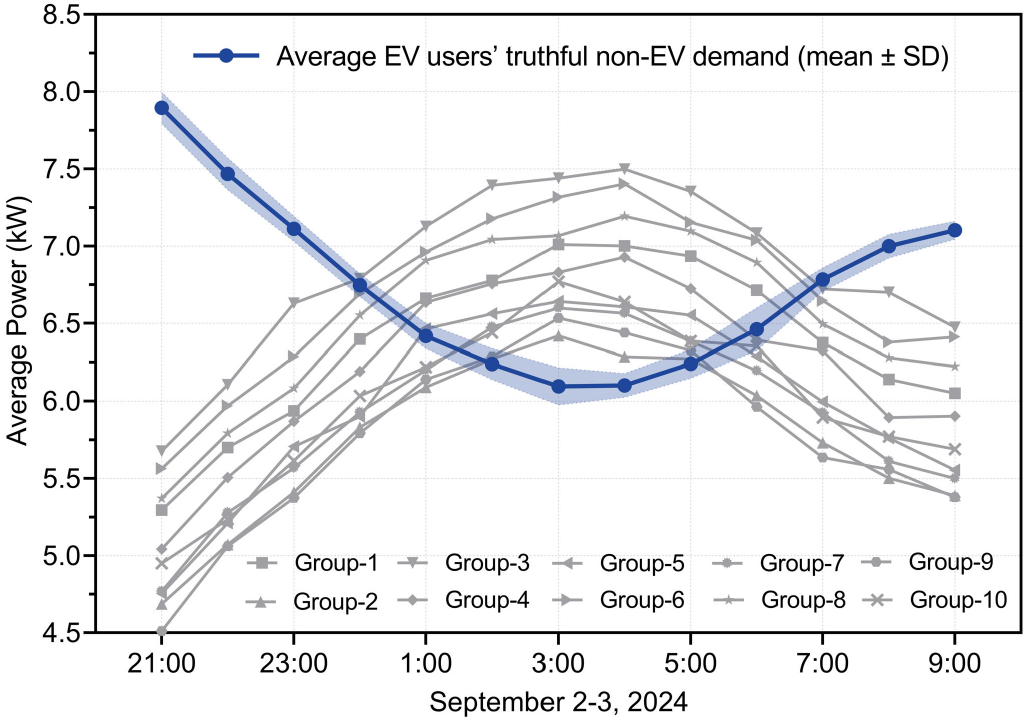}}
\quad
\subfigure[The evolution of global costs]{\label{Theorem1b}
\includegraphics[width=0.46\linewidth]{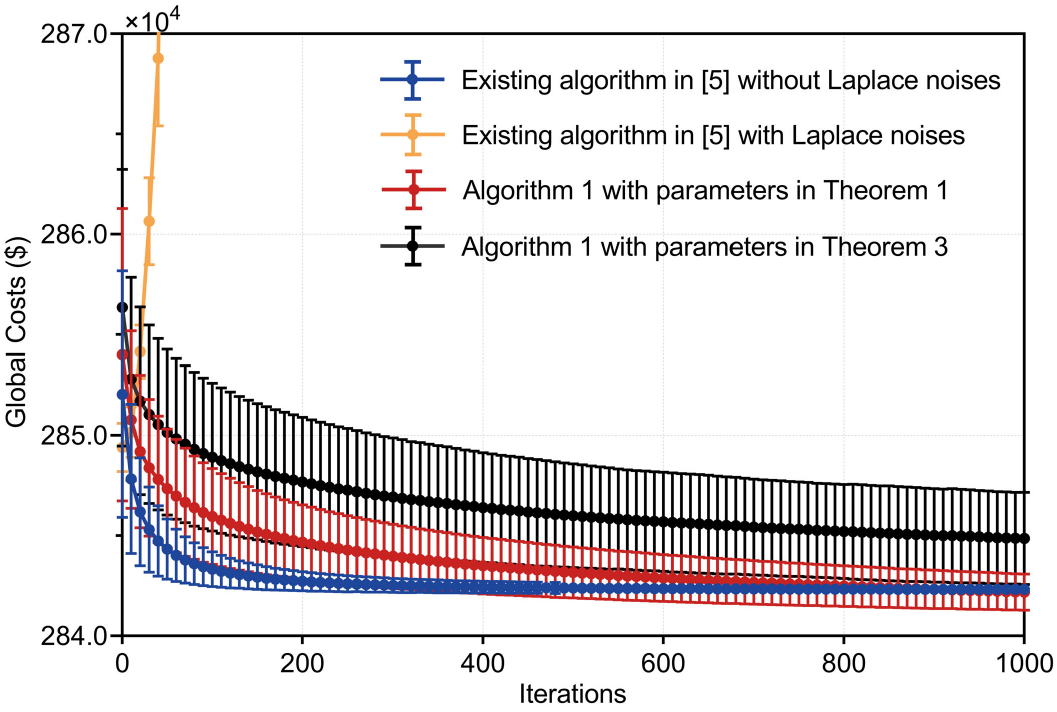}}
\caption{(a) The charging schedule computed by Algorithm~\ref{algorithm1} after $1,000$ iterations. In Fig.~\ref{Theorem1a}, the blue curve represents the average truthful non-EV demand of all EV users in the time interval $\mathcal{K}$ with the shaded area representing the variability (standard deviation-SD) among the ten groups, and the gray curve shows the average charging schedule of EVs in their respective group. (b) Comparison of global costs between the existing distributed aggregative optimization algorithm in~\cite{lixiuxian2021} and Algorithm~\ref{algorithm1} with different parameter settings.}
\label{Theorem1}
\end{figure*}
\begin{theorem}\label{T3}
Under the conditions in Theorem~\ref{T2},  Algorithm~\ref{algorithm1} is $\eta$-truthful. Furthermore, for any $T\in\mathbb{N}^{+}$, the following results always hold:

(i) If $F(x)$ is strongly convex, the rate of decaying sequence $\alpha_{t}$ satisfies $v>1$ and the rates of Laplace-noise variances satisfy $\varsigma_{\zeta}>\max\{0,\frac{1-u}{2}+w_{1}\}$ and $\varsigma_{\xi}>\max\{-\frac{w_{2}}{2},\frac{1}{2}-w_{2}\}$, then we have
\begin{equation}
\mathbb{E}[\|x_{T}-x^*\|^2]\leq \mathcal{O}\left(\frac{L_{g}^{4}L_{f}^2\bar{L}_{f}^2}{\mu|\delta_{2}|}+\frac{1}{\eta^2}\right).\label{3Tresult1}
\end{equation}

(ii) If $F(x)$ is general convex, the rate of decaying sequence $\alpha_{t}$ satisfies $v>1$ and the rates of Laplace-noise variances satisfy $\varsigma_{\zeta}>\max\{0,1-u+w_{1}\}$ and $\varsigma_{\xi}>\max\{-\frac{w_{2}}{2},\frac{1}{2}-w_{2}\}$, then we have
\begin{equation}
\frac{\textstyle\sum_{t=0}^{T}\lambda_{t}\mathbb{E}[F(x_{t})-F(x^{*})]}{\sum_{t=0}^{T}\lambda_{t}}\leq \mathcal{O}\left(\frac{L_{g}^{4}L_{f}^2\bar{L}_{f}^2}{|\delta_{2}|}+\frac{1}{\eta^2}\right).\label{3Tresult2}
\end{equation}

(iii) If $F(x)$ is nonconvex, the rate of decaying sequence $\alpha_{t}$ satisfies $v>\max\{1,u-w_{1}\}$ and the rates of Laplace-noise variances satisfy $\varsigma_{\zeta}>\max\{0,\frac{1-u}{2}+w_{1}\}$ and $\varsigma_{\xi}>\max\{-\frac{w_{2}}{2},\frac{1}{2}-w_{2}\}$, then we have
\begin{equation}
\frac{\sum_{t=0}^{T}\lambda_{t}\mathbb{E}[\|\nabla F(x_{t})\|^2]}{\sum_{t=0}^{T}\lambda_{t}}\leq \mathcal{O}\left(\frac{L_{g}^{4}L_{f}^3\bar{L}_{f}^3}{|\delta_{2}|}+\frac{1}{\eta^2}\right).\label{3Tresult3}
\end{equation}
\end{theorem}
\begin{proof}
See Appendix H.
\end{proof}
Theorem~\ref{T3} proves that our algorithm converges to a neighborhood of a solution to problem~\eqref{primal} to ensure that an agent's untruthful behavior can only gain $\eta$ in loss reduction (i.e., $\eta$-truthfulness). The size of this neighborhood is determined by a constant error $\mathcal{O}(1)$ and an error $\mathcal{O}(\frac{1}{\eta^2})$.  

The optimization error $\mathcal{O}(1)$ is due to the use of summable stepsizes, i.e., the decaying rate of stepsize $\lambda_{t}$ in Algorithm~\ref{algorithm1} satisfies $u>1$, which is key for us to ensure $\eta$-truthfulness in the infinite time horizon.  

The error $\mathcal{O}(\frac{1}{\eta^2})$ characterizes the tradeoff between the level of enabled truthfulness and optimization accuracy, which implies that achieving a higher level of truthfulness comes at the price of decreased convergence accuracy.

\section{Numerical Simulations}\label{SectionVI}
\subsection{Problem settings}
We evaluated the performance of our algorithm using a distributed EV charging problem~\cite{charging2011}. The goal of this problem is to determine a charging schedule for $m=10^{5}$ plug-in EVs that minimizes the global cost defined in~\eqref{primalEV}. We let the charging interval $\mathcal{K}$ cover the $21:00$ on one day to $9:00$ on the next day. The charging interval $\mathcal{K}$ was divided into $13$ time slots, with each of $1$-hour duration. We set the total generation capacity as $C_{\text{tot}}=1.2\times 10^{6}$ kW. We considered $10$ popular EV models, with each model's battery capacity $E^{i}$ and maximal charging rate $[x_{\max}^{i}]_{k}$ for $k=1,\cdots,13$ summarized in Table~\ref{table1}. Following the setup in~\cite{charging2011}, we divided $10^{5}$ EVs into $10$ groups, with each group containing $10\%$ of the EV population and consisting of vehicles from the same EV model. For each EV user in group $i\in[10]$, we configured its basic electricity demand profile in the time interval $\mathcal{K}$ for daily non-charging purposes as $d^{i}$ (referred to as the truthful non-EV demand of each EV user in group $i$), whose elements were independently sampled from a normal distribution with mean equal to the corresponding entry of $\frac{d}{m}$ and variance equal to $0.1$. Here, $d$ represents the load from the Midwest Independent System Operator (MISO) region on a typical summer day in 2024\footnote{Data on global non-EV demand $d$ are from \url{https://www.eia.gov/}.}. The average truthful non-EV demand of all EV users was depicted using the blue curve in Fig.~\ref{Theorem1a}, with the shaded area representing the variability (standard deviation-SD) among the ten groups. We defined the electricity price function as $p(r)\!=\!\text{col}(0.15[r]_{1}^{1.5},\cdots,0.15[r]_{13}^{1.5})$, where $[r]_{k}$ represents the $k$th element of $r$ with $k=1,2,\cdots,13$.  The EVs were connected in a random $4$-regular graph, where each EV user communicates with $4$ randomly assigned neighboring EV users. For the matrix $W$, we set $w_{ij}=0.2$ if EVs $i$
and $j$ are neighbors, and $w_{ij}=0$ otherwise.
\begin{table}[H]
\centering
\caption{Electric Vehicle Specifications\tnote{a}}
\label{table1}
\small
\begin{threeparttable}
\begin{tabular}{|l|c|c|}
\hline
\multirow{2}{*}{\textbf{Model}} & \textbf{Maximal} & \textbf{Battery} \\
& \textbf{charging rate} & \textbf{capacity} \\
\hline
Maserati GranCabrio Folgore & 22 kW & 83 kWh \\
\hline
Audi A6 Avant e-tron  & 11 kW & 75 kWh \\
\hline
Mercedes-Benz EQE 300 & 11 kW & 89 kWh\\
\hline
BMW i5 xDrive40 Sedan & 11 kW & 81 kWh\\
\hline
Kia EV3 Long Range & 11 kW & 78 kWh \\
\hline
Nissan Ariya & 7.4 kW & 87 kWh \\
\hline
Volkswagen ID.4 Pro & 11 kW & 77 kWh\\
\hline
BYD HAN & 11 kW & 85 kWh \\
\hline
Tesla Model Y Performance & 11 kW & 75 kWh \\
\hline
Hongqi E-HS9 84 kWh & 11 kW & 78 kWh \\
\hline
\end{tabular}
\begin{tablenotes}
\footnotesize
\item[a] Data on EV specifications are from~\url{https://ev-database.org/}.
\end{tablenotes}
\end{threeparttable}
\end{table}
\subsection{Evaluation results}
In each iteration, we injected noises $\zeta_{t}^{i}$ and $\xi_{t}^{i}$ into all shared variables $y_{t}^{i}$ and $\psi_{t}^{i}$, with each element of the noise vectors $\zeta_{t}^{i}$ and $\xi_{t}^{i}$ following Laplace distribution. The Laplace-noise variances, stepsizes, and decaying sequences in Algorithm~\ref{algorithm1} were set as $\sigma_{t,\zeta}^{i}\!=\!(t+1)^{-0.57}$, $\sigma_{t,\xi}^{i}\!=\!(t+1)^{-0.79}$, $\lambda_{t}\!=\!(t+1)^{-0.51}$, $\alpha_{t}\!=\!(t+1)^{-0.53}$, $\gamma_{t,1}\!=\!(t+1)^{-0.01}$, and $\gamma_{t,2}\!=\!(t+1)^{-0.01}$, respectively, which satisfy all conditions given in Theorem~\ref{T1}. In the simulation, we ran Algorithm~\ref{algorithm1} for $1,000$ iterations and calculated the average charging schedules of EVs in groups $i=1,\cdots,10$, respectively. The results are summarized in Fig.~\ref{Theorem1a}, which shows that the computed charging schedules for all EVs are valley-filling strategies.
\begin{figure*}
\centering
\subfigure[Truthful and untruthful non-EV demands]{\label{Theorem2a}
\includegraphics[width=0.43\linewidth]{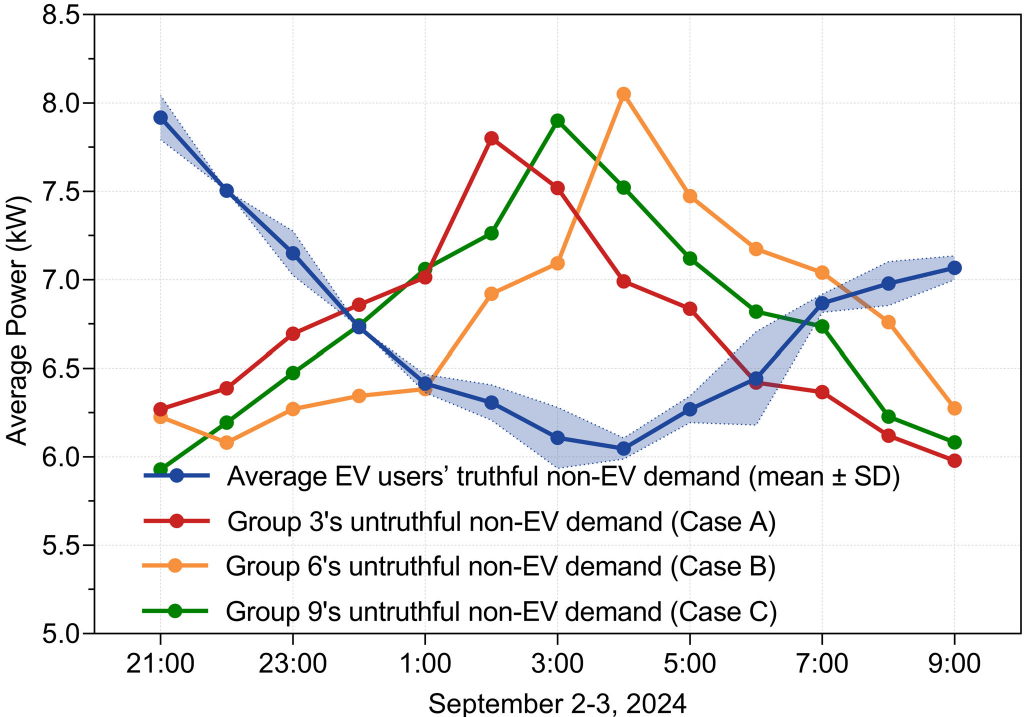}}
\quad
\subfigure[The predicted electricity prices under Cases A , B, and C]{\label{Theorem2b}
\includegraphics[width=0.44\linewidth]{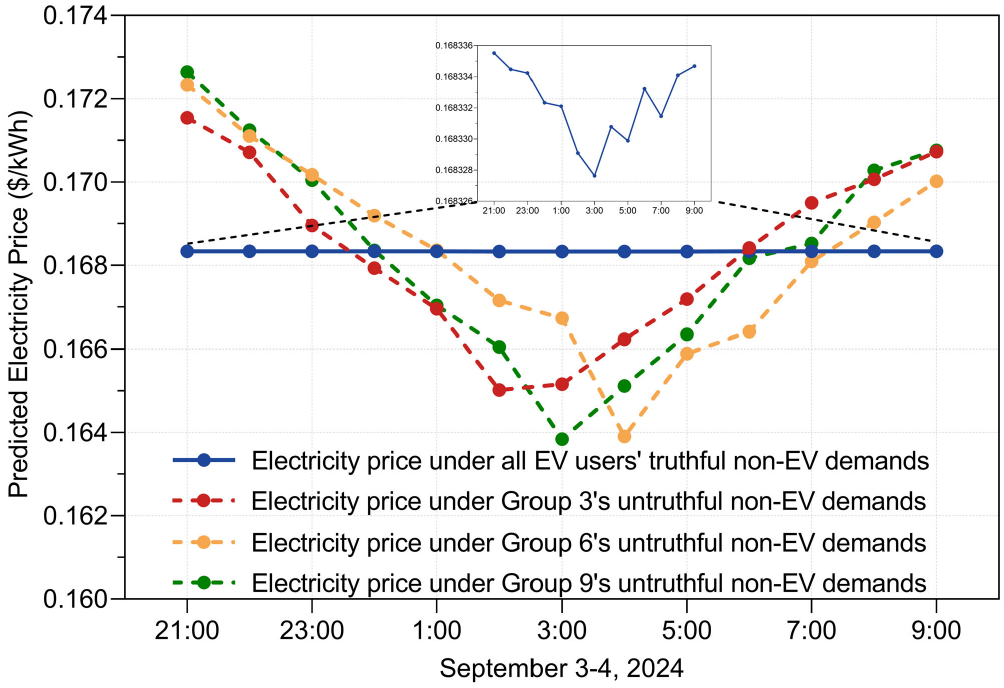}}
\subfigure[The costs of groups $i=3,6,9$ after $4000$ iterations, respectively]{\label{Theorem2c}
\includegraphics[width=0.435\linewidth]{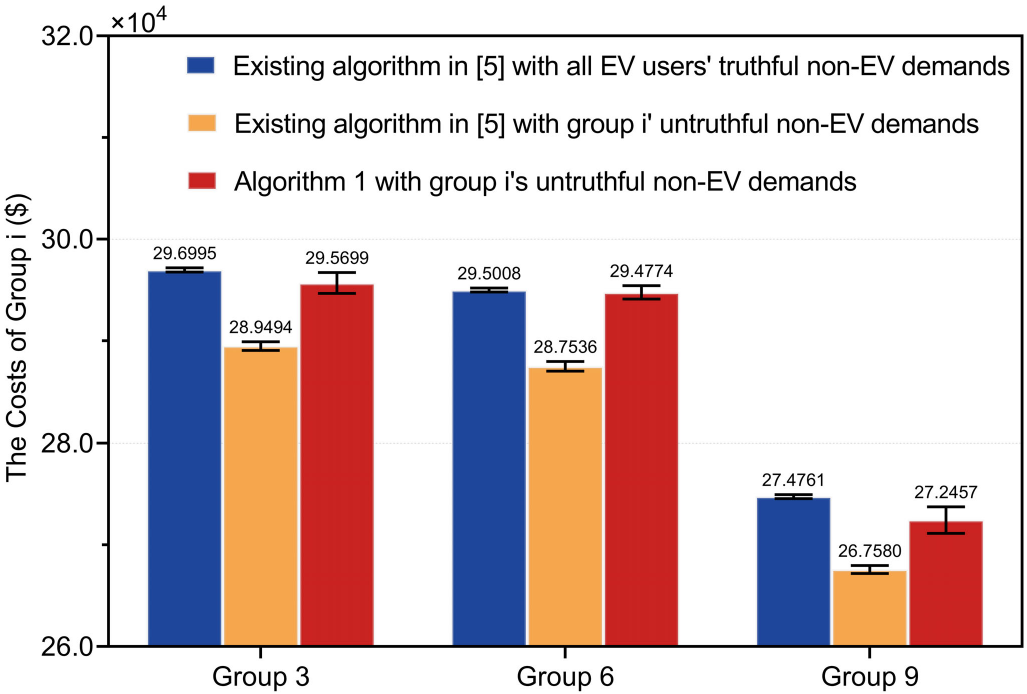}}
\quad
\subfigure[The global costs after $4,000$ iterations]{\label{Theorem2d}
\includegraphics[width=0.445\linewidth]{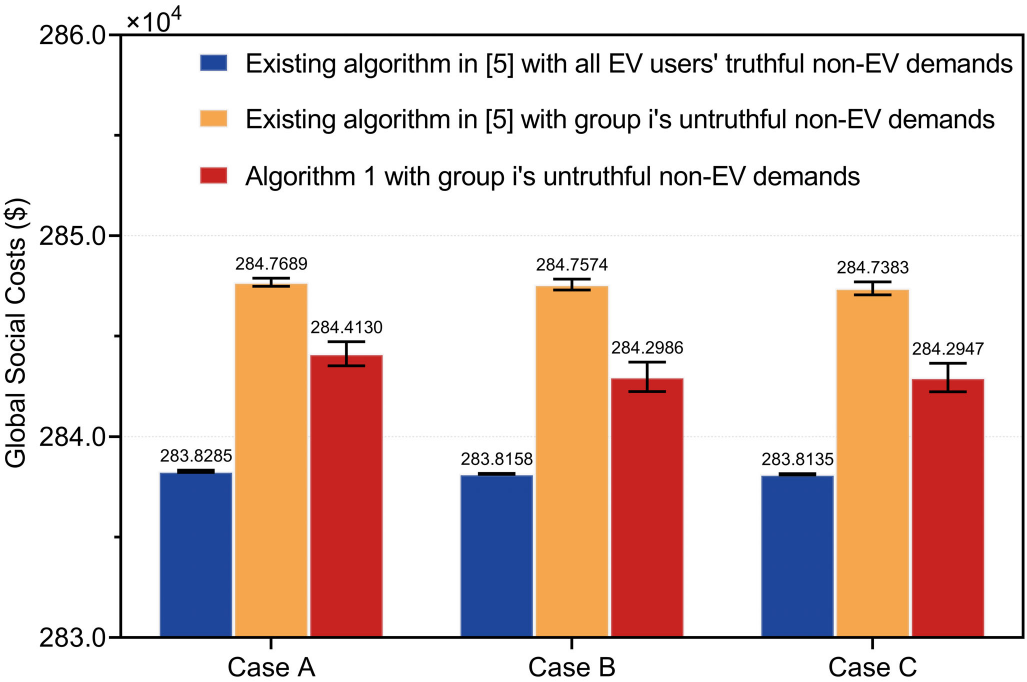}}
\caption{(a) The average truthful non-EV demand of EV users in groups $i=3,6,9$ (i.e., blue curve, with shaded areas representing the variability (standard deviation-SD) among the three groups), and the average untruthful non-EV demand of EV users in group $i$ with $i=3,6,9$, respectively. (b) The predicted electricity price under truthful non-EV demands of all EV users (i.e., blue curve) and the predicted electricity prices under untruthful non-EV demands of EV users in groups $i=3,6,9$, respectively. (c) Comparison of the total costs of EV users in group $i$ under their untruthful non-EV  demands with $i=3,6,9$. (d) Comparison of global costs under untruthful non-EV demands of EV users in groups $i=3,6,9$, respectively.}
\label{Theorem2}
\end{figure*} 

Furthermore, we set the Laplace-noise variances, stepsizes, and decaying sequences as $\sigma_{t,\zeta}^{i}\!=\!(t+1)^{-0.19}$, $\sigma_{t,\xi}^{i}\!=\!(t+1)^{-0.2}$, $\lambda_{t}\!=\!(t+1)^{-3.1}$, $\alpha_{t}\!=\!(t+1)^{-2}$, $\gamma_{t,1}\!=\!(t+1)^{-1.2}$, and $\gamma_{t,2}\!=\!(t+1)^{-0.4}$, respectively, which satisfy all conditions in Theorem~\ref{T3}. For comparison, we used the noise-free distributed aggregative optimization algorithm in~\cite{lixiuxian2021} as a baseline, with its stepsize setting to $\lambda=0.01$, in line with the guideline provided in~\cite{lixiuxian2021}. The comparison of global costs between the baseline and Algorithm~\ref{algorithm1} is summarized in Fig.~\ref{Theorem1b}. The results show that our algorithm with the parameter setting in Theorem~\ref{T1} slightly reduces the convergence rate compared with the existing algorithm in~\cite{lixiuxian2021} despite the presence of Laplace noises. The black curve in Fig.~\ref{Theorem1b} shows that convergence accuracy is compromised when we use summable stepsizes to ensure truthfulness, implying that achieving rigorous $\eta$-truthfulness in an infinite time horizon comes at the cost of reduced convergence accuracy (see Theorem~\ref{T3}). Nevertheless, we also ran the existing algorithm in~\cite{lixiuxian2021} under the same Laplace noises as ours for comparison. The orange curve in Fig.~\ref{Theorem1b} shows that the existing algorithm diverges in the presence of noises, which confirms the robustness of our algorithm to Laplace noises. 

To evaluate the effect of EV users' untruthful behaviors on their individual costs and the global cost, we ran the algorithm in~\cite{lixiuxian2021} and Algorithm~\ref{algorithm1} in three cases (called Case A, Case B, and Case C, respectively), where untruthful non-EV demands were shared by EV users in groups $i=3$, $i=6$, and $i=9$, respectively. The average truthful non-EV demand of EV users in groups $i=3,6,9$ and the average untruthful non-EV demand of EV users in these groups are shown in Fig.~\ref{Theorem2a}. The predicted electricity prices under the three cases after $4,000$ iterations are depicted in Fig.~\ref{Theorem2b}. We let the untruthful EV users in each group $i=3,6,9$ charge during the time interval when the predicted electricity price was lower than the baseline electricity price (blue curve in Fig.~\ref{Theorem2b}). The total costs of untruthful EV users in group $i=3,6,9$ are summarized in Fig.~\ref{Theorem2c}, respectively, which demonstrate that Algorithm~\ref{algorithm1} indeed restricts the cost reduction that untruthful EV users can gain through untruthful behaviors. Finally, the global costs under Cases A, B, and C after $4,000$ iterations are shown in Fig.~\ref{Theorem2d}, respectively, which confirm that our algorithm can effectively mitigate the increase in the global cost caused by untruthful behaviors of EV users.

\section{Conclusions}\label{SectionVII}
In this study, we have introduced a distributed aggregative optimization algorithm that can ensure convergence performance and truthful behaviors of participating agents. To the best of our knowledge, it is the first to successfully ensure truthfulness in a fully distributed setting without requiring the assistance of  any ``centralized" aggregators. This is in sharp contrast to all existing truthfulness approaches for cooperative optimization, which rely on a ``centralized" aggregator to collect private information/decision variables from participating agents. Moreover, we have systematically characterized the convergence rates of our algorithm under nonconvex/convex/strongly convex global objective functions despite the presence of Laplace noises. This extends existing results on distributed aggregative optimization that only consider strongly convex or convex objective functions. In addition, we have quantified the tradeoff between convergence accuracy and the level of enabled truthfulness under different convexity conditions. Simulation results using a coordinated EV charging problem on actual EV charging specifications and real load data confirm the efficiency of the proposed approach.

\section*{Appendix}
For the sake of notational simplicity, we denote
$\hat{y}_{t}^{i}=y_{t}^{i}-\bar{y}_{t}$,
$\hat{\psi}_{t}^{i}=\psi_{t}^{i}-\bar{\psi}_{t}$, $\xi_{t,w}^{i}=\sum_{j\in\mathcal{N}_{i}}w_{ij}\xi_{t}^{j}$, $\hat{\xi}_{t,w}^{i}=\xi_{t,w}^{i}-\bar{\xi}_{t,w}$, $s_{t,w}^{i}=\sum_{j\in\mathcal{N}_{i}}w_{ij}\boldsymbol{P}_{\Omega_{t}}(\zeta_{t}^{i})$,
$\hat{s}_{t,w}^{i}= s_{t,w}^{i}-\bar{s}_{t,w}$, $\sigma_{t,\zeta}=\max_{i\in[m]}\{\sigma_{t,\zeta}^{i}\}$, and $\sigma_{t,\xi}= \max_{i\in[m]}\{\sigma_{t,\xi}^{i}\}$. For any $x\!\in\!\mathbb{R}^{n}$ and $\psi\!\in\!\mathbb{R}^{md}$, we denote $f(x,\psi)\!=\!\sum_{i=1}^{m}f_{i}(x^{i},\psi^{i})\!\in\!\mathbb{R}$,  $\nabla_{1}f(x,\psi)\!=\!\text{col}(\nabla_{1}f_{1}(x^{1},\psi^{1}),\cdots,\nabla_{1}f_{m}(x^{m},\psi^{m}))\!\in\!\mathbb{R}^{n}$, $\nabla_{2}f(x,\psi)=\text{col}(\nabla_{2}f_{1}(x^{1},\psi^{1}),\cdots,\nabla_{2}f_{m}(x^{m},\psi^{m}))\in\mathbb{R}^{md}$, $g(x)=\text{col}(g_{1}(x^{1}),\cdots,g_{m}(x^{m}))\in\mathbb{R}^{md}$, and $\nabla g(x)=\text{diag}(\nabla g_{1}^{\top}(x^{1}),\cdots,\nabla g_{m}^{\top}(x^{m}))\in\mathbb{R}^{n\times md}$.

\subsection{Proof of Lemma~\ref{ybound}}\label{Asection1}
Assumption~\ref{A1}-(i) implies $\|y_{0}^{i}\|\leq L_{f,2}$ at initiation. Assuming $\|y_{t}^{i}\|\leq(1+\textstyle\sum_{p=0}^{t-1}\gamma_{p,1})L_{f,2}$ at the $t$th iteration, we next prove $\|y_{t+1}^{i}\|\leq (1+\sum_{p=0}^{t}\gamma_{p,1})L_{f,2}$ at iteration $t+1$. 

By using $w_{ii}\!=\!-\sum_{j\in{\mathcal{N}_{i}}}w_{ij}$, Line 4 in Algorithm~\ref{algorithm1} implies
\begin{equation}
\begin{aligned}
&\|y_{t+1}^{i}\|\leq\|y_{t}^{i}\!+\!{\textstyle\sum_{j\in{\mathcal{N}_{i}}}}w_{ij}\boldsymbol{P}_{\Omega_{t}}(y_{t}^{j}+\zeta_{t}^{j})\|+\|w_{ii}y_{t}^{i}\|\\
&\quad+\gamma_{t,1}\|\nabla_{2}f_{i}(x_{t}^{i},\psi_{t}^{i})\|\leq (1\!+\!{\textstyle\sum_{p=0}^{t-1}}\gamma_{p,1})L_{f,2}+\gamma_{t,1}L_{f,2}.\nonumber
\end{aligned}
\end{equation}

Using proof by induction, we obtain Lemma~\ref{ybound}.

\subsection{Proof of Lemma~\ref{FLips}}
By using \eqref{A2L1} and Assumption~\ref{A1}, we have
\begin{flalign}
&\|\nabla F(x_{1})-\nabla F(x_{2})\|\leq\bar{L}_{f,1}\|x_{1}-x_{2}\|\nonumber\\
&\quad+\bar{L}_{f,2}L_{g}\|x_{1}-x_{2}\|+\|\nabla g(x_{1})(\boldsymbol{1}_{m}\otimes\nabla_{2}\bar{f}(x_{1},\phi(x_{1})))\nonumber\\
&\quad-\nabla g(x_{2})(\boldsymbol{1}_{m}\otimes\nabla_{2}\bar{f}(x_{2},\phi(x_{2})))\|.\label{A2L2}
\end{flalign}

To characterize the third term on the right hand side of~\eqref{A2L2}, we introduce the following decomposition:
\begin{flalign}
&\|\nabla g(x_{1})(\boldsymbol{1}_{m}\otimes\nabla_{2}\bar{f}(x_{1},\phi(x_{1})))\nonumber\\
&\quad-\nabla g(x_{2})(\boldsymbol{1}_{m}\otimes\nabla_{2}\bar{f}(x_{2},\phi(x_{2})))\|\nonumber\\
&\leq \|\nabla g(x_{1})\left(\boldsymbol{1}_{m}\otimes\left(\nabla_{2}\bar{f}(x_{1},\phi(x_{1}))-\nabla_{2}\bar{f}(x_{2},\phi(x_{2}))\right)\right)\|\nonumber\\
&\quad+\|\left(\nabla g(x_{1})-\nabla g(x_{2})\right)(\boldsymbol{1}_{m}\otimes\nabla_{2}\bar{f}(x_{2},\phi(x_{2})))\|\nonumber\\
&\leq (L_{g}\bar{L}_{f,1}+\bar{L}_{f,2}L_{g}^2+\bar{L}_{g}L_{f})\|x_{1}-x_{2}\|.\label{A2L4}
\end{flalign}
Substituting~\eqref{A2L4} into \eqref{A2L2}, we prove Lemma~\ref{FLips}.

\subsection{Auxiliary lemmas}\label{Asection3}
In this subsection, we introduce some auxiliary lemmas that will be used in our subsequent convergence analysis. 

\begin{lemma}\label{AuxL1}
Denoting $\{\Phi_{t}\}$ as a nonnegative sequence, we have the following results:

(i) if there exist sequences $a_{t}=\frac{a_{0}}{(t+1)^{a}}$ and $b_{t}=\frac{b_{0}}{(t+1)^{b}}$ with some $a_{0}>b-a$, $b_{0}>0$, $1>a>0$, and $b>a$ such that $\Phi_{t+1}\leq (1-a_{t})\Phi_{t}+b_{t}$ holds, then we always have $\Phi_{t}\leq \frac{c_{\Phi}b_{t}}{a_{t}}$ with $c_{\Phi}=\frac{a_{0}}{b_{0}}\max\{\Phi_{0},\frac{b_{0}}{a_{0}-(b-a)}\}.$

(ii) if there exist sequences $a_{t}=\frac{a_{0}}{(t+1)^{a}}$ and $b_{t}=\frac{b_{0}}{(t+1)^{b}}$ with some $a_{0}>0$, $b_{0}>0$, $a>1$, and $b>1$ such that $\Phi_{t+1}\leq (1-a_{t})\Phi_{t}+b_{t}$ holds, then we always have $\Phi_{t}\leq \Phi_{0}e^{-\frac{1-(t+1)^{-(a-1)}}{a-1}}+b_{0}.$
\end{lemma}
\begin{proof}
(i) We first prove Lemma~\ref{AuxL1}-(i) using proof by induction. 

We denote $C_{0}=\max\{\Phi_{0},\frac{b_{0}}{a_{0}-(b-a)}\}$, which implies $\Phi_{0}\leq C_{0}$ at initiation. Assuming $\Phi_{t} \leq \frac{C_{0}}{(t+1)^{b-a}}$ at the $t$th iteration, we proceed to prove $\Phi_{t+1} \leq \frac{C_{0}}{(t+2)^{b-a}}$ at the $(t+1)$th iteration. 

By using the relation $\Phi_{t+1}\leq (1-a_{t})\Phi_{t}+b_{t}$, we have
\begin{equation}
\begin{aligned}
&\Phi_{t+1}\leq \textstyle\frac{C_{0}}{(t+1)^{b-a}}-\frac{C_{0}a_{0}}{(t+1)^{b}}+\frac{b_{0}}{(t+1)^{b}}\\
&\leq \textstyle\frac{C_{0}}{(t+2)^{b-a}}+\left(\frac{C_{0}}{(t+1)^{b-a}}-\frac{C_{0}}{(t+2)^{b-a}}-\frac{C_{0}a_{0}-b_{0}}{(t+1)^{b}}\right).\label{AL1}
\end{aligned}
\end{equation}
Using the mean value theorem, we have $\frac{C_{0}}{(t+1)^{b-a}}-\frac{C_{0}}{(t+2)^{b-a}}\leq \frac{C_{0}(b-a)}{(t+1)^{b-a+1}}<\frac{C_{0}(b-a)}{(t+1)^{b}}$. According to the definition of $C_{0}$, we obtain $C_{0}(b-a)\leq C_{0}a_{0}-b_{0}$, and hence, the inequality $\frac{C_{0}(b-a)}{(t+1)^{b}}-\frac{C_{0}a_{0}-b_{0}}{(t+1)^{b}}\leq 0$ holds for all $t\in\mathbb{N}$. Further using~\eqref{AL1}, we arrive at $\Phi_{t+1}\leq \frac{C_{0}}{(t+2)^{b-a}}$, which implies Lemma~\ref{AuxL1}-(i).

(ii) Iterating $\Phi_{t+1}\leq (1-a_{t})\Phi_{t}+b_{t}$ from $0$ to $t$, we have
\begin{flalign}
\Phi_{t+1}&\leq {\textstyle\prod_{p=0}^{t}}(1-a_{p})\Phi_{0}+{\textstyle\sum_{p=1}^{t}\prod_{q=p}^{t}}(1-a_{q})b_{p-1}+b_{t}\nonumber\\
&\leq \Phi_{0}e^{-\sum_{p=0}^{t}a_{p}}+{\textstyle\sum_{p=0}^{t}}b_{p},\label{AL3}
\end{flalign}
where we have used relations $\prod_{p=0}^{t}(1-a_{p})\leq e^{-\sum_{p=0}^{t}a_{p}}$ and $\prod_{q=p}^{t}(1-a_{q})<1$ in the last inequality.

According to the definition of $a_{t}$, we have that  for any $a\!>\!1$, $\sum_{p=0}^{t}\!\frac{a_{0}}{(t+1)^{a}}\geq\frac{1-(t+2)^{1-a}}{a-1}$ holds, which implies
\begin{equation}
\Phi_{0}e^{-\sum_{p=0}^{t}a_{p}}\leq \Phi_{0}e^{-\frac{1-(t+2)^{-(a-1)}}{a-1}}.\label{AL4}
\end{equation}

According to the definition of $b_{t}$, we have that for any $b>1$, the second term on the right hand side of~\eqref{AL3} satisfies 
\begin{equation}
\textstyle\sum_{p=0}^{t}b_{p}\leq b_{0}+\int_{1}^{t+1}\frac{b_{0}}{x^{b}}dx\leq b_{0}-\frac{b_{0}}{b-1}((t+1)^{1-b}+1).\label{AL5}
\end{equation}

Substituting~\eqref{AL4} and~\eqref{AL5} into~\eqref{AL3} and omitting the negative term $-\frac{b_{0}}{b-1}((t\!+\!1)^{1-b}\!+\!1)$, we obtain Lemma~\ref{AuxL1}-(ii).
\end{proof}

\begin{lemma}\label{xdiff}
Under Assumptions~\ref{A1} and \ref{A2}, the following inequality holds for Algorithm~\ref{algorithm1}:
\begin{equation}
\textstyle\mathbb{E}[\|x_{t+1}-x_{t}\|^2]\leq \frac{c_{x1}\lambda_{t}^2}{\gamma_{t,1}^2}\mathbb{E}[\|\hat{y}_{t+1}-\hat{y}_{t}\|^2]+{\textstyle\frac{c_{x2}\lambda_{t}^2\sigma_{t,\zeta}^2}{\gamma_{t,1}^2}}+c_{x3}\lambda_{t}^2,\label{xdiffresult}
\end{equation}
with $c_{x1}=4L_{g}^2$, $c_{x2}=4mL_{g}^2$, and $c_{x3}=2mL_{f}^2(2L_{g}^2+1)$.
\end{lemma}
\begin{proof}
Using the projection inequality and Line 5 in Algorithm~\ref{algorithm1}, we have
\begin{equation}
\textstyle\|x_{t+1}-x_{t}\|\leq \lambda_{t}\|\nabla_{1}f(x_{t},\psi_{t})+\nabla g(x_{t})\frac{y_{t+1}-y_{t}}{\gamma_{t,1}}\|.\nonumber
\end{equation}
Assumption~\ref{A1} implies
\begin{equation}
\mathbb{E}[\|x_{t+1}-x_{t}\|^2]\leq 2m\lambda_{t}^2L_{f}^2+{\textstyle\frac{2\lambda_{t}^2}{\gamma_{t,1}^2}}L_{g}^2\mathbb{E}[\|y_{t+1}-y_{t}\|^2].\label{X2}
\end{equation}

We proceed to characterize the last term of~\eqref{X2}. 

According to Line 4 in Algorithm~\ref{algorithm1}, we obtain
\begin{flalign}
&y_{t+1}^{i}=y_{t}^{i}+{\textstyle\sum_{j\in\mathcal{N}_{i}}}w_{ij}(\boldsymbol{P}_{\Omega_{t}}(y_{t}^{j}+\zeta_{t}^{j})-\boldsymbol{P}_{\Omega_{t}}(y_{t}^{i}+\zeta_{t}^{i}))\nonumber\\
&\quad+{\textstyle\sum_{j\in{\mathcal{N}_{i}}}}w_{ij}(\boldsymbol{P}_{\Omega_{t}}(y_{t}^{i}+\zeta_{t}^{i})-y_{t}^{i})+\gamma_{t,1}\nabla_{2}f_{i}(x_{t}^{i},\psi_{t}^{i}).\nonumber
\end{flalign}
By using the relation $\boldsymbol{1}_{m}^{\top}W=\boldsymbol{0}_{m}^{\top}$ and Lemma~\ref{ybound}, we have
\begin{equation}
\gamma_{t,1}\nabla_{2}\bar{f}(x_{t},\psi_{t})=\bar{y}_{t+1}-\bar{y}_{t}-\bar{s}_{t,w}.\label{X4}
\end{equation}

By using~\eqref{X4} and the definition $\hat{y}_{t}^{i}=y_{t}^{i}-\bar{y}_{t}$, we obtain
\begin{equation}
\begin{aligned}
&\mathbb{E}[\|y_{t+1}^{i}-y_{t}^{i}-\gamma_{t,1}\nabla_{2}\bar{f}(x_{t},\psi_{t})\|^2]\\
&=\mathbb{E}[\|\hat{y}_{t+1}^{i}-\hat{y}_{t}^{i}+\bar{s}_{t,w}\|^2]\leq\mathbb{E}[\|\hat{y}_{t+1}^{i}-\hat{y}_{t}^{i}\|^2]+\sigma_{t,\zeta}^2,\label{X5}
\end{aligned}
\end{equation}
where we have used the relation $\mathbb{E}[\|\bar{s}_{t,w}\|^2]\leq \sigma_{t,\zeta}^{2}$ in the last inequality based on Assumptions~\ref{A2} and the relation $\mathbb{E}[\|\zeta_{t}^{i}\|^2]=(\sigma_{t,\zeta}^{i})^2$ based on the ``Noise parameter setting".

According to~\eqref{X5} and Assumption~\ref{A1}-(i), we have
\begin{equation}
\begin{aligned}
&\mathbb{E}[\|y_{t+1}-y_{t}\|^2]\\
&\leq 2\mathbb{E}[\|y_{t+1}-y_{t}-\gamma_{t,1}(\boldsymbol{1}_{m}\otimes\nabla_{2}\bar{f}(x_{t},\psi_{t}))\|^2]\\
&\quad+2\mathbb{E}[\|\gamma_{t,1}(\boldsymbol{1}_{m}\otimes\nabla_{2}\bar{f}(x_{t},\psi_{t}))\|^2]\\
&\leq 2\mathbb{E}[\|\hat{y}_{t+1}-\hat{y}_{t}\|^2]+2m\sigma_{t,\zeta}^2+2mL_{f}^2\gamma_{t,1}^2.\label{X6}
\end{aligned}
\end{equation}

Substituting~\eqref{X6} into~\eqref{X2}, we arrive at~\eqref{xdiffresult}.
\end{proof}
\begin{lemma}\label{psig}
Under Assumptions~\ref{A1} and \ref{A2}, if $1>\alpha_{0}>2\varsigma_{\xi}+2w_{2}-v$, $1>v$, and $\varsigma_{\xi}>\frac{v}{2}-w_{2}$ hold, then we have
\begin{equation}
\mathbb{E}[\|\psi_{t}-\boldsymbol{1}_{m}\otimes\phi(x_{t})\|^2]\leq 2\mathbb{E}[\|\hat{\psi}_{t}\|^2]+{\textstyle\frac{c_{\bar{\psi}1}\gamma_{t,2}^2\sigma_{t,\xi}^{2}}{\alpha_{t}}},\label{psigresult}
\end{equation}
with $c_{\bar{\psi}1}=\frac{2m\alpha_{0}}{\alpha_{0}-(2w_{2}+2\varsigma_{\xi}-v)}$.
\end{lemma}
\begin{proof}
According to the definition of $\phi(x)$ in problem~\eqref{primal}, we have $\phi(x_{t})=\bar{g}(x_{t})$, which implies
\begin{equation}
\|\psi_{t}-\boldsymbol{1}_{m}\otimes\phi(x_{t})\|^2\leq 2\|\hat{\psi}_{t}\|^2+2m\|\bar{\psi}_{t}-\bar{g}(x_{t})\|^2.\label{Pg1}
\end{equation}

Line 7 in Algorithm~\ref{algorithm1} implies that the second term on the right hand side of~\eqref{Pg1} satisfies
\begin{equation}
\begin{aligned}
&\mathbb{E}[\|\bar{\psi}_{t+1}-\bar{g}(x_{t+1})\|^2]\\
&\leq(1-\alpha_{t})\mathbb{E}[\|\bar{\psi}_{t}-\bar{g}(x_{t})\|^2]+\gamma_{t,2}^{2}\sigma_{t,\xi}^{2},\label{Pg3}
\end{aligned}
\end{equation}

Combining Lemma~\ref{AuxL1}-(i) with~\eqref{Pg3}, we obtain
\begin{equation}
2m\mathbb{E}[\|\bar{\psi}_{t}-\bar{g}(x_{t})\|^2]
\leq {\textstyle\frac{c_{\bar{\psi}1}\gamma_{t,2}^2\sigma_{t,\xi}^{2}}{\alpha_{t}}},\label{Pg4}
\end{equation}
with $c_{\bar{\psi}1}=\frac{2m\alpha_{0}}{\alpha_{0}-(2w_{2}+2\varsigma_{\xi}-v)}$.

Substituting~\eqref{Pg4} into~\eqref{Pg1}, we arrive at~\eqref{psigresult}.
\end{proof}
\begin{lemma}\label{hatydiff}
Under Assumptions~\ref{A1} and~\ref{A2}, if $1>\alpha_{0}>2\varsigma_{\xi}+2w_{2}-v$, $1>v$, and $\varsigma_{\xi}>\frac{v}{2}-w_{2}$ hold, then the following inequality holds for Algorithm~\ref{algorithm1}:
\begin{flalign}
&\mathbb{E}[\|\hat{y}_{t+1}-\hat{y}_{t}\|^2]\leq (1-|\delta_{2}|+c_{y1}\lambda_{t-1}^2)\mathbb{E}[\|\hat{y}_{t}-\hat{y}_{t-1}\|^2]\nonumber\\
&\quad+(c_{y2}\gamma_{t-1,2}^2+c_{y3}\alpha_{t-1}^2)\gamma_{t-1,1}^2\mathbb{E}[\|\hat{\psi}_{t-1}\|^2]\!+\!\Psi_{t,y},\label{hatydiffresult}
\end{flalign}
where the constants $c_{y1}$ to $c_{y3}$ are given by $c_{y1}=24(1+\frac{1}{|\delta_{2}|})\bar{L}_{f}^2(3L_{g}^4+L_{g}^2)$, $c_{y2}=18|\delta_{2}|(|\delta_{2}|+1)\bar{L}_{f}^2$, and $c_{y3}=72(1+\frac{1}{|\delta_{2}|})\bar{L}_{f}^2$, respectively, and the term $\Psi_{t,y}$ is given by
\begin{equation}
\begin{aligned}
&\Psi_{t,y}=c_{y4}\sigma_{t,\zeta}^2+c_{y5}\sigma_{t-1,\zeta}^2+c_{y6}\|\gamma_{t,1}-\gamma_{t-1,1}\|^2\\
&\quad+c_{y7}\gamma_{t-1,1}^2\lambda_{t-1}^2\!+\!c_{y8}\gamma_{t-1,1}^2\gamma_{t-1,2}^2\sigma_{t-1,\xi}^2\!+\!c_{y9}\gamma_{t-1,1}^2\alpha_{t-1}^2,\nonumber
\end{aligned}
\end{equation}
with $c_{y4}=4m$, $c_{y5}=4m+24m\lambda_{0}^2(1+\frac{1}{|\delta_{2}|})\bar{L}_{f}^2(3L_{g}^4+L_{g}^2)$, $c_{y6}=12m(1+\frac{1}{|\delta_{2}|})L_{f}^2$, $c_{y7}=12m(1+\frac{1}{|\delta_{2}|})\bar{L}_{f}^2L_{f}^2(6L_{g}^4+5L_{g}^2+1)$, $c_{y8}=6(1+\frac{12m\alpha_{0}^2}{\alpha_{0}-(2w_{2}+2\varsigma_{\xi}-v)})(1+\frac{1}{|\delta_{2}|})\bar{L}_{f}^2$, and $c_{y9}=144mD_{g}^2(1+\frac{1}{|\delta_{2}|})\bar{L}_{f}^2$.
\end{lemma}
\begin{proof}
Using Line 4 in Algorithm~\ref{algorithm1}, we have
\begin{equation}
\begin{aligned}
\hat{y}_{t+1}^{i}&=\hat{y}_{t}^{i}+{\textstyle\sum_{j\in\mathcal{N}_{i}}}w_{ij}(\hat{y}_{t}^{j}-\hat{y}_{t}^{i})-\bar{s}_{t,w}\\
&\quad+{\textstyle\sum_{j\in{\mathcal{N}_{i}}}}w_{ij}(\boldsymbol{P}_{\Omega_{t}}(y_{t}^{j}+\zeta_{t}^{j})-y_{t}^{j})\\
&\quad+\gamma_{t,1}\left(\nabla_{2}f_{i}(x_{t}^{i},\psi_{t}^{i})-\nabla_{2}\bar{f}(x_{t},\psi_{t})\right).\label{Hy1}
\end{aligned}
\end{equation}

By using~\eqref{Hy1}, we can obtain the following relation:
\begin{equation}
\begin{aligned}
&\hat{y}_{t+1}^{i}-\hat{y}_{t}^{i}= \hat{y}_{t}^{i}-\hat{y}_{t-1}^{i}+{\textstyle\sum_{j\in\mathcal{N}_{i}}}w_{ij}\big((\hat{y}_{t}^{j}-\hat{y}_{t-1}^{j})\\
&\quad-(\hat{y}_{t}^{i}-\hat{y}_{t-1}^{i})\big)-(\bar{s}_{t,w}-\bar{s}_{t-1,w})+\Xi_{t}^{i}\\
&\quad+(\gamma_{t,1}-\gamma_{t-1,1})(\nabla_{2}f_{i}(x_{t}^{i},\psi_{t}^{i})-\nabla_{2}\bar{f}(x_{t},\psi_{t}))\\
&\quad+\gamma_{t-1,1}(\nabla_{2}f_{i}(x_{t}^{i},\psi_{t}^{i})-\nabla_{2}f_{i}(x_{t-1}^{i},\psi_{t-1}^{i}))\\
&\quad-\gamma_{t-1,1}(\nabla_{2}\bar{f}(x_{t},\psi_{t})-\nabla_{2}\bar{f}(x_{t-1},\psi_{t-1})),\label{Hy2}
\end{aligned}
\end{equation}
where the term $\Xi_{t}^{i}$ is given by $\Xi_{t}^{i}={\textstyle\sum_{j\in{\mathcal{N}_{i}}}}w_{ij}(\boldsymbol{P}_{\Omega_{t}}(y_{t}^{j}+\zeta_{t}^{j})-y_{t}^{j})-{\textstyle\sum_{j\in{\mathcal{N}_{i}}}}w_{ij}(\boldsymbol{P}_{\Omega_{t-1}}(y_{t-1}^{j}+\zeta_{t-1}^{j})-y_{t-1}^{j})$.

We write~\eqref{Hy2} in a compact form:
\begin{equation}
\begin{aligned}
&\hat{y}_{t+1}-\hat{y}_{t}\\
&=(I_{md}+(W\otimes I_{d}))(\hat{y}_{t}-\hat{y}_{t-1})\\
&\quad-\boldsymbol{1}_{m}\otimes(\bar{s}_{t,w}-\bar{s}_{t-1,w})+\Xi_{t}\\
&\quad+(\gamma_{t,1}-\gamma_{t-1,1})(\nabla_{2}f(x_{t},\psi_{t})-\boldsymbol{1}_{m}\otimes\nabla_{2}\bar{f}(x_{t},\psi_{t}))\\
&\quad+\gamma_{t-1,1}(\nabla_{2}f(x_{t},\psi_{t})-\nabla_{2}f(x_{t-1},\psi_{t-1}))\\
&\quad-\gamma_{t-1,1}(\boldsymbol{1}_{m}\otimes\nabla_{2}\bar{f}(x_{t},\psi_{t})-\boldsymbol{1}_{m}\otimes\nabla_{2}\bar{f}(x_{t-1},\psi_{t-1})).\nonumber
\end{aligned}
\end{equation}
Using the Young's inequality and Assumption~\ref{A2}, we have
\begin{flalign}
&\mathbb{E}[\|\hat{y}_{t+1}-\hat{y}_{t}\|^2]\nonumber\\
&\leq 2m\mathbb{E}[\|\bar{s}_{t,w}-\bar{s}_{t-1,w}\|^2]+2\mathbb{E}[\|\Xi_{t}\|^2]\nonumber\\
&\quad+(1+|\delta_{2}|)(1-|\delta_{2}|)^2\mathbb{E}[\|\hat{y}_{t}-\hat{y}_{t-1}\|^2]\nonumber\\
&\quad+{\textstyle\left(1+\frac{1}{|\delta_{2}|}\right)}\left(12mL_{f}^2\|\gamma_{t,1}-\gamma_{t-1,1}\|^2\right.\nonumber\\
&\left.\quad+6\gamma_{t-1,1}^2\bar{L}_{f}^2(\mathbb{E}[\|x_{t}\!-\!x_{t-1}\|^2]+\mathbb{E}[\|\psi_{t}\!-\!\psi_{t-1}\|^2])\right).\label{Hy4}
\end{flalign}

The noise variance $\mathbb{E}[\|\zeta_{t}^{i}\|^2]=(\sigma_{t,\zeta}^{i})^2$
implies that the first and second terms on the right hand side of~\eqref{Hy4} satisfy
\begin{equation}
m\mathbb{E}[\|\bar{s}_{t,w}-\bar{s}_{t-1,w}\|^2]+\mathbb{E}[\|\Xi_{t}\|^2]\leq 2m(\sigma_{t,\zeta}^2+\sigma_{t-1,\zeta}^2).\label{Hy5}
\end{equation}

According to Line 7 in Algorithm~\ref{algorithm1}, the last term on the right hand side of~\eqref{Hy4} satisfies
\begin{flalign}
&\mathbb{E}[\|\psi_{t}-\psi_{t-1}\|^2]\leq \gamma_{t-1,2}^2\sigma_{t-1,\xi}^2\nonumber\\
&\quad+3\gamma_{t-1,2}^2\mathbb{E}[\|(W\otimes I_{d})\hat{\psi}_{t-1}\|^2]+3L_{g}^2\mathbb{E}[\|x_{t}-x_{t-1}\|^2]\nonumber\\
&\quad +6\alpha_{t-1}^2\mathbb{E}[\|\boldsymbol{1}_{m}\otimes \bar{g}(x_{t-1})\!-\!\psi_{t-1}\|^2]\!+\!24mD_{g}^2\alpha_{t-1}^2,\label{Hy7}
\end{flalign}
where  in the derivation we have used $\mathbb{E}[\|g(x_{t-1})-\psi_{t-1}\|^2]\leq 2\mathbb{E}[\|\boldsymbol{1}_{m}\otimes \bar{g}(x_{t-1})-\psi_{t-1}\|^2]+2\mathbb{E}[\|\boldsymbol{1}_{m}\otimes \bar{g}(x_{t-1})-g(x_{t-1})\|^2]$. Moreover, since $x_{t}^{i}$ is always constrained in a compact set $\mathcal{X}_{i}$, we have $\|g_{i}(x_{t}^{i})\|\leq D_{g}$, which implies
$2\mathbb{E}[\|\boldsymbol{1}_{m}\otimes \bar{g}(x_{t-1})-g(x_{t-1})\|^2]\leq 8mD_{g}^2$.

Substituting~\eqref{xdiffresult} in Lemma~\ref{xdiff},~\eqref{psigresult} in Lemma~\ref{psig}, inequality~\eqref{Hy5}, and inequality~\eqref{Hy7} into~\eqref{Hy4} leads to~\eqref{hatydiffresult}.
\end{proof}

\begin{lemma}\label{hatpsidiff}
Under Assumptions~\ref{A1} and~\ref{A2}, the following inequality holds for Algorithm~\ref{algorithm1}:
\begin{equation}
\begin{aligned}
\mathbb{E}[\|\hat{\psi}_{t+1}\|^2]&\leq (1-\gamma_{t,2}|\delta_{2}|)\mathbb{E}[\|\hat{\psi}_{t}\|^2]\\
&\quad+{\textstyle\frac{c_{\psi1}\lambda_{t}^2}{\gamma_{t,2}\gamma_{t,1}^2}}\mathbb{E}[\|\hat{y}_{t+1}-\hat{y}_{t}\|^2]+\Psi_{t,\psi},\label{hatpsidiffresult}
\end{aligned}
\end{equation}
where the constant $c_{\psi1}$ is given by $c_{\psi1}=32(\gamma_{2}+\frac{1}{|\delta_{2}|})L_{g}^4$ and the term $\Psi_{t,\psi}$ is given by
\begin{flalign}
\Psi_{t,\psi}=\textstyle\frac{c_{\psi2}\lambda_{t}^2\sigma_{t,\zeta}^2}{\gamma_{t,2}\gamma_{t,1}^2}+\frac{c_{\psi3}\lambda_{t}^2}{\gamma_{t,2}}+\frac{c_{\psi4}\alpha_{t}^{2}}{\gamma_{t,2}}+c_{\psi5}\gamma_{t,2}^{2}\sigma_{t,\xi}^{2},\nonumber
\end{flalign}
with $c_{\psi2}=32(\gamma_{2}+\frac{1}{|\delta_{2}|})$, $c_{\psi3}=16m(\gamma_{2}+\frac{1}{|\delta_{2}|})L_{f}^2(2L_{g}^4+L_{g}^2)$, $c_{\psi4}=8mD_{g}^2(\gamma_{2}+\frac{1}{|\delta_{2}|})$, and $c_{\psi5}=4m$.
\end{lemma}
\begin{proof}
By using Line 7 in Algorithm~\ref{algorithm1} and the Young's inequality, for any $\epsilon>0$, we have
\begin{flalign}
&\mathbb{E}[\|\hat{\psi}_{t+1}\|^2]\leq\gamma_{t,2}^{2}\mathbb{E}[\|\hat{\xi}_{t,w}\|^2]\nonumber\\
&\quad+(1+\varepsilon)\|(1-\alpha_{t})\otimes I_{md}+\gamma_{t,2}(W\otimes I_{d})\|^2\mathbb{E}[\|\hat{\psi}_{t}\|^2]\nonumber\\
&\quad+2\textstyle{\left(1+\frac{1}{\varepsilon}\right)}\left(\mathbb{E}[\|g(x_{t+1})-(1-\alpha_{t})g(x_{t})\|^2]\nonumber\right.\\
&\left.\quad+\mathbb{E}[\|\boldsymbol{1}_{m}\otimes\bar{g}(x_{t+1})-(1-\alpha_{t})\boldsymbol{1}_{m}\otimes\bar{g}(x_{t})\|^2]\right).\label{Hp3}
\end{flalign}

We characterize each item on the right hand side of~\eqref{Hp3}. 

(i) By using Assumptions~\ref{A2} and the relationship $\mathbb{E}[\|\xi_{t}^{i}\|^2]=(\sigma_{t,\xi}^{i})^2$ , we have
\begin{equation}
\mathbb{E}[\|\hat{\xi}_{t,w}\|^2]=\mathbb{E}[\|\xi_{t,w}-\boldsymbol{1}_{m}\otimes\bar{\xi}_{t,w}\|^2]\leq 4m\sigma_{t,\xi}^{2}.\label{Hp4}
\end{equation}

(ii) Letting $\varepsilon=\gamma_{t,2}|\delta_{2}|$, we obtain
\begin{equation}
\begin{aligned}
&(1+\varepsilon)\|(1-\alpha_{t})\otimes I_{md}+\gamma_{t,2}(W\otimes I_{d})\|^2\mathbb{E}[\|\hat{\psi}_{t}\|^2]\\
&\leq (1-\gamma_{t,2}|\delta_{2}|)\mathbb{E}[\|\hat{\psi}_{t}\|^2],\label{Hp5}
\end{aligned}
\end{equation}
where in the derivation we have used the relationships $1-\alpha_{t}-\gamma_{t,2}|\delta_{2}|<1-\gamma_{t,2}|\delta_{2}|$ and $(1+\gamma_{t,2}|\delta_{2}|)(1-\gamma_{t,2}|\delta_{2}|)<1$.

(iii) Assumption~\ref{A1}-(ii) implies
\begin{equation}
\begin{aligned}
&\mathbb{E}[\|g(x_{t+1})-(1-\alpha_{t})g(x_{t})\|^2]\\
&\leq 2L_{g}^2\mathbb{E}[\|x_{t+1}-x_{t}\|^2]+2m\alpha_{t}^{2}D_{g}^2.\label{Hp6}
\end{aligned}
\end{equation}

Substituting~\eqref{xdiffresult} in Lemma~\ref{xdiff} and inequalities~\eqref{Hp4},~\eqref{Hp5}, and~\eqref{Hp6} into~\eqref{Hp3}, we arrive at~\eqref{hatpsidiffresult}.
\end{proof}
\begin{lemma}\label{rateyzong}
Under Assumptions~\ref{A1} and~\ref{A2}, if $1>u>w_{2}$, $1>v>w_{2}$, $\varsigma_{\zeta}>\max\{w_{1},\frac{w_{2}}{2}\}$, and $\varsigma_{\xi}>\frac{v}{2}-w_{2}$ hold, then we have
\begin{flalign}
\textstyle\left(1-\frac{|\delta_{2}|}{4}\right)\mathbb{E}[\|\hat{y}_{t+1}-\hat{y}_{t}\|^2]+\mathbb{E}[\|\hat{\psi}_{t}\|^2]\leq\frac{\kappa_{1}}{(t+1)^{\beta_{1}}},\label{rateyzongresult}
\end{flalign}
where the rate $\beta_{1}$ is given by $\beta_{1}=\min\{2\varsigma_{\zeta}-w_{2},2\varsigma_{\xi}+w_{2},2u-2w_{2},2v-2w_{2}\}$ and the constant $\kappa_{1}$ is given by $\kappa_{1}=\max\{(1-\frac{|\delta_{2}|}{4})\mathbb{E}[\|\hat{y}_{1}\|^2]+\mathbb{E}[\|\hat{\psi}_{1}\|^2],\frac{c_{s1}}{4^{-1}|\delta_{2}|\gamma_{2}-\beta_{1}}\}$ with $c_{s1}$ given in~\eqref{Rs5}.
\end{lemma}
\begin{proof}
Summing both sides of~\eqref{hatydiffresult} in Lemma~\ref{hatydiff} and~\eqref{hatpsidiffresult} in Lemma~\ref{hatpsidiff}, we obtain
\begin{flalign}
&\left(1-\textstyle\frac{c_{\psi1}\lambda_{t}^2}{\gamma_{t,2}\gamma_{t,1}^2}\right)\mathbb{E}[\|\hat{y}_{t+1}-\hat{y}_{t}\|^2]+\mathbb{E}[\|\hat{\psi}_{t+1}\|^2]+\mathbb{E}[\|\hat{\psi}_{t}\|^2]\nonumber\\
&\leq  \textstyle\left(1-|\delta_{2}|+c_{y1}\lambda_{t-1}^2+\frac{c_{\psi1}\lambda_{t-1}^2}{\gamma_{t-1,2}\gamma_{t-1,1}^2}\right)\mathbb{E}[\|\hat{y}_{t}-\hat{y}_{t-1}\|^2]\nonumber\\
&\quad+(1-\gamma_{t,2}|\delta_{2}|)\mathbb{E}[\|\hat{\psi}_{t}\|^2]+\big(1-\gamma_{t-1,2}|\delta_{2}|\nonumber\\
&\quad+(c_{y2}\gamma_{t-1,2}^2+c_{y3}\alpha_{t-1}^2)\gamma_{t-1,1}^2\big)\mathbb{E}[\|\hat{\psi}_{t-1}\|^2]\nonumber\\
&\quad+\Psi_{t,y}+\Psi_{t,\psi}+\Psi_{t-1,\psi},\label{Rs1}
\end{flalign}
where $\Psi_{t,y}$ and $\Psi_{t,\psi}$ are given in~\eqref{hatydiffresult} and~\eqref{hatpsidiffresult}, respectively.

We let the initial values of decaying sequences satisfy $\textstyle\frac{|\delta_{2}|}{4}\geq \frac{c_{\psi1}\lambda_{0}^2}{\gamma_{2}\gamma_{1}^2}$, $\frac{|\delta_{2}|}{2}\geq c_{y1}\lambda_{0}^2+\frac{c_{\psi1}\lambda_{0}^2}{\gamma_{2}\gamma_{1}^2}$, $\gamma_{2}<1$, and $\frac{|\delta_{2}|\gamma_{2}}{2}\geq(c_{y2}\gamma_{2}^2+c_{y3}\alpha_{0}^2)\gamma_{1}^2$.  

Then, \eqref{Rs1} can be rewritten as follows:
\begin{flalign}
&\textstyle\left(1-\frac{|\delta_{2}|}{4}\right)\mathbb{E}[\|\hat{y}_{t+1}-\hat{y}_{t}\|^2]+\mathbb{E}[\|\hat{\psi}_{t+1}\|^2]+\mathbb{E}[\|\hat{\psi}_{t}\|^2]\nonumber\\
&\leq \textstyle\left(1-\frac{\gamma_{t,2}|\delta_{2}|}{4}\right)\left[\left(1-\frac{|\delta_{2}|}{4}\right)\mathbb{E}[\|\hat{y}_{t}-\hat{y}_{t-1}\|^2]\nonumber\right.\\
&\left.\quad+\mathbb{E}[\|\hat{\psi}_{t}\|^2]+\mathbb{E}[\|\hat{\psi}_{t-1}\|^2]\right]+\Psi_{t,1},\label{Rs4}
\end{flalign}
where the term $\Psi_{t,1}$ is given by
\begin{flalign}
\Psi_{t,1}&\textstyle=\frac{(c_{y4}+2^{2\varsigma_{\zeta}}c_{y5})\sigma_{\zeta}^2}{(t+1)^{2\varsigma_{\zeta}}}+\frac{(1+2^{2w_{1}+2})w_{1}^2c_{y6}\gamma_{1}^2}{(t+1)^{2w_{1}+2}}+\frac{2^{2w_{1}+2u}c_{y7}\gamma_{1}^2\lambda_{0}^2}{(t+1)^{2w_{1}+2u}}\nonumber\\
&\textstyle+\frac{2^{2w_{1}+2w_{2}+2\varsigma_{\xi}}c_{y8}\gamma_{1}^2\gamma_{2}^2\sigma_{\xi}^2}{(t+1)^{2w_{1}+2w_{2}+2\varsigma_{\xi}}}\!+\!\frac{2^{2v+2w_{1}}c_{y9}\gamma_{1}^2\alpha_{0}^2}{(t+1)^{2v+2w_{1}}}\!+\!\frac{(1+2^{2u-w_{2}})c_{\psi3}\lambda_{0}^2}{\gamma_{2}(t+1)^{2u-w_{2}}}\nonumber\\
&\textstyle+\frac{(1+2^{2u+2\varsigma_{\zeta}-2w_{1}-w_{2}})c_{\psi2}\lambda_{0}^2\sigma_{\zeta}^2}{\gamma_{2}\gamma_{1}^2(t+1)^{2u+2\varsigma_{\zeta}-2w_{1}-w_{2}}}+\frac{(1+2^{2v-w_{2}})c_{\psi4}\alpha_{0}^{2}}{\gamma_{2}(t+1)^{2v-w_{2}}}\nonumber\\
&\textstyle+\frac{(1+2^{2w_{2}+2\varsigma_{\xi}})c_{\psi5}\gamma_{2}^{2}\sigma_{\xi}^{2}}{(t+1)^{2w_{2}+2\varsigma_{\xi}}}\leq\frac{c_{s_1}}{(t+1)^{\beta_{1}+w_{2}}},\label{Rs5}
\end{flalign}
with $c_{s_1}=\Psi_{0,1}$ and $\beta_{1}+w_{2}=\min\{2\varsigma_{\zeta},2\varsigma_{\xi}+2w_{2},2u-w_{2},2v-w_{2}\}$.

Recalling the conditions $u>w_{2}$, $v>w_{2}$, $\varsigma_{\zeta}>\frac{w_{2}}{2}$, and $\varsigma_{\xi}>\frac{v}{2}-w_{2}$ given in the lemma statement, we have that $\beta_{1}+w_{2}>w_{2}$ always holds. 

Combining Lemma~\ref{AuxL1}-(i) with~\eqref{Rs4}, we arrive~\eqref{rateyzongresult}.
\end{proof}
\begin{lemma}\label{speedy}
Under the conditions in Lemma~\ref{rateyzong}, the following inequality holds for Algorithm~\ref{algorithm1}:
\begin{equation}
\textstyle \mathbb{E}[\|\hat{y}_{t+1}-\hat{y}_{t}\|^2]\leq \frac{\kappa_{2}}{(t+1)^{\beta_{2}}},\label{speedyresult}
\end{equation}
where  $\beta_{2}=\min\{2\varsigma_{\zeta},2\varsigma_{\xi}+2w_{1}+2w_{2},2w_{1}+2u,2w_{1}+2v\}$ and $\kappa_{2}=(\frac{4\beta_{2}}{e\ln(\frac{4}{4-|\delta_{2}|})})^{\beta_{2}}(\frac{\mathbb{E}[\|\hat{y}_{1}\|^2](2-|\delta_{2}|)}{2c_{s_2}}+\frac{4}{|\delta_{2}|})$ with $c_{s_2}$ given in~\eqref{Spe1}.
\end{lemma}
\begin{proof}
Combining~\eqref{rateyzongresult} with~\eqref{hatydiffresult}, we have
\begin{flalign}
&\textstyle\mathbb{E}[\|\hat{y}_{t+1}-\hat{y}_{t}\|^2]\textstyle\leq (1-|\delta_{2}|+c_{y1}\lambda_{t-1}^2)\mathbb{E}[\|\hat{y}_{t}-\hat{y}_{t-1}\|^2]\nonumber\\
&\textstyle\quad+\frac{\kappa_{1}(c_{y2}\gamma_{2}^2+c_{y3}\alpha_{0}^2)\gamma_{1}^2}{t^{2w_{1}+2w_{2}+\beta_{1}}}+\frac{(c_{y4}+c_{y5})\sigma_{\zeta}^2}{t^{2\varsigma_{\zeta}}}+\frac{c_{y6}w_{1}^2\gamma_{1}^2}{t^{2w_{1}+2}}\nonumber\\
&\textstyle\quad+\frac{c_{y7}\gamma_{1}^2\lambda_{0}^2}{t^{2w_{1}+2u}}+\frac{c_{y8}\gamma_{1}^2\gamma_{2}^2\sigma_{\xi}^2}{t^{2w_{1}+2w_{2}+2\varsigma_{\xi}}}+\frac{c_{y9}\gamma_{1}^2\alpha_{0}^2}{t^{2w_{1}+2v}}\nonumber\\
&\textstyle\leq \Big({\textstyle1-\frac{|\delta_{2}|}{2}}\Big)\mathbb{E}[\|\hat{y}_{t}-\hat{y}_{t-1}\|^2]+\frac{c_{s_2}}{(t+1)^{\beta_{2}}},\label{Spe1}
\end{flalign}
with $c_{s_2}=2^{\beta_{2}}(\kappa_{1}(c_{y2}\gamma_{2}^2+c_{y3}\alpha_{0}^2)\gamma_{1}^2+(c_{y4}+c_{y5})\sigma_{\zeta}^2+c_{y6}w_{2}^2\gamma_{1}^2+c_{y7}\gamma_{1}^2\lambda_{0}^2+c_{y8}\gamma_{1}^2\gamma_{2}^2\sigma_{\xi}^2+c_{y9}\gamma_{1}^2\alpha_{0}^2)$. 

Combining 
Lemma~11 in~\cite{zijiGT} with~\eqref{Spe1}, we arrive at~\eqref{speedyresult}.
\end{proof}

\subsection{Proof of Lemma~\ref{Lemhyperg}}\label{Asection4}
(i) We first prove inequality~\eqref{hypergresult} in Lemma~\ref{Lemhyperg}.

By using~\eqref{A2L1},~\eqref{Hg1}, and Assumption~\ref{A1}, we obtain
\begin{flalign}
&\mathbb{E}[\|\nabla\tilde{F}(x_{t})-\nabla F(x_{t})\|^2]\leq 2\bar{L}_{f}^2\mathbb{E}\left[\|\psi_{t}-\boldsymbol{1}_{m}\otimes\phi(x_{t})\|^2\right]\nonumber\\
&+{\textstyle\frac{2L_{g}^2}{\gamma_{t,1}^2}}\mathbb{E}\left[\|y_{t+1}\!-\!y_{t}\!-\!\gamma_{t,1}(\boldsymbol{1}_{m}\otimes \nabla_{2}\bar{f}(x_{t},\phi(x_{t})))\|^2\right].\label{Hg2}
\end{flalign}

The second term on the right hand side of~\eqref{Hg2} satisfies
\begin{equation}
\begin{aligned}
&{\textstyle\frac{2L_{g}^2}{\gamma_{t,1}^2}}\mathbb{E}\left[\|y_{t+1}-y_{t}-\gamma_{t,1}(\boldsymbol{1}_{m}\otimes \nabla_{2}\bar{f}(x_{t},\phi(x_{t})))\|^2\right]\\
&\leq {\textstyle\frac{4L_{g}^2}{\gamma_{t,1}^2}}\mathbb{E}[\|\hat{y}_{t+1}-\hat{y}_{t}\|^2]+{\textstyle\frac{4mL_{g}^2\sigma_{t,\zeta}^2}{\gamma_{t,1}^2}}\\
&\quad+4L_{g}^2\bar{L}_{f}^2\mathbb{E}\left[\|\psi_{t}-\boldsymbol{1}_{m}\otimes\phi(x_{t})\|^2\right].\label{Hg7}
\end{aligned}
\end{equation}

Substituting~\eqref{Hg7} into~\eqref{Hg2}, we obtain
\begin{equation}
\begin{aligned}
&\mathbb{E}[\|\nabla\tilde{F}(x_{t})-\nabla F(x_{t})\|^2]\\
&\leq 2\bar{L}_{f}^2(1+2L_{g}^2)\mathbb{E}\left[\|\psi_{t}-\boldsymbol{1}_{m}\otimes\phi(x_{t})\|^2\right]\\
&\quad+{\textstyle\frac{4L_{g}^2}{\gamma_{t,1}^2}}\mathbb{E}[\|\hat{y}_{t+1}-\hat{y}_{t}\|^2]+{\textstyle\frac{4mL_{g}^2\sigma_{t,\zeta}^2}{\gamma_{t,1}^2}}.\label{Hg8}
\end{aligned}
\end{equation}

By substituting~\eqref{rateyzongresult} into~\eqref{psigresult}, the first term on the right hand side of~\eqref{Hg8} satisfies
\begin{equation}
\textstyle\mathbb{E}\left[\|\psi_{t}\!-\!\boldsymbol{1}_{m}\otimes\phi(x_{t})\|^2\right]\!\leq\!\frac{2\kappa_{1}}{(t+1)^{\beta_{1}}}+\frac{c_{\bar{\psi}1}\gamma_{2}^{2}\sigma_{\xi}^{2}}{\alpha_{0}(t+1)^{2w_{2}+2\varsigma_{\xi}-v}}.\label{Hg6}
\end{equation}

Substituting~\eqref{speedyresult} and~\eqref{Hg6} into~\eqref{Hg8}, we arrive at
\begin{align}
&\mathbb{E}[\|\nabla\tilde{F}(x_{t})-\nabla F(x_{t})\|^2]\nonumber\\
&\leq\textstyle \frac{4\kappa_{1}\bar{L}_{f}^2(1+2L_{g}^2)}{(t+1)^{\beta_{1}}}+\frac{2\bar{L}_{f}^2(1+2L_{g}^2)c_{\bar{\psi}1}\gamma_{2}^{2}\sigma_{\xi}^{2}}{\alpha_{0}(t+1)^{2w_{2}+2\varsigma_{\xi}-v}}\nonumber\\
&\textstyle\quad+\frac{4L_{g}^2\kappa_{2}}{\gamma_{1}^2(t+1)^{\beta_{2}-2w_{1}}}+\frac{4mL_{g}^2\sigma_{\zeta}^2}{\gamma_{1}^2(t+1)^{2\varsigma_{\zeta}-2w_{1}}}\leq \frac{C_{1}}{(t+1)^{\beta}},\label{hypergresult2}
\end{align}
where the rate $\beta$ is given in the lemma statement and the constant $C_{1}$ is given by $C_{1}=4\kappa_{1}\bar{L}_{f}^2(2L_{g}^2+1)+\frac{4L_{g}^2(\kappa_{2}+m\sigma_{\zeta}^2)}{\gamma_{1}^2}+\frac{4m\gamma_{2}^{2}\sigma_{\xi}^{2}\bar{L}_{f}^2(2L_{g}^2+1)}{\alpha_{0}-(2w_{2}+2\varsigma_{\xi}-v)}$ with $\kappa_{1}$ given in~\eqref{rateyzongresult} and $\kappa_{2}$ given in~\eqref{speedyresult}.

We use the big-$\mathcal{O}$ notation to express the explicit constant coefficient
$C_{1}$ in~\eqref{hypergresult2}, which leads to~\eqref{hypergresult} in Lemma~\ref{Lemhyperg}.
\vspace{0.2em}

(ii) We next prove inequality~\eqref{tildeFresult} in Lemma~\ref{Lemhyperg}.

By using~\eqref{Hg1}, we obtain
\begin{flalign}
&\mathbb{E}[\|\nabla \tilde{F}(x_{t})\|^2]\leq 4\mathbb{E}[\|\nabla_{1}f(x_{t},\psi_{t})-\nabla_{1}f(x_{t},\phi(x_{t}))\|^2]\nonumber\\
&+\!4\mathbb{E}[\|\nabla_{1}f(x_{t},\phi(x_{t}))\|^2]\!+\!2\mathbb{E}\!\left[{\textstyle\frac{\|\nabla g(x_{t})\|^2\|y_{t+1}-y_{t}\|^2}{\gamma_{t,1}^{2}}}\!\right].\label{Tf2}
\end{flalign} 
Incorporating~\eqref{Hg6} into the first and second terms on the right hand side of~\eqref{Tf2} and using Assumption~\ref{A1}-(i), we obtain
\begin{flalign}
&\mathbb{E}[\|\nabla_{1}f(x_{t},\psi_{t})\!-\!\nabla_{1}f(x_{t},\phi(x_{t}))\|^2]\!+\!\mathbb{E}[\|\nabla_{1}f(x_{t},\phi(x_{t}))\|^2]\nonumber\\
&\leq\textstyle\frac{2\kappa_{1}\bar{L}_{f}^2}{(t+1)^{\beta_{1}}}+\frac{2m\gamma_{2}^{2}\sigma_{\xi}^{2}\bar{L}_{f}^2}{(\alpha_{0}-(2w_{2}+2\varsigma_{\xi}-v))(t+1)^{2w_{2}+2\varsigma_{\xi}-v}}+mL_{f}^2.\label{Tf3}
\end{flalign}
Using~\eqref{X6} and~\eqref{speedyresult}, the last term of~\eqref{Tf2} satisfies
\begin{equation}
\begin{aligned}
&\textstyle\mathbb{E}\Big[{\textstyle\|\nabla g(x_{t})\|^2\frac{\|y_{t+1}-y_{t}\|^2}{\gamma_{t,1}^{2}}}\Big]\leq \frac{2L_{g}^2\kappa_{2}}{\gamma_{1}^2(t+1)^{\beta_{2}-2w_{1}}}\\
&\textstyle\quad+\frac{2L_{g}^2m\sigma_{\zeta}^2}{\gamma_{1}^2(t+1)^{2\varsigma_{\zeta}-2w_{1}}}+2L_{g}^2L_{f}^2m.\label{Tf4}
\end{aligned}
\end{equation}

Substituting~\eqref{Tf3} and~\eqref{Tf4} into~\eqref{Tf2} leads to
\begin{equation}
\mathbb{E}[\|\nabla\tilde{F}(x_{t})\|^2]\leq C_{2},\label{tildeFresult2}
\end{equation}
where $C_{2}=8\bar{L}_{f}^2\kappa_{1}+\frac{4\bar{L}_{f}^2c_{\bar{\psi}1}\gamma_{2}^{2}\sigma_{\xi}^{2}}{\alpha_{0}}+4mL_{f}^2+\frac{4L_{g}^2(\kappa_{2}+m\sigma_{\zeta}^2)}{\gamma_{1}^2}+4L_{g}^2L_{f}^2m$ with $\kappa_{1}$ and $\kappa_{2}$ given in~\eqref{rateyzongresult} and~\eqref{speedyresult}, respectively.

We use the big-$\mathcal{O}$ notation to express the explicit constant coefficient $C_{2}$ in~\eqref{tildeFresult2}, which leads to~\eqref{tildeFresult} in Lemma~\ref{Lemhyperg}.

\subsection{Proof of Theorem~\ref{T1}}\label{Asection5}
(i) \textit{Convergence rate when $F(x)$ is strongly convex.}

To prove~\eqref{T1result1}, we first use the projection inequality and the dynamics of $x_{t}^{i}$ in Algorithm~\ref{algorithm1} to decompose the optimization error $\|x_{t+1}-x^*\|^2$ into the following two parts:
\begin{equation}
\begin{aligned}
&\|x_{t+1}-x^*\|^2\leq (1+\varepsilon)\|x_{t}-\lambda_{t}\nabla F(x_{t})-x^*\|^2\\
&\quad+\left({\textstyle1+\frac{1}{\varepsilon}}\right)\lambda_{t}^{2}\|\nabla \tilde{F}(x_{t})-\nabla F(x_{t})\|^2,\label{TS1}
\end{aligned}
\end{equation}
for any $\epsilon>0$, in which the first term on the right hand side represents the optimization error of a standard gradient descent using the global gradient $\nabla F(x_{t})$ and the second term on the right hand side is the estimation error of the gradient $\nabla F(x_{t})$.

The convergence result on the second term on the right hand side of~\eqref{TS1} is provided in~\eqref{hypergresult} in Lemma~\ref{Lemhyperg}. 
Therefore, we proceed to estimate an upper bound on the first term on the right hand side of~\eqref{TS1}, which satisfies
\begin{equation}
\begin{aligned}
&\|x_{t}-\lambda_{t}\nabla F(x_{t})-x^*\|^2=\|x_{t}-x^*\|^2\\
&\quad-2\lambda_{t}\langle\nabla F(x_{t}),x_{t}\!-\!x^*\rangle\!+\!\lambda_{t}^{2}\|\nabla F(x_{t})\!-\!\nabla F(x^{*})\|^2.\label{TS2}
\end{aligned}
\end{equation}
The strongly convexity of $F(x)$ implies $\langle\nabla F(x_{t}),x_{t}-x^*\rangle\geq F(x_{t})-F(x^{*})+\frac{\mu}{2}\|x_{t}-x^{*}\|^2,$ which further implies
\begin{equation}
-2\lambda_{t}\langle\nabla F(x_{t}),x_{t}-x^*\rangle\leq -\lambda_{t}\mu\|x_{t}-x^{*}\|^2.\label{TS3}
\end{equation}
The $L_{F}$-Lipschitz continuous of $\nabla F(x)$ implies
\begin{equation}
\lambda_{t}^{2}\|\nabla F(x_{t})-\nabla F(x^{*})\|^2\leq \lambda_{t}^2L_{F}^2\|x_{t}-x^{*}\|^2.\label{TS4}
\end{equation}

Substituting~\eqref{TS3} and~\eqref{TS4} into~\eqref{TS2}, we arrive at
\begin{equation}
\|x_{t}-\lambda_{t}\nabla F(x_{t})-x^*\|^2\leq (1-\lambda_{t}\mu+\lambda_{t}^{2}L_{F}^2)\|x_{t}-x^{*}\|^2.\label{TS5}
\end{equation}

By substituting~\eqref{hypergresult}  and~\eqref{TS5} into~\eqref{TS1} and letting $\epsilon=\frac{\lambda_{t}\mu}{2}$, we obtain the following recursive inequality:
\begin{equation}
\textstyle\mathbb{E}[\|x_{t+1}-x^*\|^2]\!\leq\! \left(1-\frac{\lambda_{t}\mu}{4}\right)\mathbb{E}[\|x_{t}-x^{*}\|^2]+\frac{(\lambda_{0}+\frac{2}{\mu})C_{1}\lambda_{0}}{(t+1)^{\beta+u}},\label{TS7}
\end{equation}
where in the derivation we have used the relation $\textstyle(1+\textstyle\frac{\lambda_{t}\mu}{2})(1-\lambda_{t}\mu+\lambda_{t}^2L_{F}^2)\leq 1-\frac{\lambda_{t}\mu}{4}-\frac{\lambda_{t}\mu}{4}+(1+\frac{\lambda_{t}\mu}{2})\lambda_{t}^2L_{F}^2\leq 1-\frac{\lambda_{t}\mu}{4}$
with $\lambda_{0}$ any value satisfying $(1+\frac{\lambda_{0}\mu}{2})\lambda_{0}^2L_{F}^2<\frac{\lambda_{0}\mu}{4}$. Here, the rate $\beta$ is given in~\eqref{hypergresult} and $C_{1}$ is given in~\eqref{hypergresult2}.

By combining Lemma~\ref{AuxL1}-(i) with~\eqref{TS7}, we arrive at
\begin{equation}
\mathbb{E}[\|x_{t}-x^*\|^2]\leq C(t+1)^{-\beta},\label{xres}
\end{equation}
with $C=\max\{\mathbb{E}[\|x_{0}-x^*\|^2],\frac{(\lambda_{0}+\frac{2}{\mu})\lambda_{0}C_{1}}{4^{-1}\lambda_{0}\mu-\beta}\}$. Inequality~\eqref{xres} directly implies~\eqref{T1result1}.

(ii) \textit{Convergence rate when $F(x)$ is general convex}.

To prove~\eqref{T1result2}, we expand the squared distance between 
$x_{t+1}$ and $x^*$ and use the dynamics of $x_{t}^{i}$ to obtain
\begin{equation}
\begin{aligned}
&\mathbb{E}[\|x_{t+1}-x^*\|^2]\leq \mathbb{E}[\|x_{t}-x^*\|^2]\\
&\quad+\lambda_{t}^2\mathbb{E}[\|\nabla \tilde{F}(x_{t})\|^2]-2\mathbb{E}[\langle x_{t}-x^*,\lambda_{t}\nabla \tilde{F}(x_{t})\rangle].\label{TC1}
\end{aligned}
\end{equation}

The convexity condition of $F(x)$ implies $\langle\nabla F(x_{t}),x_{t}-x^*\rangle\geq F(x_{t})-F(x^*),$ which further implies
\begin{equation}
\begin{aligned}
&-2\mathbb{E}[\langle x_{t}-x^*,\lambda_{t}\nabla \tilde{F}(x_{t})\rangle]\!\leq\! -2\lambda_{t}\mathbb{E}[F(x_{t})-F(x^{*})]\\
&\quad+\varepsilon_{t}\mathbb{E}[\|x_{t}-x^*\|^2]+\textstyle\frac{\lambda_{t}^2}{\varepsilon_{t}}\mathbb{E}[\|\nabla\tilde{F}(x_{t})-\nabla F(x_{t})\|^2],\label{TC2}
\end{aligned}
\end{equation}
where $\varepsilon_{t}$ is given by $\varepsilon_{t}=\frac{1}{(t+1)^{r}}$ with $1<r<2u$. Here, since the convexity condition applies to the global objective function $F(x)$ rather than its estimate $\tilde{F}(x)$, there exists an additional estimation error term in~\eqref{TC2}.

Substituting~\eqref{TC2} into~\eqref{TC1} leads to
\begin{equation}
\begin{aligned}
&\mathbb{E}[\|x_{t+1}-x^*\|^2]\leq -2\lambda_{t}\mathbb{E}[F(x_{t})-F(x^{*})]\\
&\quad+(1+\varepsilon_{t})\mathbb{E}[\|x_{t}-x^*\|^2]+\Theta_{t},\label{TC3}
\end{aligned}
\end{equation}
where the term $\Theta_{t}$ is given by 
\begin{equation}
\Theta_{t}=\lambda_{t}^2\mathbb{E}[\|\nabla \tilde{F}(x_{t})\|^2]+\textstyle\frac{\lambda_{t}^2}{\varepsilon_{t}}\mathbb{E}[\|\nabla\tilde{F}(x_{t})-\nabla F(x_{t})\|^2].\label{TC4}
\end{equation}

Since the relation $\tilde{F}(x_{t})\geq F(x^{*})$ always holds, we drop the negative term $-2\lambda_{t}\mathbb{E}[F(x_{t})-F(x^{*})]$ in~\eqref{TC3} to obtain
\begin{equation}
\begin{aligned}
&\mathbb{E}[\|x_{t+1}-x^*\|^2]\leq (1+\varepsilon_{t})\mathbb{E}[\|x_{t}-x^*\|^2]+\Theta_{t}\\
&\leq \left(\textstyle\prod_{t=0}^{T}(1+\varepsilon_{t})\right)\left(\mathbb{E}[\|x_{0}-x^*\|^2]+\textstyle\sum_{t=0}^{T}\Theta_{t}\right).\label{TC5}
\end{aligned}
\end{equation}

By using the relation $\ln(1+u)\leq u$ holding for any $u>0$ and the definition of $\varepsilon_{t}$ in~\eqref{TC2}, we have 
\begin{equation}
\begin{aligned}
\ln\left(\textstyle\prod_{t=0}^{T}(1+\varepsilon_{t})\right)&=\textstyle\sum_{t=0}^{T}\ln(1+\varepsilon_{t})\leq \textstyle\sum_{t=0}^{T}\varepsilon_{t}\\
&\leq\varepsilon_{0}+\textstyle\int_{1}^{\infty}\frac{1}{x^{r}}dx\leq \frac{\varepsilon_{0}(r-1)}{r-1},\label{TC6}
\end{aligned}
\end{equation}
which implies $\prod_{t=0}^{T}(1+\varepsilon_{t})\leq e^{\frac{\varepsilon_{0}(r-1)}{r-1}}$. 

Then, inequality~\eqref{TC5} can be rewritten as follows:
\begin{equation}
\begin{aligned}
&\mathbb{E}[\|x_{t+1}-x^*\|^2]\\
&\leq e^{\frac{\varepsilon_{0}(r-1)}{r-1}}\left(\mathbb{E}[\|x_{0}-x^*\|^2]+\textstyle\sum_{t=0}^{T}\Theta_{t}\right).\label{TC7}
\end{aligned}
\end{equation}

Substituting~\eqref{hypergresult} and~\eqref{tildeFresult} into~\eqref{TC4}, we obtain $\sum_{t=0}^{T}\Theta_{t}\leq \sum_{t=0}^{T}\frac{\lambda_{0}^2C_{2}}{(t+1)^{2u}}+\sum_{t=0}^{T}\frac{\lambda_{0}^2C_{1}}{(t+1)^{2u+\beta-r}}.$ By using the relation:
\begin{equation}
\textstyle\sum_{t=0}^{T}\frac{1}{(t+1)^{r}}\leq 1+\int_{1}^{\infty}\frac{1}{x^{r}}dx\leq \frac{r}{r-1},\label{TC8}
\end{equation}
which holds for any $r>1$, we have
\begin{equation}
\begin{aligned}
&\textstyle\sum_{t=0}^{T}\Theta_{t}\leq \frac{2u\lambda_{0}^2C_{2}}{2u-1}+\lambda_{0}^2C_{1}\max\left\{\frac{4u-2w_{2}-r}{4u-2w_{2}-r-1},\right.\\
&\textstyle\left.\quad\frac{2u+2v-2w_{2}-r}{2u+2v-2w_{2}-r-1},\frac{2u+2w_{2}+2\varsigma_{\xi}-v-r}{2u+2w_{2}+2\varsigma_{\xi}-v-r-1},\right.\\
&\textstyle\left.\quad\frac{2u+2\varsigma_{\zeta}-2w_{1}-r}{2u+2\varsigma_{\zeta}-2w_{1}-r-1},\frac{2u+2\varsigma_{\zeta}-w_{2}-r}{2u+2\varsigma_{\zeta}-w_{2}-r-1}\right\}\triangleq C_{3}. \label{TC9}
\end{aligned}
\end{equation}

Substituting~\eqref{TC9} into~\eqref{TC7}, we arrive at
\begin{equation}
\mathbb{E}[\|x_{t+1}-x^*\|^2]\leq e^{\frac{\varepsilon_{0}(r-1)}{r-1}}(\mathbb{E}[\|x_{0}-x^*\|^2]+C_{3}).\label{TC10}
\end{equation}

We proceed to sum both sides of~\eqref{TC3} from $0$ to $T$ ($T$ can be any positive integer):
\begin{flalign}
&\textstyle\sum_{t=0}^{T}2\lambda_{t}\mathbb{E}[F(x_{t})-F(x^{*})]\leq -\sum_{t=0}^{T}\mathbb{E}[\|x_{t+1}-x^*\|^2]\nonumber\\
&\textstyle\quad+\sum_{t=0}^{T}(1+\varepsilon_{t})\mathbb{E}[\|x_{t}-x^*\|^2]+\sum_{t=0}^{T}\Theta_{t}.\label{TC11}
\end{flalign}

The first and second terms on the right hand side of~\eqref{TC11} can be simplified as follows:
\begin{flalign}
&\textstyle\sum_{t=0}^{T}(1+\varepsilon_{t})\mathbb{E}[\|x_{t}-x^*\|^2]-\sum_{t=0}^{T}\mathbb{E}[\|x_{t+1}-x^*\|^2]\nonumber\\
&\textstyle= \varepsilon_{0}\mathbb{E}[\|x_{0}-x^*\|^2]+\sum_{t=1}^{T}\varepsilon_{t}\mathbb{E}[\|x_{t}-x^*\|^2]+\mathbb{E}[\|x_{0}-x^*\|^2]\nonumber\\
&\quad-\mathbb{E}[\|x_{T+1}-x^*\|^2]\nonumber\\
&\leq\! \textstyle\left(\frac{re^{\frac{\varepsilon_{0}(r-1)}{r-1}}}{r-1}\!+1+\varepsilon_{0}\right)\mathbb{E}[\|x_{0}-x^*\|^2]\!+\!\frac{C_{3}r}{r-1}\triangleq C_{4},\label{TC12}
\end{flalign}
where in the derivation we have used~\eqref{TC7} and~\eqref{TC8}.

Substituting~\eqref{TC9} and~\eqref{TC12} into~\eqref{TC11}, we obtain
\begin{equation}
\textstyle\sum_{t=0}^{T}2\lambda_{t}\mathbb{E}[F(x_{t})-F(x^{*})]\leq C_{3}+C_{4}.\label{3TB11}
\end{equation}
Dividing both sides of~\eqref{3TB11} by $2\sum_{t=0}^{T}\lambda_{t}$, we arrive at
\begin{equation}
\textstyle \frac{{\textstyle\sum_{t=0}^{T}}\lambda_{t}\mathbb{E}[F(x_{t})-F(x^{*})]}{{\textstyle\sum_{t=0}^{T}}\lambda_{t}}\leq \frac{(C_{3}+C_{4})(1-u)}{2\lambda_{0}(1-\frac{1}{2^{1-u}})(T+1)^{1-u}},\label{convex2}
\end{equation}
which in the derivation we have used the relations $\textstyle\sum_{t=0}^{T}\lambda_{t}\geq \int_{0}^{T}\frac{\lambda_{0}}{(x+1)^{u}}dx\geq \frac{\lambda_{0}((T+1)^{1-u}-1)}{1-u}$ and $(T+1)^{1-u}-1\geq (1-\frac{1}{2^{1-u}})(T+1)^{1-u}$. Inequality~\eqref{convex2} directly implies~\eqref{T1result2}.

In addition, given that $\lambda_{t}$ is a decaying sequence, we have
$\lambda_{T}\sum_{t=0}^{T}\mathbb{E}[F(x_{t})-F(x^{*})]\leq \sum_{t=0}^{T}\lambda_{t}\mathbb{E}[F(x_{t})-F(x^{*})]$, which, combined with~\eqref{3TB11} yields
\begin{equation}
\textstyle\frac{1}{T+1}\sum_{t=0}^{T}\mathbb{E}[F(x_{t})-F(x^{*})]\leq \frac{C_{3}+C_{4}}{2\lambda_{0}(T+1)^{1-u}}.\label{1Tind}
\end{equation}

(iii) \textit{Convergence rate when $F(x)$ is nonconvex}.

To prove~\eqref{TT1result2N9}, we first use Lemma~\ref{FLips} to obtain
\begin{equation}
F(x_{t+1})-F(x_{t})\leq\langle\nabla F(x_{t}),x_{t+1}-x_{t}\rangle+\textstyle\frac{L_{F}}{2}\|x_{t+1}-x_{t}\|^2.\label{TN1}
\end{equation}
By using Line 5 in Algorithm~\ref{algorithm1}, we obtain
\begin{equation}
\begin{aligned}
&\mathbb{E}[\langle\nabla F(x_{t}),x_{t+1}-x_{t}\rangle]+\textstyle\frac{L_{F}}{2}\mathbb{E}[\|x_{t+1}-x_{t}\|^2]\\
&\leq -\mathbb{E}[\langle\nabla F(x_{t}),\lambda_{t}\nabla\tilde{F}(x_{t})\rangle]+\textstyle\frac{L_{F}\lambda_{t}^2}{2}\mathbb{E}[\|\nabla\tilde{F}(x_{t})\|^2].\label{TN2}
\end{aligned}
\end{equation}

The first term on the right hand side of~\eqref{TN2} satisfies
\begin{flalign}
&-\mathbb{E}[\langle\nabla F(x_{t}),\lambda_{t}\nabla\tilde{F}(x_{t})\rangle]\nonumber\\
&\leq -\textstyle\frac{\lambda_{t}}{2}\mathbb{E}[\|\nabla F(x_{t})\|^2]\!+\!\textstyle\frac{\lambda_{t}}{2}\mathbb{E}[\|\nabla \tilde{F}(x_{t})\!-\!\nabla F(x_{t})\|^2].\label{TN3}
\end{flalign}
By substituting~\eqref{TN2} and~\eqref{TN3} into~\eqref{TN1}, we arrive at
\begin{flalign}
&\mathbb{E}[F(x_{t+1})-F(x_{t})]\leq -\textstyle\frac{\lambda_{t}}{2}\mathbb{E}[\|\nabla F(x_{t})\|^2]\nonumber\\
&\textstyle+\frac{\lambda_{t}}{2}\mathbb{E}[\|\nabla \tilde{F}(x_{t})\!-\!\nabla F(x_{t})\|^2]\!+\!\frac{L_{F}\lambda_{t}^2}{2}\mathbb{E}[\|\nabla \tilde{F}(x_{t})\|^2].\label{TN4}
\end{flalign}

Summing both sides of~\eqref{TN4} from $0$ to $T$ and using the relationship $F(x^*)\leq F(x_{t+1})$, we obtain
\begin{equation}
{\textstyle\sum_{t=0}^{T}\frac{\lambda_{t}}{2}}\mathbb{E}[\|\nabla F(x_{t})\|^2]\!\leq\! \mathbb{E}[F(x_{0})\!-\!F(x^*)]\!+\!{\textstyle\sum_{t=0}^{T}}\Theta_{t},\label{TN5}
\end{equation}
where the term $\Theta_{t}$ is given by 
\begin{equation}
\Theta_{t}={\textstyle\frac{\lambda_{t}}{2}}\mathbb{E}[\|\nabla \tilde{F}(x_{t})-\nabla F(x_{t})\|^2]+{\textstyle\frac{L_{F}\lambda_{t}^2}{2}}\mathbb{E}[\|\nabla \tilde{F}(x_{t})\|^2].\nonumber
\end{equation}

By using~\eqref{hypergresult},~\eqref{tildeFresult}, and~\eqref{TC8}, we have
\begin{equation}
\begin{aligned}
&\textstyle\sum_{t=0}^{T}\Theta_{t}\leq \frac{u\lambda_{0}^2C_{2}L_{F}}{2u-1}+\frac{\lambda_{0}C_{1}}{2}\max\left\{\frac{3u-2w_{2}}{3u-2w_{2}-1},\right.\\
&\textstyle\left.\quad\frac{u+2v-2w_{2}}{u+2v-2w_{2}-1},\frac{u+2\varsigma_{\xi}+2w_{2}-v}{u+2\varsigma_{\xi}+2w_{2}-v-1},\right.\\
&\textstyle\left.\quad\frac{u+2\varsigma_{\zeta}-2w_{1}}{u+2\varsigma_{\zeta}-2w_{1}-1},\frac{u+2\varsigma_{\zeta}-w_{2}}{u+2\varsigma_{\zeta}-w_{2}-1}\right\}\triangleq C_{5}. \label{TN8}
\end{aligned}
\end{equation}
Substituting~\eqref{TN8} into~\eqref{TN5} and defining $C_{6}\triangleq\mathbb{E}[F(x_{0})-F(x^*)]$, we obtain
\begin{equation}
\textstyle\sum_{t=0}^{T}\frac{\lambda_{t}}{2}\mathbb{E}[\|\nabla F(x_{t})\|^2]\leq C_{5}+C_{6}
\end{equation}

Following an argument similar to the derivation of~\eqref{convex2}, we arrive at
\begin{equation}
\textstyle\frac{\sum_{t=0}^{T}\lambda_{t}\mathbb{E}[\|\nabla F(x_{t})\|^2]}{\sum_{t=0}^{T}\lambda_{t}}\leq\frac{2(C_{5}+C_{6})(1-u)}{\lambda_{0}(1-\frac{1}{2^{1-u}})(T+1)^{1-u}},\label{TN11}
\end{equation}
which directly implies~\eqref{TT1result2N9}.

In addition, following an argument similar to the derivation of~\eqref{1Tind}, we have that~\eqref{TT1result2N9} is equivalent to
\begin{equation}
\textstyle\frac{1}{T+1}\sum_{t=0}^{T}\mathbb{E}[\|\nabla F(x_{t})\|^2]\leq \frac{2(C_{5}+C_{6})}{\lambda_{0}(T+1)^{1-u}},\nonumber
\end{equation}
which finishes the proof.
\vspace{-0.8em}
\subsection{Proof of Lemma~\ref{JDP}}\label{Asection6}
To prove Lemma~\ref{JDP}, we first provide a definition for the sensitivity of  an execution $\mathcal{A}_{\boldsymbol{\vartheta}_{0}}(\mathcal{P})$ of Algorithm~\ref{algorithm1}:
\begin{definition}[Sensitivity~\cite{tailoring}]\label{Dsensitivity}
Given an aggregative optimization problem $\mathcal{P}$, the set of observation sequences $\mathcal{O}_{s}$, and an initial state $\boldsymbol{\vartheta}_{0}$, we denote the internal iterates of an execution $\mathcal{A}_{\boldsymbol{\vartheta}_{0}}(\mathcal{P})$ of  Algorithm~\ref{algorithm1} at each iteration $t$ as 
$\mathcal{R}_{t,\mathcal{P},\boldsymbol{\vartheta}_{0}}(\mathcal{O}_{s})$.  Then, for two adjacent problems $\mathcal{P}$ and $\mathcal{P}'$, the sensitivity of  $\mathcal{A}_{\boldsymbol{\vartheta}_{0}}(\mathcal{P})$ at each iteration $t$ is defined as
\begin{equation}
\begin{aligned}
\Delta_{t}\triangleq \sup_{\mathcal{O}_{s}\subset\mathbb{O}_{s}}\left\{\sup_{\boldsymbol{\theta}\in\mathcal{R}_{t,\mathcal{P},\boldsymbol{\vartheta}_{0}}(\mathcal{O}_{s}),~{\boldsymbol{\theta}'}\in\mathcal{R}_{t,\mathcal{P'},\boldsymbol{\vartheta}_{0}}(\mathcal{O}_{s})}\|\boldsymbol{\theta}_{t}-{\boldsymbol{\theta}'}_{\hspace{-0.1cm}t}\|_{1}\right\},\nonumber
\end{aligned}
\end{equation}
where $\mathbb{O}_{s}$ denotes the set of all possible observations.
\end{definition}

By observing Algorithm~\ref{algorithm1}, we have that the observation part of $\mathcal{A}_{\boldsymbol{\vartheta}_{0}}(\mathcal{P})$ is the sequence of shared messages  $\{\text{col}(y_{t}+\zeta_{t},\psi_{t}+\xi_{t})\}_{t=1}^{T}$ with $y_{t}+\zeta_{t}=\text{col}(y_{t}^{1}+\zeta_{t}^{1},\cdots,y_{t}^{m}+\zeta_{t}^{m})$ and $\psi_{t}+\xi_{t}=\text{col}(\psi_{t}^{1}+\xi_{t}^{1},\cdots,\psi_{t}^{m}+\xi_{t}^{m})$. Hence, according to Definition~\ref{Dsensitivity}, the execution $\mathcal{A}_{\boldsymbol{\vartheta}_{0}}(\mathcal{P})$ at each iteration $t$ involves two sensitivities: $\Delta_{t,y}$ and $\Delta_{t,\psi}$, which correspond to the two shared variables $y_{t}$ and $\psi_{t}$, respectively. 

With this understanding, we proceed to prove Lemma~\ref{JDP}:
\begin{proof}
We first analyze the sensitivities of $\mathcal{A}_{\boldsymbol{\vartheta}_{0}}(\mathcal{P})$. Given two adjacent problems $\mathcal{P}$ and $\mathcal{P}'$, the set of observation sequences $\mathcal{O}_{s}$, and an initial state $\boldsymbol{\vartheta}_{0}$, the sensitivities at each iteration $t$ depend on $\|y_{t}-y'_{t}\|_{1}$ and $\|\psi_{t}-\psi'_{t}\|_{1}$. Since in $\mathcal{P}$ and $\mathcal{P}'$, there is only one entry that is different, we represent this different entry as the $i$th one, i.e., $\text{col}(\mathcal{X}_{i},f_{i},g_{i})$ in $\mathcal{P}$ and $\text{col}(\mathcal{X}'_{i},f'_{i},g'_{i})$ in $\mathcal{P}'$, without loss of generality.

Because the initial states, constraint sets, functions, and observations of $\mathcal{P}$ and $\mathcal{P}'$ are identical for all $j\neq i$, we have $y_{t}^{j}={y'}_{\hspace{-0.1cm}t}^{j}$ and $\psi_{t}^{j}={\psi'}_{\hspace{-0.1cm}t}^{j}$ for all $j\neq i$ and $t\in\mathbb{N}^{+}$. Therefore, $\|y_{t}-y'_{t}\|_{1}$ and $\|\psi_{t}-\psi'_{t}\|_{1}$ are always equal to  $\|y_{t}^{i}-{y'}_{\hspace{-0.1cm}t}^{i}\|_{1}$ and $\|\psi_{t}^{i}-{\psi'}_{\hspace{-0.1cm}t}^{i}\|_{1}$, respectively.

According to Line 7 in Algorithm~\ref{algorithm1}, we have
\begin{equation}
\begin{aligned}
&\psi_{t+1}^{i}-{\psi'}_{\hspace{-0.1cm}t+1}^{i}=(1-\alpha_{t}-\gamma_{t,2}|w_{ii}|)(\psi_{t}^{i}-{\psi'}_{\hspace{-0.1cm}t}^{i})\\
&\quad+(g_{i}(x_{t+1}^{i})-g_{i}(x_{t}^{i}))-({g'}_{\hspace{-0.1cm}i}({x'}_{\hspace{-0.1cm}t+1}^{i})-{g'}_{\hspace{-0.1cm}i}({x'}_{\hspace{-0.1cm}t}^{i}))\\
&\quad+\alpha_{t}(g_{i}(x_{t}^{i})-{g'}_{\hspace{-0.1cm}i}({x'}_{\hspace{-0.1cm}t}^{i})),\label{Dp1}
\end{aligned}
\end{equation}
where we have used the fact that the observations $\psi_{t}^{j}+\xi_{t}^{j}$ and ${\psi'}_{\hspace{-0.1cm}t}^{j}+{\xi'}_{\hspace{-0.1cm}t}^{j}$ for any $j\neq i$ and $t\in\mathbb{N}^{+}$ are the same.

Therefore, taking the $1$-norm on both sides of~\eqref{Dp1} yields
\begin{equation}
\begin{aligned}
&\|\psi_{t+1}^{i}-{\psi'}_{\hspace{-0.1cm}t+1}^{i}\|_{1}\leq(1-\alpha_{t}-\gamma_{t,2}|w_{ii}|)\|\psi_{t}^{i}-{\psi'}_{\hspace{-0.1cm}t}^{i}\|_{1}\\
&\quad+\sqrt{d}L_{g}(\|x_{t+1}^{i}-x_{t}^{i}\|_{2}+\|{x'}_{\hspace{-0.1cm}t+1}^{i}-{x'}_{\hspace{-0.1cm}t}^{i}\|_{2})+2\alpha_{t}\sqrt{d}D_{g},\label{Dp2}
\end{aligned}
\end{equation}
where in the derivation we have used Assumption~\ref{A1}-(ii), the relationship $\|x\|_{1}\leq \sqrt{d}\|x\|_{2}$ for any $x\in \mathbb{R}^{d}$, and the fact that $x_{t}^{i}$ is always constrained in a compact set $\mathcal{X}_{i}$, which implies $\|g_{i}(x_{t}^{i})\|_{2}\leq D_{g}$ for some $D_{g}>0$.

We proceed to characterize the term $\|x_{t+1}^{i}-x_{t}^{i}\|_{2}$ in~\eqref{Dp2}. By using the projection inequality, we have
\begin{equation}
\begin{aligned}
&\|x_{t+1}^{i}-x_{t}^{i}\|_{2}=\|\boldsymbol{P}_{\mathcal{X}_{i}}(x_{t+1}^{i})-\textstyle\boldsymbol{P}_{\mathcal{X}_{i}}(x_{t}^{i})\|_{2}\\
&\leq \lambda_{t}\|\nabla_{1}f_{i}(x_{t}^{i},\psi_{t}^{i})+{\textstyle\frac{1}{\gamma_{t,1}}}\nabla g_{i}(x_{t}^{i})(y_{t+1}^{i}-y_{t}^{i})\|_{2}.\label{Dp3}
\end{aligned}
\end{equation}
Using Lemma~\ref{ybound} and the relation $\sum_{p=0}^{t-1}\gamma_{p,1}\leq 1+\int_{1}^{\infty}\frac{\gamma_{1}}{x^{w_{1}}}=\frac{\gamma_{1}w_{1}}{w_{1}-1}$ for any $w_{1}>1$, we have $\|y_{t}^{i}\|\leq (1+\frac{\gamma_{1}w_{1}}{w_{1}-1})L_{f,2}$, which further implies $\|x_{t+1}^{i}-x_{t}^{i}\|_{2}\leq c_{0}\frac{\lambda_{t}}{\gamma_{t,1}}$ with $c_{0}\triangleq\gamma_{1}L_{f,1}+2L_{g}(1+\frac{\gamma_{1}w_{1}}{w_{1}-1})L_{f,2}$. Here, we have used Assumptions~\ref{A1} in the derivation. Furthermore, by incorporating $\|x_{t+1}^{i}-x_{t}^{i}\|_{2}\leq c_{0}\frac{\lambda_{t}}{\gamma_{t,1}}$ into~\eqref{Dp2}, we have that the sensitivity $\Delta_{t+1,\psi}$ satisfies
\begin{equation}
\Delta_{t+1,\psi}\leq(1-\gamma_{t,2}\hat{w})\Delta_{t,\psi}+\textstyle\frac{2c_{0}\sqrt{d}L_{g}\lambda_{t}}{\gamma_{t,1}}+2\sqrt{d}D_{g}\alpha_{t},\label{Dp5}
\end{equation}

Recalling the relationship $v>u-w_{1}$ given in the lemma statement, we can choose $\alpha_{0}\!\leq\! \frac{\lambda_{0}}{\gamma_{1}}$ such that $\alpha_{t}\!\leq\!\frac{\lambda_{t}}{\gamma_{t,1}}$ always holds for any $t\!>\!0$. In this case, \eqref{Dp5} can be rewritten as
\begin{equation}
\textstyle\Delta_{t+1,\psi}\leq(1-\gamma_{t,2}\hat{w})\Delta_{t,\psi}+\frac{(2c_{0}\sqrt{d}L_{g}+2\sqrt{d}D_{g})\lambda_{t}}{\gamma_{t,1}}.\label{Dp51}
\end{equation}
Combing~Lemma~\ref{AuxL1}-(i) with~\eqref{Dp51}, we obtain
\begin{equation}
\Delta_{t,\psi}\leq\textstyle\frac{c_{1}\lambda_{t}}{\gamma_{t,1}\gamma_{t,2}},\label{Dp6}
\end{equation}
with  $c_{1}=\frac{\hat{w}\gamma_{2}}{\hat{w}\gamma_{2}-(u-w_{1}-w_{2})}$, where the constant $c_{1}$ is obtained by omitting term $\Delta_{0,\psi}$ due to $\Delta_{0,\psi}=0$.

According to the ``Noise parameter settings",  we have that each element of noise vectors $\xi_{t}^{i}$ follows Laplace distribution $\text{Lap}(\nu_{t,\psi}^{i})$ with $2(\nu_{t,\psi}^{i})^{2}=(\sigma_{t,\xi}^{i})^2$. By defining $\nu_{t,\psi}=\min_{i\in[m]}\{\nu_{t,\psi}^{i}\}$ and $\hat{\varsigma}_{\xi}=\max_{i\in[m]}\{\varsigma_{\xi}^{i}\}$ and using~\eqref{Dp6}, we arrive at
\begin{equation}
\textstyle\sum_{t=1}^{T}\frac{\Delta_{t,\psi}}{\nu_{t,\psi}}\leq\sum_{t=1}^{T}\frac{\sqrt{2}c_{1}\lambda_{0}(\min_{i\in[m]}\{\sigma_{\xi}^{i}\}\gamma_{1}\gamma_{2})^{-1}}{(t+1)^{u-w_{1}-w_{2}-\hat{\varsigma}_{\xi}}}.\label{Dp7}
\end{equation}

We further analyze the sensitivity $\Delta_{t,y}$. For any agent $i\in[m]$, since its observations $\tilde{y}_{t}^{j}$ and $\tilde{y}_{t}^{\prime j}$ for any $j\neq i$ and $t\in\mathbb{N}^{+}$ are the same, we have $\boldsymbol{P}_{\Omega_{t}}(\tilde{y}_{t}^{j})=\boldsymbol{P}_{\Omega_{t}}(\tilde{y}_{t}^{\prime j})$, which combined with Line 4 in Algorithm~\ref{algorithm1} leads to
\begin{flalign}
&\textstyle y_{t+1}^{i}-y_{t+1}^{\prime i}\!=\!(1+w_{ii})(y_{t}^{i}-y_{t}^{\prime i})\!+\!\sum_{j\in\mathcal{N}_{i}}w_{ij}\boldsymbol{P}_{\Omega_{t}}(\tilde{y}_{t}^{j})\nonumber\\
&\quad\textstyle-\!\sum_{j\in\mathcal{N}_{i}}w_{ij}\boldsymbol{P}_{\Omega_{t}}(\tilde{y}_{t}^{\prime j})\!+\!\gamma_{t,1}(\nabla_{2}f_{i}(x_{t}^{i},\psi_{t}^{i})\!-\!\nabla_{2}f_{i}^{\prime}(x_{t}^{\prime i},\psi_{t}^{\prime i}))\nonumber\\
&=(1+w_{ii})(y_{t}^{i}-y_{t}^{\prime i})\nonumber\\
&\quad+\gamma_{t,1}(\nabla_{2}f_{i}(x_{t}^{i},\psi_{t}^{i})-\nabla_{2}f_{i}^{\prime}(x_{t}^{\prime i},\psi_{t}^{\prime i})).\label{proj}
\end{flalign}
Therefore, the sensitivity $\Delta_{t+1,y}$ satisfies
\begin{equation}
\Delta_{t+1,y}\leq (1-\hat{w})\Delta_{t,y}+2\sqrt{d}L_{f}\gamma_{t,1},\label{Dp9}
\end{equation}

Combining Lemma 11 in~\cite{zijiGT} with~\eqref{Dp9}, we arrive at
\begin{equation}
\Delta_{t,y}\leq c_{2}\gamma_{t,1},\label{Dp10}
\end{equation}
where the constant $c_{2}$ is given by  $c_{2}=(\frac{4w_{1}}{e\ln(\frac{2}{2-\hat{w}})})^{w_{1}}\frac{2}{\hat{w}}$.

The ``Noise parameter settings"  implies that each element of noise vectors $\zeta_{t}^{i}$ follows Laplace distribution $\text{Lap}(\nu_{t,y}^{i})$ with $2(\nu_{t,y}^{i})^{2}=(\sigma_{t,\zeta}^{i})^2$. By defining $\nu_{t,y}=\min_{i\in[m]}\{\nu_{t,y}^{i}\}$ and $\hat{\varsigma}_{\zeta}=\max_{i\in[m]}\{\varsigma_{\zeta}^{i}\}$ and using~\eqref{Dp10}, we arrive at
\begin{equation}
\textstyle\sum_{t=1}^{T}\frac{\Delta_{t,y}}{\nu_{t,y}}\leq\sum_{t=1}^{T}\frac{\sqrt{2}c_{2}\gamma_{1}(\min_{i\in[m]}\{\sigma_{\zeta}^{i}\})^{-1}}{(t+1)^{w_{1}-\hat{\varsigma}_{\zeta}}}.\label{Dp11}
\end{equation}

By using Lemma 2 in~\cite{Huang2015}, we have  that the execution $\mathcal{A}_{\boldsymbol{\vartheta}_{0}}(\mathcal{P})$ of Algorithm~\ref{algorithm1} for internal iteration sequences $\{y_{t}^{i}\}_{t=1}^{T}$ and $\{\psi_{t}^{i}\}_{t=1}^{T}$ is $\epsilon$-differentially private for any $i\in[m]$ with a cumulative privacy budget bounded by $\epsilon=\sum_{t=1}^{T}(\frac{\Delta_{t,\psi}}{\nu_{t,\psi}}+\frac{\Delta_{t,y}}{\nu_{t,y}})$. According to~\eqref{Dp7} and~\eqref{Dp11}, we have that the cumulative privacy budget is finite even when $T\rightarrow\infty$ since $u-w_{1}-w_{2}-\hat{\varsigma}_{\xi}>1$ and $w_{1}-\hat{\varsigma}_{\zeta}>1$ always hold.

We proceed to prove that the execution $\mathcal{A}_{\boldsymbol{\vartheta}_{0}}(\mathcal{P})$ of Algorithm~\ref{algorithm1} which outputs the decision variable $x_{T}$ is $\epsilon$-jointly differentially private. According to Algorithm~\ref{algorithm1} Line 5, there must exist some function $h_{1}$ such that the update of $x_{t}^{j}$ can be expressed as
\begin{equation}
x_{t+1}^{j}=h_{1}(x_{t}^{j},\psi_{t}^{j},y_{t+1}^{j},y_{t}^{j}, \mathcal{X}_{j},f_{j}(x_{t}^{j},\psi_{t}^{j}),g_{j}(x_{t}^{j})).
\end{equation}
By using induction from $1$ to $T-1$, we obtain
\begin{equation}
x_{T}^{j}=h_{2}(x_{1}^{j},\{y_{t}^{j}\}_{t=1}^{T}(\mathcal{P}),\{\psi_{t}^{j}\}_{t=1}^{T-1}(\mathcal{P}),\mathcal{P}_{j}),\label{26}
\end{equation}
for some function $h_{2}$. Here, since the dynamics of $y_{t}^{j}$ and $\psi_{t}^{j}$ indirectly rely on all neighboring agents' dynamics of $y_{t}^{l}$ and $\psi_{t}^{l}$ with $l\in\mathcal{N}_{j}$, we consider the worst scenario where the dynamics of $y_{t}^{j}$ and $\psi_{t}^{j}$ are affected by all agents' dynamics of $y_{t}^{l}$ and $\psi_{t}^{l}$ with $l\in[m]$, which depend on the global problem $\mathcal{P}$. Therefore, we use notions $\{y_{t}^{j}\}_{t=1}^{T}(\mathcal{P})$ and $\{\psi_{t}^{j}\}_{t=1}^{T-1}(\mathcal{P})$ to emphasize the dependence of $\{y_{t}^{j}\}_{t=1}^{T}$ and $\{\psi_{t}^{j}\}_{t=1}^{T-1}$ on $\mathcal{P}$. 

Since the two adjacent problems $\mathcal{P}$ and $\mathcal{P}'$ only differ in the $i$th element, we have $\mathcal{P}'_{j}=\mathcal{P}_{j}$ for all $j\neq i$, which implies
\begin{equation}
\begin{aligned}
{x'}_{\hspace{-0.1cm}T}^{j}&=h_{2}(x_{1}^{j},\{y_{t}^{j}\}_{t=1}^{T}(\mathcal{P}'),\{\psi_{t}^{j}\}_{t=1}^{T-1}(\mathcal{P}'),{\mathcal{P}'}_{\hspace{-0.1cm}j})\\
&=h_{2}(x_{1}^{j},\{y_{t}^{j}\}_{t=1}^{T}(\mathcal{P}'),\{\psi_{t}^{j}\}_{t=1}^{T-1}(\mathcal{P}'),{\mathcal{P}}_{j}).\label{T2a}
\end{aligned}
\end{equation}

Equations \eqref{26} and~\eqref{T2a} imply that for any $i\in[m]$, the outputs of $\mathcal{A}_{\boldsymbol{\vartheta}_{0}}(\mathcal{P})^{-i}$, i.e., $\text{col}(x_{T}^{1},\cdots,x_{T}^{i-1},x_{T}^{i+1},\cdots,x_{T}^{m})$, can be viewed as a post-processing result of the execution $\mathcal{A}_{\boldsymbol{\vartheta}_{0}}(\mathcal{P})$ for internal sequences $\{y_{t}^{j}\}_{t=1}^{T}(\mathcal{P})$ and $\{\psi_{t}^{j}\}_{t=1}^{T}(\mathcal{P})$, which is $\epsilon$-differentially private based on~\eqref{Dp7} and~\eqref{Dp11}. Therefore, based on the post-processing theorem in~\cite{Dwork2014}, we have that $\mathcal{A}_{\boldsymbol{\vartheta}_{0}}(\mathcal{P})^{-i}$ is $\epsilon$-differentially private, which corresponds to $\epsilon$-JDP of an execution $\mathcal{A}_{\boldsymbol{\vartheta}_{0}}(\mathcal{P})$ of Algorithm~\ref{algorithm1}.
\end{proof}
\subsection{Auxiliary lemmas for proving Theorem~\ref{T3}}
\begin{lemma}\label{Bpsidiff}
Under the conditions in Theorem~\ref{T2}, if  the rate of Laplace-noise variances satisfies $1>\varsigma_{\xi}>\max\{-\frac{w_{2}}{2},\frac{1}{2}-w_{2}\}$, we have the following inequality for Algorithm~\ref{algorithm1}:
\begin{equation}
\mathbb{E}[\|\psi_{t}-\boldsymbol{1}_{m}\otimes\phi(x_{t})\|^2]\leq\hat{\kappa}_{1},\label{Bpsidiffresult}
\end{equation}
where $\hat{\kappa}_{1}$ is given by $\hat{\kappa}_{1}\!=\!2c_{\hat{s}_1}+2m\gamma_{2}^{2}\sigma_{\xi}^{2}$ with $c_{\hat{s}_1}=\frac{c_{\hat{s}_2}}{|\delta_{2}|\gamma_{2}-\hat{s}_{1}}$,
$c_{\hat{s}_2}\!=\!\frac{8L_{g}^2(\gamma_{2}|\delta_{2}|+1)c_{\hat{x}}\lambda_{0}^2}{\gamma_{1}^2\gamma_{2}|\delta_{2}|}+4m\gamma_{2}^{2}\sigma_{\xi}^2\!+\!\frac{8mD_{g}^2(\gamma_{2}|\delta_{2}|+1)\alpha_{0}^2}{\gamma_{2}|\delta_{2}|}$, $c_{\hat{x}}=2mL_{f}^2\gamma_{1}^2+\frac{8mL_{g}^2L_{f,2}^2((\gamma_{1}+1)w_{1}-1)^2}{(w_{1}-1)^2}$, and $\hat{s}_{1}=\min\{2u-2w_{1}-2w_{2},2v-2w_{2},w_{2}+2\varsigma_{\xi}\}$.
\end{lemma}
\begin{proof}
By using the relation $\phi(x_{t})=\bar{g}(x_{t})$, we have
\begin{equation}
\|\psi_{t}-\boldsymbol{1}_{m}\otimes\phi(x_{t})\|^2\leq 2\|\hat{\psi}_{t}\|^2+2m\|\bar{\psi}_{t}-\bar{g}(x_{t})\|^2.\label{2Pg1}
\end{equation}

Substituting~\eqref{Hp4},~\eqref{Hp5}, and~\eqref{Hp6} into~\eqref{Hp3}, we have
\begin{equation}
\begin{aligned}
&\mathbb{E}[\|\hat{\psi}_{t+1}\|^2]\leq(1-\gamma_{t,2}|\delta_{2}|)\mathbb{E}[\|\hat{\psi}_{t}\|^2]\\
&\quad+\textstyle\frac{8L_{g}^2(\gamma_{2}|\delta_{2}|+1)}{|\delta_{2}|}\frac{1}{\gamma_{t,2}}\mathbb{E}[\|x_{t+1}-x_{t}\|^2]\\
&\quad+\textstyle\frac{8mD_{g}^2(\gamma_{2}|\delta_{2}|+1)}{|\delta_{2}|}\frac{\alpha_{t}^2}{\gamma_{t,2}}+4m\gamma_{t,2}^{2}\sigma_{t,\xi}^2.\label{2TS1}
\end{aligned}
\end{equation}

We characterize the term $\mathbb{E}[\|x_{t+1}-x_{t}\|^2]$ in~\eqref{2TS1}. 

Lemma~\ref{ybound} implies $\|y_{t}^{i}\|\leq (1+\sum_{p=0}^{t-1}\gamma_{p,1})L_{f,2}$. By using the relation $\sum_{p=0}^{t-1}\gamma_{p,1}\leq \gamma_{1}+\int_{1}^{\infty}\frac{\gamma_{1}}{x^{w_{1}}}dx=\frac{\gamma_{1}w_{1}}{w_{1}-1}$, we have 
\begin{equation}
\textstyle\|y_{t}^{i}\|\leq \frac{(\gamma_{1}+1)w_{1}-1}{w_{1}-1}L_{f,2}.\label{xxxx}
\end{equation}
Combining~\eqref{xxxx} with~\eqref{X2}, we obtain
\begin{equation}
\textstyle\mathbb{E}[\|x_{t+1}-x_{t}\|^2]\leq \frac{c_{\hat{x}}\lambda_{t}^2}{\gamma_{t,1}^2},\label{Bxdiffresult}
\end{equation}
with $c_{\hat{x}}=2mL_{f}^2\gamma_{1}^2+\frac{8mL_{g}^2L_{f,2}^2((\gamma_{1}+1)w_{1}-1)^2}{(w_{1}-1)^2}$.

Substituting~\eqref{Bxdiffresult} into~\eqref{2TS1}, we arrive at
\begin{flalign}
&\textstyle\mathbb{E}[\|\hat{\psi}_{t+1}\|^2]\leq\left(1-\frac{\gamma_{2}|\delta_{2}|}{(t+1)^{w_{2}}}\right)\mathbb{E}[\|\hat{\psi}_{t}\|^2]+\frac{c_{\hat{s}_1}}{(t+1)^{\hat{s}_{1}+w_{2}}},\label{2TS2}
\end{flalign}
where $\hat{s}_{1}$ and $c_{\hat{s}_1}$ are given in the lemma statement.

Since the relations $\hat{s}_{1}>w_{2}$ and $0<w_{2}<1$ hold based on the lemma statement, we use Lemma~\ref{AuxL1}-(i) and~\eqref{2TS2} to obtain
\begin{equation}
\textstyle\mathbb{E}[\|\hat{\psi}_{t}\|^2]\leq c_{\hat{s}_1},\label{Bhatpsiresult}
\end{equation}
where $c_{\hat{s}_1}$ is given by $c_{\hat{s}_1}=\frac{c_{\hat{s}_2}}{|\delta_{2}|\gamma_{2}-\hat{s}_{1}}$.

Based on~\eqref{Pg3}, we have
\begin{equation}
\begin{aligned}
&\mathbb{E}[\|\bar{\psi}_{t+1}-\bar{g}(x_{t+1})\|^2]\\
&\leq(1-\alpha_{t})\mathbb{E}[\|\bar{\psi}_{t}-\bar{g}(x_{t})\|^2]+\gamma_{t,2}^{2}\sigma_{t,\xi}^{2}.\label{2Pg3}
\end{aligned}
\end{equation}
Combining Lemma~\ref{AuxL1}-(ii) with~\eqref{2Pg3}, we have
\begin{equation}
\mathbb{E}[\|\bar{\psi}_{t}-\bar{g}(x_{t})\|^2]\leq \gamma_{2}^{2}\sigma_{\xi}^{2}.\label{2Pg4}
\end{equation}

Substituting~\eqref{Bhatpsiresult} and~\eqref{2Pg4} into~\eqref{2Pg1}, we arrive at~\eqref{Bpsidiffresult}.
\end{proof}
\begin{lemma}\label{Bhatydiff}
Under the conditions in Lemma~\ref{Bpsidiff}, we have the following inequality for Algorithm~\ref{algorithm1}:
\begin{equation}
\textstyle\mathbb{E}[\|y_{t+1}-y_{t}-\gamma_{t,1}(\boldsymbol{1}_{m}\otimes\nabla_{2}\bar{f}(x_{t},\psi_{t}))\|^2]\leq\frac{\hat{\kappa}_{2}}{(t+1)^{2\varsigma_{\zeta}}},\label{Bhatdiffresult}
\end{equation}
with $\hat{\kappa}_{2}=(\frac{8\varsigma_{\zeta}}{e\ln(\frac{2}{2-|\delta_{2}|})})^{2\varsigma_{\zeta}}\frac{18}{|\delta_{2}|}+3m\sigma_{\zeta}^{2}$.
\end{lemma}
\begin{proof}
By using~\eqref{Hy1} and the definition of $\hat{y}_{t}$, we obtain
\begin{equation}
\begin{aligned}
\hat{y}_{t+1}&=(I_{md}+(W\otimes I_{d}))\hat{y}_{t}+\boldsymbol{1}_{m}\otimes\bar{s}_{t,w}+\tilde{\Xi}_{t}\\
&\gamma_{t,1}(\nabla_{2}f(x_{t},\psi_{t})\!-\!\boldsymbol{1}_{m}\otimes\nabla_{2}\bar{f}(x_{t},\psi_{t})),\nonumber
\end{aligned}
\end{equation}
where $\tilde{\Xi}_{t}$ is given by $\tilde{\Xi}_{t}=\text{col}(\tilde{\Xi}_{t}^{1},\cdots,\tilde{\Xi}_{t}^{m})$ with $\tilde{\Xi}_{t}^{i}={\textstyle\sum_{j\in{\mathcal{N}_{i}}}}w_{ij}(\boldsymbol{P}_{\Omega_{t}}(y_{t}^{j}+\zeta_{t}^{j})-y_{t}^{j})$.

By using the Young's inequality and Assumption~\ref{A1}, we have
\begin{flalign}
&\mathbb{E}[\|\hat{y}_{t+1}\|^2]\nonumber\\
&\leq 2m\mathbb{E}[\|\bar{s}_{t,w}\|^2]+2\mathbb{E}[\|\tilde{\Xi}_{t}\|^2]\nonumber\\
&\quad+(1+|\delta_{2}|)(1-|\delta_{2}|)^2\mathbb{E}[\|\hat{y}_{t}\|^2]+\left({\textstyle1+\frac{1}{|\delta_{2}|}}\right)4mL_{f}^2\!\gamma_{t,1}^2.\nonumber
\end{flalign}
Following an argument similar to the derivation of~\eqref{Hy5} and using the relation $(1-|\delta_{2}|)^2(1+|\delta_{2}|)<(1-|\delta_{2}|)$, we have
\begin{equation}
\mathbb{E}[\|\hat{y}_{t+1}\|^2]\leq (1-|\delta_{2}|)\mathbb{E}[\|\hat{y}_{t}\|^2]\textstyle+\frac{\left(1+\frac{1}{|\delta_{2}|}\right)4mL_{f}^2\gamma_{1}^2+4m\sigma_{\zeta}^2}{(t+1)^{\min\{2w_{1},2\varsigma_{\zeta}\}}}.\label{BHy1}
\end{equation}

Combining Lemma 11 in~\cite{zijiGT} with~\eqref{BHy1}, we arrive at
\begin{equation}
\textstyle\mathbb{E}[\|\hat{y}_{t}\|^2]\leq \frac{c_{\hat{s}_3}}{(t+1)^{2\varsigma_{\zeta}}},\label{Bhatyresult}
\end{equation}
where in the derivation we have used that the relation $w_{1}>\varsigma_{\zeta}$ leads to $\min\{2w_{1},2\varsigma_{\zeta}\}=2\varsigma_{\zeta}$. Moreover, the constants $c_{\hat{s}_3}$ is given by $c_{\hat{s}_3}=(\frac{8\varsigma_{\zeta}}{e\ln(\frac{2}{2-|\delta_{2}|})})^{2\varsigma_{\zeta}}\frac{2}{|\delta_{2}|}$ based on $\|\hat{y}_{0}\|=0$.

By using the relation $a-b-c=a-\bar{a}-(b-\bar{b})+\bar{a}-\bar{b}-c$ for any $a, b, c\in\mathbb{R}^{md}$, we have
\begin{equation}
\begin{aligned}
&\mathbb{E}[\|y_{t+1}-y_{t}-\gamma_{t,1}(\boldsymbol{1}_{m}\otimes\nabla_{2}\bar{f}(x_{t},\psi_{t}))\|^2]\\
&\leq 3\mathbb{E}[\|\hat{y}_{t+1}\|^2]+3\mathbb{E}[\|\hat{y}_{t}\|^2]\\
&\quad+3m\mathbb{E}[\|\bar{y}_{t+1}-\bar{y}_{t}-\gamma_{t,1}\nabla_{2}\bar{f}(x_{t},\psi_{t})\|].\label{BHy2}
\end{aligned}
\end{equation}

Substituting~\eqref{Bhatyresult} into~\eqref{BHy2} and using $\mathbb{E}[\|\bar{y}_{t+1}-\bar{y}_{t}-\gamma_{t,1}\nabla_{2}\bar{f}(x_{t},\psi_{t})\|^2]=\mathbb{E}[\|\bar{s}_{t,w}\|^2]$, we arrive at~\eqref{Bhatdiffresult}.
\end{proof}
\begin{lemma}\label{Bhyperg}
Under the conditions in Theorem~\ref{T2}, if  the rate of Laplace-noise variances satisfies $1>\varsigma_{\xi}>\max\{-\frac{w_{2}}{2},\frac{1}{2}-w_{2}\}$, we have the following inequalities for Algorithm~\ref{algorithm1}:
\begin{align}
&\textstyle\mathbb{E}[\|\nabla\tilde{F}(x_{t})-\nabla F(x_{t})\|^2]\leq \frac{\hat{C}_{1}}{(t+1)^{2\varsigma_{\zeta}-2w_{1}}},\label{Bhyperg1}\\
&\textstyle\mathbb{E}[\|\nabla \tilde{F}(x_{t})\|^2]\leq \frac{\hat{C}_{2}}{(t+1)^{2\varsigma_{\zeta}-2w_{1}}},\label{Bhyperg2}
\end{align}
where $\hat{C}_{1}=2\bar{L}_{f}^2(1+2L_{g}^2)\hat{\kappa}_{1}+4L_{g}^2\hat{\kappa}_{2}\gamma_{1}^{-2}$ and $\hat{C}_{2}=4\bar{L}_{f}^2\hat{\kappa}_{1}+4L_{g}^{2}\hat{\kappa}_{2}\gamma_{1}^{-2}+4mL_{f}^2(1+L_{g}^2)$ with the constants $\hat{\kappa}_{1}$ and $\hat{\kappa}_{2}$ given in~\eqref{Bpsidiffresult} and~\eqref{Bhatdiffresult}, respectively.
\end{lemma}
\begin{proof}
By substituting the first inequality in~\eqref{Hg7} into~\eqref{Hg2}, we derive
\begin{equation}
\begin{aligned}
&\mathbb{E}[\|\nabla\tilde{F}(x_{t})-\nabla F(x_{t})\|^2]\\
&\leq 2\bar{L}_{f}^2(1+2L_{g}^2)\mathbb{E}\left[\|\psi_{t}-\boldsymbol{1}_{m}\otimes\phi(x_{t})\|^2\right]\\
&\quad+ {\textstyle\frac{4L_{g}^2}{\gamma_{t,1}^2}}\mathbb{E}\left[\|y_{t+1}\!-\!y_{t}\!-\!\gamma_{t,1}(\boldsymbol{1}_{m}\otimes \nabla_{2}\bar{f}(x_{t},\psi_{t}))\|^2\right].\label{BH1}
\end{aligned}
\end{equation}

Substituting~\eqref{Bpsidiffresult} and~\eqref{Bhatdiffresult} into~\eqref{BH1} and using the relation $\varsigma_{\zeta}<w_{1}$, we arrive at~\eqref{Bhyperg1}.

Following an argument similar to the derivation of~\eqref{Tf2}, we have
\begin{flalign}
&\mathbb{E}[\|\nabla \tilde{F}(x_{t})\|^2]\leq 4\bar{L}_{f}^2\mathbb{E}\left[\|\psi_{t}-\boldsymbol{1}_{m}\otimes\phi(x_{t})\|^2\right]+4mL_{f}^2\nonumber\\
&\quad+{\textstyle\frac{4L_{g}^2}{\gamma_{t,1}^2}}\mathbb{E}\left[\|y_{t+1}-y_{t}-\gamma_{t,1}(\boldsymbol{1}_{m}\otimes\nabla_{2}\bar{f}(x_{t},\psi_{t}))\|^2\right]\nonumber\\
&\quad+4L_{g}^2\mathbb{E}\left[\|(\boldsymbol{1}_{m}\otimes\nabla_{2}\bar{f}(x_{t},\psi_{t}))\|^2\right].\label{BH2}
\end{flalign}
Substituting~\eqref{Bpsidiffresult} and~\eqref{Bhatdiffresult} into~\eqref{BH2}, we arrive at~\eqref{Bhyperg2}.
\end{proof}

\subsection{Proof of Theorem~\ref{T3}}
The truthfulness result follows directly from Theorem~\ref{T2}.

(i) \textit{Convergence analysis when $F(x)$ is strongly convex.}

Substituting~\eqref{TS5} and~\eqref{Bhyperg1} into~\eqref{TS1}, we obtain
\begin{equation}
\begin{aligned}
\mathbb{E}[\|x_{t+1}-x^*\|^2]&\leq \big(1-\textstyle\frac{\lambda_{t}\mu}{4}\big)\mathbb{E}[\|x_{t}-x^{*}\|^2]\\
&\quad\textstyle+\big(\frac{\lambda_{0}\mu+2}{\mu}\big)\frac{\lambda_{0}\hat{C}_{1}}{(t+1)^{u+2\varsigma_{\zeta}-2w_{1}}}.\label{3T3}
\end{aligned}
\end{equation}

Recalling the relationships $u+2\varsigma_{\zeta}-2w_{1}>1$ and $u>1$ given in the theorem statement and combining Lemma~\ref{AuxL1}-(ii) with~\eqref{3T3}, we can obtain $\mathbb{E}[\|x_{t}-x^*\|^2]\leq \hat{C}_{3}$,
where $\hat{C}_{3}=\mathbb{E}[\|x_{0}-x^*\|^2]+\mu^{-1}\lambda_{0}\hat{C}_{1}(\lambda_{0}\mu+2)$ with $\hat{C}_{1}$ given in~\eqref{Bhyperg1}. 

According to the definition of $\hat{C}_{1}$ in Lemma~\ref{Bhyperg}, we have 
\begin{equation}
\textstyle\mathbb{E}[\|x_{T}-x^*\|^2]\leq\mathcal{O}\left(\frac{L_{g}^{4}L_{f}^2\bar{L}_{f}^2}{\mu|\delta_{2}|}+\sigma_{\zeta}^2+\sigma_{\xi}^2\right).\nonumber
\end{equation}
Given that $\eta$ has the same order as $\epsilon$ from Theorem~\ref{T2} and $\epsilon$ is on the order of $\mathcal{O}(\frac{1}{\sigma_{\zeta}}+\frac{1}{\sigma_{\xi}})$ from Lemma~\ref{JDP}, we arrive at~\eqref{3Tresult1}.

(ii) \textit{Convergence analysis when $F(x)$ is general convex.}

Following an argument similar to the derivation of~\eqref{TC3}, we have
\begin{equation}
\begin{aligned}
&\mathbb{E}[\|x_{t+1}-x^*\|^2]\leq -2\lambda_{t}\mathbb{E}[F(x_{t})-F(x^{*})]\\
&\quad+(1+\varepsilon_{t})\mathbb{E}[\|x_{t}-x^*\|^2]+\Theta_{t},\label{3TB1}
\end{aligned}
\end{equation}
where the term $\Theta_{t}$ is given by 
\begin{equation}
\Theta_{t}=\lambda_{t}^2\mathbb{E}[\|\nabla \tilde{F}(x_{t})\|^2]+\textstyle\frac{\lambda_{t}^2}{\varepsilon_{t}}\mathbb{E}[\|\nabla\tilde{F}(x_{t})-\nabla F(x_{t})\|^2],\label{3TB2}
\end{equation}
where $\varepsilon_{t}$ is given by $\varepsilon_{t}=\frac{1}{(t+1)^{r}}$ with $1<r<2u$.

Substituting~\eqref{Bhyperg1} and~\eqref{Bhyperg2} into~\eqref{3TB2} and using~\eqref{TC8}, we obtain
\begin{flalign}
&\textstyle\sum_{t=0}^{T}\Theta_{t}\leq \sum_{t=0}^{T}\frac{\lambda_{0}^2\hat{C}_{2}}{(t+1)^{2u+2\varsigma_{\zeta}-2w_{1}}}+\sum_{t=0}^{T}\frac{\lambda_{0}^2\hat{C}_{1}}{(t+1)^{2u+2\varsigma_{\zeta}-2w_{1}-r}}\nonumber\\
&\textstyle\leq \frac{\lambda_{0}^2\hat{C}_{2}(2u+2\varsigma_{\zeta}-2w_{1})}{2u+2\varsigma_{\zeta}-2w_{1}-1}+\frac{\lambda_{0}^2\hat{C}_{1}(2u+2\varsigma_{\zeta}-2w_{1}-r)}{2u+2\varsigma_{\zeta}-2w_{1}-r-1}\triangleq \hat{C}_{4},\label{3TB5}
\end{flalign}
with $\hat{C}_{1}$ and $\hat{C}_{2}$ 
given in~\eqref{Bhyperg1} and~\eqref{Bhyperg2}, respectively.

Using~\eqref{3TB5} and following an argument similar to the derivation of~\eqref{TC10}, we have
\begin{equation}
\mathbb{E}[\|x_{t+1}-x^*\|^2]\leq e^{\frac{\varepsilon_{0}(r-1)}{r-1}}\left(\mathbb{E}[\|x_{0}-x^*\|^2]+\hat{C}_{4}\right).\label{3TB8}
\end{equation}

We proceed to sum up both sides of~\eqref{3TB1} from $0$ to $T$ ($T$ can be any positive integer):
\begin{flalign}
&\textstyle\sum_{t=0}^{T}2\lambda_{t}\mathbb{E}[F(x_{t})-F(x^{*})]\leq -\sum_{t=0}^{T}\mathbb{E}[\|x_{t+1}-x^*\|^2]\nonumber\\
&\textstyle\quad+\sum_{t=0}^{T}(1+\varepsilon_{t})\mathbb{E}[\|x_{t}-x^*\|^2]+\sum_{t=0}^{T}\Theta_{t}.\label{3TB9}
\end{flalign}

Following an argument similar to the derivation of~\eqref{TC12} and using~\eqref{3TB8}, we have that the first and second terms on the right hand side of~\eqref{3TB9} satisfies
\begin{flalign}
&\textstyle\sum_{t=0}^{T}(1+\varepsilon_{t})\mathbb{E}[\|x_{t}-x^*\|^2]-\sum_{t=0}^{T}\mathbb{E}[\|x_{t+1}-x^*\|^2]\nonumber\\
&\leq\! \textstyle\left(\frac{re^{\frac{\varepsilon_{0}(r-1)}{r-1}}}{r-1}\!+1+\varepsilon_{0}\right)\mathbb{E}[\|x_{0}-x^*\|^2]\!+\!\frac{\hat{C}_{4}r}{r-1}\triangleq \hat{C}_{5}.\label{3TB10}
\end{flalign}
with $\hat{C}_{4}$ given in~\eqref{3TB5}.

Substituting~\eqref{3TB5} and~\eqref{3TB10} into~\eqref{3TB9}, we have
\begin{equation}
\textstyle\sum_{t=0}^{T}2\lambda_{t}\mathbb{E}[F(x_{t})-F(x^{*})]\leq \hat{C}_{4}+\hat{C}_{5}.\label{3TB111}
\end{equation}
Dividing both sides of~\eqref{3TB111} by $2\sum_{t=0}^{T}\lambda_{t}$, we arrive at
\begin{equation}
\textstyle\frac{{\textstyle\sum_{t=0}^{T}}\lambda_{t}\mathbb{E}[F(x_{t})-F(x^{*})]}{\sum_{t=0}^{T}\lambda_{t}}\leq \frac{(\hat{C}_{4}+\hat{C}_{5})(u-1)}{2\lambda_{0}(1-(T+1)^{1-u})},\label{3TB15}
\end{equation}
where we have used the relations $\sum_{t=0}^{T}\lambda_{t}\geq \int_{0}^{T}\frac{\lambda_{0}}{(x+1)^{u}}dx\geq \frac{\lambda_{0}(1-(T+1)^{1-u})}{u-1}$ and $\lim_{T\rightarrow \infty}(T+1)^{1-u}=0$ for any $u>1$.

According to the definitions of $\hat{C}_{4}$ and $\hat{C}_{5}$, we have $\hat{C}_{4}+\hat{C}_{5}\approx\mathcal{O}(\frac{L_{g}^{4}L_{f}^2\bar{L}_{f}^2}{|\delta_{2}|}+\sigma_{\zeta}^2+\sigma_{\xi}^2)$. Given that $\eta$ and $\epsilon$  have the same order from Theorem~\ref{T2} and $\epsilon$ is on the order of $\mathcal{O}(\frac{1}{\sigma_{\zeta}}+\frac{1}{\sigma_{\xi}})$ from Lemma~\ref{JDP}, we arrive at~\eqref{3Tresult2}.

(iii) \textit{Convergence analysis when $F(x)$ is nonconvex.}

Following an argument similar to the derivation of~\eqref{TN5}, we have
\begin{equation}
\textstyle\sum_{t=0}^{T}\frac{\lambda_{t}}{2}\mathbb{E}[\|\nabla F(x_{t})\|^2]\leq \mathbb{E}[F(x_{0})-F(x^*)]+{\textstyle\sum_{t=0}^{T}}\Theta_{t},\label{3TB21}
\end{equation}
where the term $\Theta_{t}$ is given by
\begin{equation}
\Theta_{t}\textstyle=\frac{\lambda_{t}}{2}\mathbb{E}[\|\nabla \tilde{F}(x_{t})-\nabla F(x_{t})\|^2]+\textstyle\frac{L_{F}\lambda_{t}^2}{2}\mathbb{E}[\|\nabla \tilde{F}(x_{t})\|^2].\label{3TB23}
\end{equation}

Substituting~\eqref{Bhyperg1} and~\eqref{Bhyperg2} into~\eqref{3TB23} and using~\eqref{TC8}, we have
\begin{equation}
\textstyle\sum_{t=0}^{T}\Theta_{t}\leq \frac{\lambda_{0}\hat{C}_{1}(u+2\varsigma_{\zeta}-2w_{1})}{2(u+2\varsigma_{\zeta}-2w_{1}-1)}+\frac{L_{F}\lambda_{0}^2\hat{C}_{2}(2u+2\varsigma_{\zeta}-2w_{1})}{2(2u+2\varsigma_{\zeta}-2w_{1}-1)}\triangleq \hat{C}_{6}. \label{3TB24}
\end{equation}
Incorporating~\eqref{3TB24} into~\eqref{3TB21} and defining $\hat{C}_{7}\triangleq\mathbb{E}[F(x_{0})-F(x^*)]$, we obtain
\begin{equation}
\textstyle\sum_{t=0}^{T}\frac{\lambda_{t}}{2}\mathbb{E}[\|\nabla F(x_{t})\|^2]\leq \hat{C}_{6}+\hat{C}_{7}.\nonumber
\end{equation}

By using an argument similar to the derivation of~\eqref{3TB15}, we arrive at
\begin{equation}
\textstyle \frac{\sum_{t=0}^{T}\lambda_{t}\mathbb{E}[\|\nabla F(x_{t})\|^2]}{\sum_{t=0}^{T}\lambda_{t}}\leq \frac{2(\hat{C}_{6}+\hat{C}_{7})(u-1)}{\lambda_{0}(1-(T+1)^{1-u})}.\nonumber
\end{equation}

By using the definitions of $\hat{C}_{6}$ and $\hat{C}_{7}$, we have $\hat{C}_{6}+\hat{C}_{7}\leq\mathcal{O}\left(\frac{L_{g}^{4}L_{f}^3\bar{L}_{f}^3}{|\delta_{2}|}+\sigma_{\zeta}^2+\sigma_{\xi}^2\right)$. Given that $\eta$ has the same order as $\epsilon$ from Theorem~\ref{T2} and $\epsilon$ is on the order of $\mathcal{O}(\frac{1}{\sigma_{\zeta}}+\frac{1}{\sigma_{\xi}})$ from Lemma~\ref{JDP}, we arrive at~\eqref{3Tresult3}.

\bibliographystyle{ieeetr}  
\bibliography{nonconvexquantization}
\end{document}